\documentclass[12pt]{report}
\usepackage[backend=biber, style=alphabetic, date=year, url=false, doi=false, isbn=false, eprint=true, firstinits=true, maxcitenames=5, maxbibnames=5]{biblatex}
\AtEveryBibitem{%
  \ifentrytype{online}{%
  }{%
    \clearfield{eprint}%
    \clearfield{urldate}%
  }%
}
\usepackage{tikz}
\usepackage{tkz-euclide}
\usetikzlibrary{positioning,arrows,patterns}
\usetikzlibrary{decorations.markings}
\usetikzlibrary{decorations.pathmorphing}
\usetikzlibrary{calc}
\tikzset{
  photon/.style={decorate, decoration={snake}, draw=black},
  fermion/.style={draw=black, postaction={decorate},decoration={markings,mark=at position .55 with {\arrow{>}}}},
  ghost/.style={dashed, postaction={decorate},decoration={markings,mark=at position .55 with {\arrow{>}}}},
    antighost/.style={dashed, postaction={decorate},decoration={markings,mark=at position .55 with {\arrow{<}}}},
  vertex/.style={draw,shape=circle,fill=black,minimum size=5pt,inner sep=0pt},
particle/.style={thick,draw=black},
particle2/.style={thick,draw=blue},
avector/.style={thick,draw=black, postaction={decorate},
    decoration={markings,mark=at position 1 with {\arrow[black]{triangle 45}}}},
gluon/.style={ draw=black,decorate,
    decoration={coil,aspect=0.9,amplitude=0.2cm}}
 }

\usepackage{amsmath}
\usepackage{amssymb}
\usepackage{amsthm}
\usepackage{paralist}
\usepackage{fullpage}
\usepackage{hyperref}
\usepackage{tikz-cd}
\usepackage{todonotes}
\usepackage{array}
\usepackage{subcaption}
\usepackage{footnote}
\makesavenoteenv{tabular}
\makesavenoteenv{table}
\theoremstyle{plain} 
\newtheorem{thm}{Theorem}[section]
\newtheorem{lem}[thm]{Lemma}
\newtheorem{cor}[thm]{Corollary}
\newtheorem{prop}[thm]{Proposition}

\theoremstyle{definition}
\newtheorem{defn}[thm]{Definition}
\newtheorem{expl}[thm]{Example}

\newtheorem*{defn*}{Definition}
\theoremstyle{remark}
\newtheorem{rem}[thm]{Remark}
\newtheorem{ass}[thm]{Assumption}
\newtheorem{exc}{Exercise}
\newtheorem{prob}{Problem}
\newcommand{\ii}{\mathrm{i}}
\newcommand{\N}{\mathbb{N}}
\newcommand{\C}{\mathbb{C}}
\newcommand{\Z}{\mathbb{Z}}
\newcommand{\R}{\mathbb{R}}

\newcommand{\g}{\mathfrak{g}}
\newcommand{\m}{\mathfrak{m}}
\newcommand{\de}{\partial}

\newcommand{\Hom}{\operatorname{Hom}}
\newcommand{\End}{\operatorname{End}}

\newcommand{\Spin}{\mathrm{Spin}}

\newcommand{\calP}{\mathcal{P}}
\newcommand{\calV}{\mathcal{V}}
\newcommand{\calF}{\mathcal{F}}
\newcommand{\calA}{\mathcal{A}}
\newcommand{\calY}{\mathcal{Y}}
\newcommand{\calL}{\mathcal{L}}
\newcommand{\sfa}{\mathsf{a}}

\newcommand\restr[2]{{
  \left.\kern-\nulldelimiterspace 
  #1 
  \vphantom{\big|} 
  \right|_{#2} 
  }}

\title{Lecture Notes on Chern-Simons Perturbation Theory}
\author{Konstantin Wernli}
\date{\today}
\begin{document}
\maketitle
\tableofcontents
\chapter*{Introduction}
\addcontentsline{toc}{chapter}{Introduction}
The goal of this study of the perturbative quantization of Chern-Simons theory is a better mathematical understanding of the Feynman path integral \cite{Feynman1942} in the domain of quantum field theory. Let us briefly introduce this notion. \\
\section*{The Feynman path integral}
Let us take the simplifying viewpoint that physics is about the prediction of values of observables, i.e. numbers attached to outcome of an experiment. The \emph{principle of least action} provides the following recipe for this computation in classical physics. By a \emph{$d$-dimensional classical field theory}, we mean the following set of data: 
\begin{itemize}
\item A $d$-dimensional manifold $M$, called ``spacetime''
\item A ``space of fields'' $F_M$ 
\item An ``action functional'' $S_M \colon F_M \to \R$. 
\end{itemize}
It is a basic requirement from physics that the assignment $M \mapsto (F_M,S_M)$ be \emph{local}.\footnote{There are different mathematical manifestations of this physical concept. One possibility is to ask that is that $F_M$ be the space of sections of a bundle $E$ over $M$, and that $S_M$ is the integral over $M$ of a \emph{local Lagrangian density}: A density-valued function $L$ on the $k$-th jet bundle $J^k(E)$of $E$ (the minimal such $k$ is called order of the theory, remarkably, most physical theories are either of first or second order).}
The principle of least action then states that the physical field configuration satisfies 
\begin{equation}
\delta S[\phi_0] = 0,
\end{equation}
the \emph{Euler-Lagrange equation}\footnote{Of course, this equation is only necessary for $\phi_0$ to be an actual minimizer of the action, but nowhere near sufficient: usual questions about critical points apply, with the added complication of $F_M$ usually being infinite-dimensional.} (here $\delta$ denotes the variational derivative). Supposing for a moment a unique solution $\phi_0$ to this equation, the value of any observable $O \colon F_M \to \R$ can be computed as $O(\phi_0)$ (that is, if we were to measure the observable $O$, the outcome would be $O(\phi_0)$.) Let us briefly look at two easy examples. \\
A common example is \emph{free scalar field theory}. Here $M$ is a Riemannian (or Lorentzian) manifold, for example $M=\R^4$ with standard Euclidean (or Minkowski) metric. The space of fields is $F_M = C^\infty(M, \R)$ and the action functional is given by 
\begin{equation}
S_M = \frac12\int_M (||d\phi||^2 + m^2\phi^2) dvol_g = \frac12\int_M (\phi \Delta \phi + m^2\phi^2)dvol_g,
\end{equation}
where $||\cdot||$ is the norm induced by the metric on the cotangent bundle, and $\Delta$ the Laplacian induced by the metric. Variation of this action functional leads to the Helmholtz equation 
$$\Delta \phi + m^2\phi = 0$$
and the analysis proceeds by looking for eigenfunctions of the Laplace operator. Possible observables are for instance $O_{x_0}(\phi) = \phi(x_0)$ for a point $x_0 \in M$ or $O_\eta(\phi) = \int \eta \phi dvol_g$ for a test function $\eta \in C_c^\infty(M)$. \\
Another example is classical mechanics. Here, the ``space-time'' is just a (time) interval $I$ and the space of fields is $C^\infty(I,N)$ where $N$ is a Riemannian manifold modeling space. The action functional is \begin{equation}
S_I = \int_I \frac12 m ||\dot{\gamma}(t)||^2 - V(\gamma(t))dt,
\end{equation}
where $V \colon N \to \R$ is called the potential. The Euler-Lagrange equation is Newton's Law 
$$ m\ddot{\gamma} = (\nabla V)(\gamma).$$
Possible observables include the position or velocity at some time $t_0$: $O_{t_0}(\gamma) = \gamma(t_0), O'_{t_0}(\gamma) = ||\gamma'(t_0)||$ and so on. \\
We now want to pass from classical to quantum physics. This is a major conceptual jump that we cannot possibly do justice here\footnote{The precise nature of the relation between classical and quantum physics has been the subject of research for over a century, and still not been fully understood.}. A main feature of quantum physics is that one can no longer predict values of observables with certainty, but only with certain probabilities. We are taking here (a very much simplified version of) Feynman's approach to quantum field theory \cite{Feynman1942,Feynman1949,Feynman1949a,Feynman1950}. It dictates that the expectation value of an observable can be computed as 
\begin{equation}
\langle O \rangle_M = \frac{1}{Z_M}\int_{F_M}O(\phi)e^{\frac{\ii}{\hbar}S_M[\phi]}D\phi,\label{eq:ExpValueO}
\end{equation}
where the \emph{partition function} $Z$ is given by 
\begin{equation}
Z_M = \int_{F_M}e^{\frac{\ii}{\hbar}S_M[\phi]}D\phi.\label{eq:DefZ}
\end{equation}
Integrals \eqref{eq:ExpValueO} and \eqref{eq:DefZ} are examples of what is known as \emph{Feynman path integrals} (or also \emph{functional integrals}). We will not delve further into their physical origins and interpretations (which are very elegant and interesting), but rather investigate the mathematical nature of these integrals. For a rigorous definition of an integral, a measure is required. But, the definition of sensible\footnote{For instance, countably additive Borel} measures on spaces of field seems not possible  in general, with the remarkable exception of the two examples above, see \cite{Glimm1987} for a deeper discussion of this fact and further references. \\
Let us for the moment focus our discussion on the partition function $Z_M$ and the assignment \begin{equation}
Z \colon M \mapsto Z_M. \label{eq:Zassignment}
\end{equation}
If this cannot be defined via an actual integral, the question is how else it can be rigorously mathematically defined? The continued success of the use of functional integration techniques suggests to mimick some properties of integrals in assignment \eqref{eq:Zassignment}. Two ideas as to which properties to use are prominent in the mathematical community: 
\begin{enumerate}[i)]
\item Try to implement Fubini's theorem in $M \mapsto Z_M$, 
\item Try to implement the asymptotics\footnote{The physical constant $\hbar$ is very small and it makes sense to ask about the properties of the semiclassical limit $\hbar \to 0$.} of oscillatory integrals in $M \mapsto Z_M$.
\end{enumerate}
The first approach leads to the idea of functorial quantum field theory (or FQFT). The second approach leads to what we call perturbative quantization in this lecture\footnote{Perturbation theory in physics is usually performed in the coupling constant, and not in $\hbar$ (which is usually set to 1). So one could argue that we are speaking about the semiclassical approximation, rather than perturbative quantization. However, the two approaches differ only by an overall rescaling, and are mathematically equivalent.}. We briefly illuminate both approaches (the second one will be discussed in detail later in the text). 
\section*{Functorial Quantum Field Theory}
Let us briefly explain how supposing that the path integral satisfies Fubini's theorem can be interpreted as Functorial QFT. The basic idea is that the locality assumption of field theories allows to cut the spacetime into pieces. Namely, for manifolds with boundary $M$ one can define a space of boundary fields $F_{\de M}$ with a surjective map $\pi \colon F_M \to F_{\de M}$ which in the simplest case is just restriction of the fields to the boundary\footnote{Typically this is only the case in first order theories, in theories of higher order one needs to also consider normal derivatives.}. If $M$ is a manifold that can be represented as two manifolds $M_1$,$M_2$ joined along a common boundary $\Sigma$: $M = M_1 \cup_{\Sigma} M_2$ (for instance, $S^2 = D^2 \cup_{S^1} D^2$ can be glued from two disks along their common boundary circle $S^1$) then, locality dictates that \begin{equation}F_M = F_{M_1} \times_{F_{\Sigma}} F_{M_2} = \{(\phi_1,\phi_2)|\pi_1(\phi_1) = \pi_2(\phi_2)\}
\end{equation}
and 
\begin{equation}
S_M[(\phi_1,\phi_2)] = S_{M_1}[\phi_1] + S_{M_2}[\phi_2].
\end{equation} In other words, $F_M$ is a fiber bundles over $F_{\Sigma}$ with fiber $(F_M)_b= \pi^{-1}_1(b) \times \pi^{-1}_2(b)$. Now, applying a formal Fubini theorem\footnote{Together with its generalized cousin for integration along fibers. } implies that we can factorize the integral \eqref{eq:DefZ} as 
\begin{align}
Z_M &= \int_{F_M}e^{\frac{\ii}{\hbar}S_M[\phi]}D\phi = \int_{F_\Sigma}\left(\int_{(F_M)_b}e^{\frac{\ii}{\hbar}S_M[\phi]}D\phi\right)Db \notag \\
&= \int_{F_\Sigma}\left(\int_{\pi^{-1}_1(b) \subset(F_{M_1})}e^{\frac{i}{\hbar}S_{M_1}[\phi_1]}D\phi_1 \right)\left(\int_{\pi_2^{-1}(b)F_{M_2}}e^{\frac{\ii}{\hbar}S_{M_2}[\phi_2]}D\phi_2\right)Db \label{eq:FQFT1}
\end{align}
(here we use integration along fibers in the first line and Fubini theorem in the second). Introduce the vector space\footnote{Here $\mathrm{Fun}$ denotes some space of functions whose precise nature is irrelevant to the present heuristic discussion.} $\mathcal{H}_\Sigma = \mathrm{Fun}(F_{\Sigma})$ with the (formal) pairing
$\langle f_1,f_2\rangle_{\mathcal{H}_\Sigma} = \int_{F_{\Sigma}}f_1(b)f_2(b)Db$, then we can rewrite Equation \eqref{eq:FQFT1} as 
\begin{equation}
Z_M = \langle Z_{M_1},Z_{M_2}\rangle_{\mathcal{H}_\Sigma}.\label{FQFT2}
\end{equation} From this formal considerations, one extracts the following definition (usually attributed to Atiyah \cite{Atiyah1988} and Segal \cite{Segal1988}). 
\begin{defn*}
A \emph{functorial QFT} $Z$ associates 
\begin{itemize}
\item
To every $d-1$-dimensional manifold $\Sigma$ a vector space $\mathcal{H}$ with an inner product $\langle\cdot,\cdot\rangle_\mathcal{H}$
\item To a $d$-dimensional manifold $M$ with boundary a vector $M \in \mathcal{H}_{\de M}$
\end{itemize}
such that 
\begin{itemize}
\item $H_\varnothing = \R$
\item If $M = M_1 \cup_\Sigma M_2$, then 
\begin{equation}
Z_M = \langle Z_{M_1}, Z_{M_2}\rangle_{\Sigma}.
\end{equation}
\end{itemize}
\end{defn*}
The adjective functorial stems from the fact that  such an assignment can be made made into functor\footnote{Mathematicians love functors.} from cobordism categories to vector spaces. Functors out of cobordism categories became of great interest to the mathematical community (and, partly, also physicists) thanks to this interpretation of the path integral, and over the last thirty years a considerable amount of research has gone into this area. 
\section*{Perturbative Quantization}
We will very brief here since this issue is discussed in great detail later.
Another observation about Feynman integrals is that since $S_M$ is real, the function $\exp(\ii/\hbar S_M)$ oscillates very wildly as $\hbar \to 0$, with exception of the critical points\footnote{Sometimes this is rephrased in Euclidean setting where one considers $\exp(-S/\hbar)$ as using method of steepest descent. The two approaches essentially produce equivalent answers to our question of understanding the path integral.}. For finite-dimensional integrals, this leads to the well-known principle of stationary phase. Suppose $F$ is a finite-dimensional manifold (with a reference density $\mu)$) and $S \colon F \to \R$ has isolated non-degenerate critical points.  Then as $\hbar \to 0$, the integral 
\begin{equation}
I[\hbar] = \int_M \exp(\ii/\hbar S) \mu 
\end{equation}
is concentrated in a neighbourhood of the critical points of $S$. Around a critical point one can expand $S$ in a Taylor series and then compute integrals of terms explicitly using Fresnel integrals\footnote{The equivalent of Gaussian integrals for complex exponents.}. In particular, there is an explicit expression in form of a formal power series for the asymptotic behaviour of $I[\hbar]$ 
\begin{equation}
I[\hbar]\simeq_{\hbar \to 0}
\frac{1}{2\pi \hbar^{\dim F/2}}\sum_{x_0\in Crit(S)}e^{\ii/\hbar S(x_0)}
 \frac{e^{\frac{\ii\pi}{4} \mathrm{sign} S''(x_0)}}{\det{}^{1/2}S''(x_0)}(1 + O(\hbar))
\end{equation}
 which uses only the Taylor expansion of $S$ at the critical point and the inverse of its quadratic part at the critical point (and its determinant). The terms in this formal power series of order\footnote{Ignoring the overall constant $\frac{1}{2\pi \hbar^{\dim F/2}}$} $O(\hbar)$ can be conveniently labeled by diagrams - later we will identify them as Feynman diagrams. This has the major advantage that it can be generalized to infinite dimensions, if we can make sense of the Taylor expansion, the detereminant of the inverse, and so on (but this has proven to be a lot simpler than finding appropriate measures). We will discuss all of these issues in Chapter \ref{ch:PertQuant}.
\section*{Why study Chern-Simons Theory?}
To make a long story short, the answer to that question is that Chern-Simons theory has been studied using a variety of approaches and viewpoints. In this sense, it is one of the best opportunities to understand the Feynman path integral, because so many answers are available that one can compare to and use in the task. However,  the question about the precise relationships between the different results still remains wide open. A better understanding of these relationships will deepen our understanding of the concept of quantization itself. \\
To be slightly more precise, after the seminal paper of Witten \cite{Witten1989}, interest in Chern-Simons theory in the mathematical physics community exploded, making it one of the most well-studied field theories at both classical and quantum level (in some sense the \emph{drosophila melanogaster} of quantum field theory). An exhaustive review of the literature is next to impossible, and we restrict ourselves to mentioning a few results. \\
It was Witten who argued that Chern-Simons theory was linked closely to knot and 3-manifold invariants in \cite{Witten1989}. To be more precise, he argued - using holomorphic quantization of the reduced phase space - that expectation values of Wilson loop observables were given by the Jones polynomial. Around the same time, Fr\"ohlich and King \cite{Froehlich1989} that also the perturbative quantization of Chern-Simons theory on $\R^3$ with Wilson lines leads to knot invariants, via the Khnizhnik-Zamolodchikov connection (these invariants later became known as the Kontsevich integral \cite{Kontsevich1993a}). From this moment on it was clear that Chern-Simons theory was intimately connected with the vast subjects of knot and 3-manifold invariants, and conformal field theory. \\
Shortly after Witten, Reshetikhin and Turaev \cite{Reshetikhin1991} defined a TQFT that led to answers similar to Witten's, and this TQFT is widely considered as the correct non-perturbative quantization of Chern-Simons theory, even though - to the best of the author's knowledge - there is no conclusive proof of a mathematical formulation of this statement. \\
The perturbative quantization of Chern-Simons was considered in various formulations and guises, starting in the more physics-oriented literature with \cite{Guadagnini1989a}. Shortly thereafter, Axelrod-Singer \cite{Axelrod1991a},\cite{Axelrod1994} and - in a different way - Kontsevich \cite{Kontsevich1994} showed that one can obtain 3-manifold invariants from the perturbative quantization of Chern-Simons theory on arbitrary 3-manifolds. The precise link between these invariants and the ``non-perturbative'' ones defined by Reshetikhin and Turaev is still unclear (and one of the main motivations for this lecture). The perturbative approach after Axelrod-Singer was developed further by Bott, Cattaneo, and Mn\"ev in the papers \cite{Bott1998},\cite{Bott1999},\cite{Cattaneo2000},\cite{Cattaneo2008}, the main sources for these lecture notes. \\
Let us also mention that Chern-Simons theory has been studied from the viewpoint of geometric quantization \cite{Axelrod1991} and conformal field theory (see e.g. \cite{Andersen2015}). It is the author's belief that a thorough understanding of the relationships between these different approaches to the quantum Chern-Simons theory will elucidate the evasive mathematical underpinnings behind the Feynman path integral. 
\chapter{Classical Chern-Simons Theory}\label{ch:ClassicalCS}
In this chapter we introduce the classical Chern-Simons action functional on closed manifolds, as well as its symmetries and critical points. There are some slightly subtle geometric effects in the definition of that action functional, related to the trivializability of bundles. For that reason we review the concepts of vector bundles, principal bundles and connections. More details can be found in any textbook on Gauge Theory, e.g. the book by Taubes \cite{Taubes2011} or the lecture notes by Baum \cite{Baum2014}. 
\section{Preliminaries}
\subsection{Vector bundles}
We start with the definition of a vector bundle. 
\begin{defn}[Vector bundle]
Let $M$ be a manifold and $k \in \{\R,\C\}$.
A \emph{rank $n$ $k$-vector bundle over $M$} is  
a pair $(E,\pi)$, where $E$ is a manifold and $\pi \colon E \to M$ is a surjective submersion, such that there is a cover $\mathfrak{U}=\{U_\alpha\}_{\alpha \in A}$ of $M$ satisfying
\begin{enumerate}[i)]
\item{The cover $\mathfrak{U}$ \emph{trivializes} $E$, that is, for every $\alpha \in A$ there exists a diffeomorphism $\Psi_\alpha\colon \pi^{-1}(U_\alpha) \to U_\alpha \times k^n$ such that 
\begin{equation} \begin{tikzcd}
\pi^{-1}(U_\alpha) \arrow[rd, "\pi"] \arrow[rr, "\psi_{\alpha}"] &          & U_\alpha \times k^n \arrow[ld, "\pi_1"'] \\
                                                              & U_\alpha &                                       
\end{tikzcd}
\label{eq:defVB1}
\end{equation}
commutes, 
}
\item For all $\alpha \in A$ and $u \in U_\alpha$, $\pi^{-1}(u)$ is a $k$-vector space and the map \begin{equation}
\restr{\psi_\alpha}{\pi^{-1}(u)}\colon \pi^{-1}(u) \to \{u\} \times k^n \label{eq:defVB2}
\end{equation}
is an isomorphism of vector spaces. 
\end{enumerate}
\end{defn}
Let us introduce some terminology. $M$ is called the \emph{base} (or \emph{base space}) of the vector bundle. $E$ is called the \emph{total space}, and $\pi$ the projection. For $u\in M$, $\pi^{-1}(u)$ is called the \emph{fiber (of $E$) over $u$}, and denoted $E_u$. $(U_\alpha,\psi_\alpha)$ is called a \emph{local trivialization} and $\mathfrak{U}$ is called a \emph{trivializing cover}. For $\alpha,\beta \in A$, let $U_{\alpha\beta} = U_\alpha \cap U_\beta.$ By diagram \eqref{eq:defVB1} and \eqref{eq:defVB2}, the maps\footnote{Often, in the literature one finds opposite convention for the indices. However, we find this intuitive because it is the transition map \emph{from $\alpha$ to $\beta$}.}
\begin{equation}
\tilde{g}_{\alpha\beta}= \psi_\beta\circ\psi^{-1}_\alpha \colon U_{\alpha\beta}\times k^n \to U_{\alpha\beta} \times k^n
\end{equation}
satisfy $\tilde{g}_{\alpha\beta}(u,v) = (u,g_{\alpha\beta}(u)v)$, where $g_{\alpha\beta}(u) \in GL_n(k)$. The corresponding maps 
\begin{equation}
g_{\alpha\beta}\colon U_{\alpha\beta} \to GL_n(k)
\end{equation}
are called the \emph{gluing maps}. By construction, they satisfy, for all $\alpha,\beta,\gamma\in A$  and 
\begin{subequations}\label{eqs:gluingmaps}
\begin{align}
g_{\alpha\alpha}(u) &= \mathrm{id}_{k^n} \forall, \quad u \in U_{\alpha} \\
g_{\alpha\beta}(u) &= g_{\beta\alpha}(u)^{-1} ,\quad \forall u \in U_{\alpha\beta} \\
g_{\beta\gamma}(u)g_{\alpha\beta}(u) &= g_{\alpha\gamma(u)},\quad\forall u \in U_{\alpha\beta\gamma} = U_{\alpha} \cap U_{\beta} \cap U_{\gamma}.
\end{align}
\end{subequations}
\begin{defn}[Vector bundle morphisms]
If $E,F$ are vector bundles over $M$ then a \emph{vector bundle morphism} is a smooth map $\Psi\colon E \to F$ such that the diagram 
\begin{equation}
\begin{tikzcd}
E \arrow[rd, "\pi"] \arrow[rr, "\Psi"] &   & F \arrow[ld, "\pi'"'] \\
                                       & M &                      
\end{tikzcd}
\end{equation}
commutes and $\restr{\Psi}{E_u}=:\Psi_u \colon E_u \to F_u$ is linear. A \emph{vector bundle isomorphism} is a vector bundle morphism which is also a diffemorphism. The set of vector bundle morphisms from $E$ to $F$ is denoted $\underline{\Hom}(E,F)$.
\end{defn}
\begin{rem}
 Given two vector bundles $E$ and $F$ over $M$, we can always find a cover of $M$ that trivializes both. Namely, given a trivializing cover $\{U_\alpha\}_{\alpha \in A}$ of $E$ and $\{V_\beta\}_{\beta \in B}$ of $F$, the cover $\{U_\alpha \cap V_\beta\}_{(\alpha,\beta) \in A \times B}$ is a trivializing cover of both $E$ and $F$. 
\end{rem}
\begin{exc} 
Suppose the rank of $E$ is $n$ and the rank of $F$ is $m$. Prove that a vector bundle morphism $\Psi \colon E \to F$ is given by a collection  of maps $\Psi_{\alpha}\colon U_\alpha \to \Hom(k^n,k^m)$ such that 
\begin{equation}
\Psi_\beta = g^F_{\alpha\beta}\Psi_\alpha g^E_{\beta\alpha}.
\end{equation}
\end{exc}
\begin{expl}
\begin{enumerate}[i)]
\item A vector space is a vector bundle over a point. 
\item For every manifold $M$, the tangent bundle $TM$ is a vector bundle over $M$. The transition maps of the tangent bundle $TM$ can be computed in the following way. Let $(U_\alpha,\varphi_\alpha)$ be an atlas of $M$. Then $d\varphi_{\alpha\beta}(u)\colon U_{\alpha\beta} \to GL_n(\R)$ are the transition maps of $TM$.
\item For any manifold $M$ and natural number $n$ there is the \emph{trivial rank $n$ vector bundle over $k$}, simply given by the direct product $M \times k^n$ with the canonical projection to $M$. This bundle is often denoted $\underline{k}^n$.

\item Recall that $\C\mathbb{P}^n$ is the space of lines in $\C^n$, i.e. $\C\mathbb{P}^n = \C\mathbb{P}^{n+1}/\C^\times$, where $\C^\times$ acts diagonally. The quotient map $\pi\colon\C^{n+1}\to \C\mathbb{P}^n$ is a rank 1 complex vector bundle over $\C\mathbb{P}^n$. Working out the details of this is a marvelous exercise.  
\end{enumerate}
\end{expl}
It is an important fact that the vector bundle is entirely determined up to isomorphism by its trivializing cover and the gluing maps. 
\begin{prop}
Given a cover $\mathfrak{U} = \{U_{\alpha}\}_\alpha$ of a manifold $M$ and a family of snooth maps $g_{\alpha\beta}\colon U_{\alpha\beta} \to GL_n(k)$ satisfying \eqref{eqs:gluingmaps}, there exists a unique (up to isomorphism) vector bundle $\pi\colon E \to M$ with trivializing cover $\mathfrak{U}$ and gluing maps $g_{\alpha\beta}$. 
\end{prop}
\begin{proof}
\emph{Existence:}
We can assume that each $U_{\alpha}$ is a contained in a domain of a chart of $M$ (otherwise, cover each $U_{\alpha}$ by charts $V_{\alpha\beta}$ and note that the transition restricted to each $V_{\alpha\beta}$ still satisfy \eqref{eqs:gluingmaps}.) First, construct the fiber over $u$ by \begin{equation} E_u:= \left(\coprod_{\alpha\in A,u\in U_{\alpha}}k^n\right)/\sim = \left(\{\alpha\in A,u \in U_{\alpha}\} \times k^n\right)/\sim
\end{equation}
where $(\alpha,v) \sim (\beta,w)$ if $g_{\alpha\beta}(u)v = w$. This is an equivalence relation \emph{since $g_{\alpha\beta}$ satisfy \eqref{eqs:gluingmaps}.} Then, let $E := \coprod_{u\in M} E_u$ and $\psi_{\alpha}[(u,\alpha,v)] = (u,v).$ Since $U_\alpha$ is contained in a chart, composition with this chart yields a chart of $E$. It is easily checked that this is indeed a smooth atlas.  \\
\emph{Uniqueness:} It is enough to show that two vector bundle with the same trivializing cover and gluing maps are isomorphic. Let $E$,$F$ be such vector bundles. Then, we construct the isomorphism over $U_\alpha$ by the diagram 
\[\begin{tikzcd}
F_\alpha\colon\pi^{-1}(U_\alpha) \arrow[rd, "\pi"] \arrow[r, "\psi^E_{\alpha}"] & U_\alpha \times k^n \arrow[r, "(\psi_\alpha^F)^{-1}"] & (\pi')^{-1}(U_\alpha) \arrow[ld, "\pi'"'] \\
                                                                  & U_\alpha                                              &                                          
\end{tikzcd}\]
Since the gluing maps are the same, the maps $F_\alpha$ and $F_\beta$ agree on $U_{\alpha\beta}$: We have 

\begin{align} F_\beta &= (\psi_\beta^F)^{-1} \circ \psi_\beta^E  \\
&= ( \tilde{g}_{\alpha\beta}^F)\circ \psi_\alpha^F)^{-1} \circ (\tilde{g}_{\alpha\beta}^E\circ \psi_\alpha^E) \\
&= (\psi_{\alpha}^F)^{-1}\circ (\tilde{g}^F_{\alpha\beta})^{-1}\circ \tilde{g}_{\alpha\beta}^E \circ \psi_\alpha^E = F_\alpha, 
\end{align}
since $(\tilde{g}^F_{\alpha\beta})^{-1}\circ \tilde{g}_{\alpha\beta}^E = \mathrm{id}_{U_{\alpha} \times k^n}$.
\end{proof}

This central fact will often help us define vector bundles via trivializing covers and gluing maps. Given a manifold $M$, we can define a rank $n$ vector bundle $E$ over $k$ by specifying a trivializing cover $\mathfrak{U}$ and transition maps $g_{\alpha\beta}\colon U_{\alpha\beta} \to GL_n(k)$, and we write $E = (\mathfrak{U},g_{\alpha\beta})$ for this vector bundle.
\begin{expl}
Consider the circle $S^1 = \R/\Z$ with the open cover $U_1=(0,1), U_2=(1/2,3/2)$. Then the M\"obius band is the vector bundle with transition function $g_{12}\colon U_{12} \to GL_1(\R) = \R^\times$ given by 
$$g_{12}(x) = \begin{cases} 1 & 
x \in (1/2,1) \\
-1 &x \in (1,3/2)
\end{cases}
$$
\end{expl}
\begin{rem}
One can show this is the only non-trivial vector bundle over $S^1$. In fact, the tangent bundle of $S^1$ is trivial $TS^1 \cong S^1 \times \R$. 
\end{rem}
\begin{defn}[Section]
A \emph{section} of a vector bundle $\pi \colon E \to M$ is a smooth map $\sigma \colon M \to E$ such that $\pi \circ \sigma = \mathrm{id}_M$. 
\end{defn}
The set of sections of $E$ is denoted $\Gamma(M,E)$ or simply $\Gamma(E)$ when no confusion is possible. Note that since $E_x$ is a vector space for all $x \in M$, we can naturally add sections and multiply them by scalars: 
\begin{equation}
(\sigma_1 + \sigma_2)(x) = \sigma_1(x) + \sigma_2(x), \qquad, (\lambda\sigma)(x) = \lambda\sigma(x).
\end{equation}
Thus, $\Gamma(E)$ is a $k$-vector space. 
\begin{expl}
\begin{enumerate}[a)]
\item A section of a trivial bundle $M \times k^n$ is given by $\sigma(x)= (x,f(x))$, where $f\colon M \to k^n$ is a smooth map. Thus, $\Gamma(M,\underline{k}^n) \cong C^\infty(M,k^n)$.
\item 
Over a trivializing cover $\mathfrak{U} = \{U_\alpha\}_{\alpha_\in A}$, a section is given by smooth functions $\sigma_\alpha\colon U_\alpha \to k^n$ satisfying 
\begin{equation}
\sigma_\beta(x) = g_{\alpha\beta}(x)\sigma_\alpha(x).
\end{equation}
\item  A section of the tangent bundle $TM$ is called a \emph{vector field}. 
\end{enumerate}
\end{expl}
The natural constructions on vector spaces, such as dualizing, direct sums and tensor products, carry over to vector bundles. Here the description in terms of transition functions comes in handy. 
\begin{defn} Let $E = (\mathfrak{U},g_{\alpha\beta})$ and $F=(\mathfrak{U},h_{\alpha\beta})$ be two vector bundles over the same trivializing cover. Then we define the following bundles:
\begin{enumerate}[i)]
\item The \emph{dual bundle} $E^*$ by
\begin{equation}
E^* = (\mathfrak{U},(g_{\alpha\beta}^*)^{-1})
\end{equation}
\item The \emph{direct sum} $E \oplus F$ by 
\begin{equation}
E \oplus F = (\mathfrak{U},g_{\alpha\beta}\oplus h_{\alpha\beta})
\end{equation}
\item The \emph{tensor product} $E \otimes F$ by 
\begin{equation}
E \otimes F = (\mathfrak{U},g_{\alpha\beta} \otimes h_{\alpha\beta})
\end{equation}
\item The \emph{symmetric} and \emph{exterior powers} $\mathrm{Sym}^kE$ and $\bigwedge^kE$by 
\begin{align}
\mathrm{Sym}^kE &= (\mathfrak{U},\mathrm{Sym}^kg_{\alpha\beta})\\
\bigwedge^kE &= \left(\mathfrak{U},\bigwedge^k g_{\alpha\beta}\right).
\end{align}
\item The \emph{determinant line} 
$\det E$ by 
\begin{equation}
\det E = \bigwedge^{rk(E)}E.
\end{equation}
\end{enumerate}
\end{defn}
\begin{expl}
\begin{enumerate}[a)]
\item The dual of the tangent bundle $TM$ is called the cotangent bundle and denoted $(TM)^* = T^*M$. 
\item Sections of $\bigwedge^k T^*M$ are called differential $k$-forms on $M$. 
\item If $E$ is a vector bundle, then sections of $\bigwedge^k T^*M \otimes E$ are called 
\emph{differential k-forms with values in $E$}.
\item Sections of $E^*\otimes F$ are the same as vector bundle morphisms $E \to F$: $$\Gamma(E^*\otimes F) \cong \underline{\Hom}(E,F)$$
(Exercise!)
\end{enumerate}
\end{expl}

\subsection{Principal bundles}
We begin with a definition. 
\begin{defn}[Principal bundle]
Let $G$ be a Lie group and $M$ be a manifold. A principal $G$-bundle is a triple $(P,\pi,M)$ such that
\begin{enumerate}[i)]
\item $\pi\colon P \to M$ is a smooth submersion, 
\item There is a free and transitive right action $P \times G \to P$ such that $\pi$ is $G$-invariant (that is, $\pi(pg) = \pi(p)$), 
\item There exists a \emph{trivializing cover} $\mathfrak{U} = \{U_\alpha\}_{\alpha \in A}$, that is, a cover of $M$ with the property that for every $\alpha \in A$ there exists a diffeomorphism $\Psi_{\alpha}\colon \pi^{-1}(U_\alpha) \to U_\alpha \times G$ such that 
\begin{equation}
\begin{tikzcd}
\pi^{-1}(U_\alpha) \arrow[rd, "\pi"] \arrow[rr, "\Psi_\alpha"] &          & U_\alpha \times G \arrow[ld, "\pi_1"'] \\
                                                               & U_\alpha &                                       
\end{tikzcd}
\end{equation}
commutes and $\Psi_\alpha(pg) = \Psi(p)g$. 
\end{enumerate}
\end{defn}
Again, we call $M$ the \emph{base} and $P$ the \emph{total space} of the bundle. The Lie group $G$ is called the structure group of the bundle. 
\begin{defn}[Morphism of principal bundles]
A \emph{morphism} of principal $G$-bundles $P$ and $P'$ is a smooth map $f \colon (P,\pi) \to (P',\pi')$ that commutes both with the right $G$-action and the projections, that is 
\begin{align}
f(pg) &= f(p)g \\
\pi(f(p)) &= \pi(p),
\end{align}
and an \emph{isomorphism} of principal $G$-bundles is a morphism which is also a diffeomorphism. 
\end{defn}

Again, we can define  
\begin{equation}
\tilde{g}_{\alpha\beta} =\Psi_\beta \circ \Psi_{\alpha}^{-1} \colon U_{\alpha\beta} \times G \to U_{\alpha\beta} \times G
\end{equation}
which are given by \begin{equation}
\tilde{g}_{\alpha\beta}(u,h) = (u,g_{\alpha\beta}(u)h)
\end{equation} 
since the trivializations commute with the projections. Note that the transition functions act from the left, since they commute with the right $G$-action.  The maps 
\begin{equation}
g_{\alpha\beta}\colon U_{\alpha\beta}\to G
\end{equation}
are called the transition or gluing maps.  They satisfy 
\begin{subequations}\label{eqs:gluingMapsPrincipal}
\begin{align}
g_{\alpha\alpha}(u) &= 1_G \\
g_{\alpha\beta}(u) &= g_{\beta\alpha}^{-1} \\
g_{\alpha\gamma}(u) &= g_{\beta\gamma}(u)g_{\alpha\beta}(u)
\end{align}
\end{subequations} We also have a similar proposition: 
\begin{prop}
Let $\mathfrak{U} = \{U_\alpha\}_{\alpha \in A}$ be a cover of $M$ and suppose $g_{\alpha\beta} \colon U_{\alpha\beta} \to G$ satisfy \eqref{eqs:gluingMapsPrincipal}. Then there exists a principal $G$-bundle $P$ over $M$ with trivializing cover $\mathfrak{U}$ and gluing maps $g_{\alpha\beta}$, and this bundle is unique up to isomorphim. 
\end{prop}

Even though vector bundles and principal bundles are different objects, in some sense they are like two sides of the same coin. This is explained by the following fundamental example: 
\begin{expl}
If $V$ is a vector spaces, then a \emph{frame} of $V$ is an ordered basis $\underline{e}= (e_1,\ldots,e_r)$. The set of frames is denoted by $\mathrm{Fr}(V)$. Let $\pi\colon E \to M$ be a rank $n$ $k$-vector bundle. Then, the \emph{frame bundle} $\mathrm{Fr}(E)$ of $E$ is the smooth manifold $\mathrm{Fr}(E) = \sqcup_{x \in M}\mathrm{Fr}(E_x)$. This manifold has a natural projection $\overline{\pi}\colon \mathrm{Fr(E)} \to M$. We can define a right $GL_n(k)$ action on $\mathrm{Fr}(E)$ in the following way. Let $\{U_\alpha\}_\alpha$ be a trivializing cover for $E$. Over a trivializing chart $\psi_\alpha \colon \pi^{-1}(U_\alpha) \to U_\alpha \times k^n$, the right action is given by \begin{equation}\underline{e}\cdot g = \psi_\alpha^{-1}(g^{-1}\psi_\alpha(e_1),\ldots,g^{-1}\psi(e_n)).\end{equation}
 This gives $\mathrm{Fr}(E)$ the structure of a principal $GL_n(k)$-bundle. If $\psi_\alpha = (\pi, A_\alpha$, then a trivialization of $\mathrm{Fr}(E)$ is given by 
 \begin{align*}
 \overline{\psi}_\alpha(\underline{e}) = (\overline{\pi}(\underline{e}), A_\alpha(\underline{e})).
 \end{align*}
 Here $A_\alpha(\underline{e}) = (A_\alpha e_1,\ldots A_\alpha e_n) \in GL_n(k)$
 Now, one can check that this principal $GL_n(k)$-bundle has \emph{the same} gluing maps $g_{\alpha\beta} = A_\beta A_\alpha^{-1}$: 
 $${A}_\beta(\underline{e}) = A_\beta (A_\alpha A_\alpha^{-1})(\underline{e}) = g_{\alpha\beta}A_\alpha(\underline{e}),$$
 hence $g_{\alpha\beta}$ are the transition functions of $P$ since they satisfy the defining equation
 
$$ \overline{\Psi}_\beta \circ \overline{\Psi}_\alpha(u,h) = (u, g_{\alpha\beta}h).$$
\end{expl}
This fact is important so we record it again: 
\par
\begin{center}
\fbox{\parbox{0.7\textwidth}{Let $E=(\mathfrak{U},g_{\alpha\beta})$ be a vector bundle. Then, its frame bundle is the principal  $GL_n(k)$-bundle $P=(\mathfrak{U},g_{\alpha\beta})$.}}
\end{center}
Thus, we can use the same data to define either  vector bundles or principal $GL_n(k)$ bundles. It is in this sense that we mean they are two sides of the same coin. However, we can construct vector bundles also from principal bundles with other structure groups.
\begin{defn}
Given a principal $G$-bundle $P = (\mathfrak{U},g_{\alpha\beta})$, and a representation $\rho \to GL_n(k)$, we define the \emph{associated vector bundle} $P \times_\rho k^n$ by 
\begin{equation}
E = (\mathfrak{U},\rho(g_{\alpha\beta}))
\end{equation}
\end{defn}
We say that a vector bundle $E$ has a $(G,\rho)$-structure if $E = P\times_\rho k^n$ for a principal $G$-bundle $P$. We denote this bundle by $\mathrm{Ad} P$. 
\begin{expl}
If $G \subset GL_n(k)$ is a subgroup, it has the trivial representation $\iota \colon G \hookrightarrow GL_n(k)$. We say that $E$ has a $G$-structure if it has a $(G,\iota)$ structure. This means that we can find $\mathfrak{U}$ and $g_{\alpha\beta}$ such that $E= (\mathfrak{U},g_{\alpha\beta})$, where the $g_{\alpha\beta}$ take values in $G \subset GL_n(k)$. For example, an orientation of $M$ is the same as an $SL_n(k)$-structure on $TM$.
\end{expl}
\begin{expl}[Adjoint bundle]
Let $G$ be a Lie group and $\mathfrak{g}$ be its Lie algebra. Then $G$ acts on $\mathfrak{g}$ via the adjoint action (if $G$ is a matrix group then this action is given by conjugation $g \cdot X = g X g^{-1}$). Hence, for every principal $G$-bundle $P$ we have the \emph{adjoint bundle} $P \times_{\rho}\mathfrak{g}$. 
\end{expl}
\subsection{Connections on vector bundles}
Very roughly, a \emph{connection} on a fiber bundle is a consistent way to move from one fiber in the bundle to the other. The concept of connection exists over both vector bundles and principal bundles.  We start with the concept of a connection on vector bundles. 

\begin{defn}[Connection on vector bundle]
Let $\pi\colon E \to M$. Then a \emph{connection on $E$} is a linear map
\[\nabla \colon \Gamma(E) \to \Gamma(T^*M \otimes E)\]
such that, for $f \in C^\infty(M)$ and $\sigma \in \Gamma(E)$,  the Leibniz rule
\begin{equation}
\nabla(f\sigma) = df \otimes \sigma + f \nabla \sigma
\end{equation}
holds. 
\end{defn}
The connection $\nabla$ induces a \emph{covariant derivative} along vector fields on sections on $E$.
\begin{defn}[Covariant derivative]
Let $\nabla$ be a connection on the vector bundle $E$ over $M$. Let $X$ be a vector field on $M$. Then, the map $\nabla_X \colon \Gamma(E) \to \Gamma(E)$ given by 
\begin{equation}
\nabla_X\sigma = \iota_X\nabla\sigma
\end{equation}
is called the \emph{covariant derivative of $\sigma$ along $X$ (with respect to $\nabla$)}.
\end{defn}
The Leibniz rule for the covariant derivative is 
\begin{equation}
\nabla_X(f\sigma) = (L_Xf)\sigma + f \nabla_X\sigma. 
\end{equation}
\begin{expl}
On the trivial bundle $M \times k^n$ we have a connection given by the de Rham differential $(f_1,\ldots f_n) \mapsto (df_1, \ldots df_n)$. This connection is called the \emph{trivial connection}.  
\end{expl}
\begin{prop}
If it is not empty, the space of connections $\mathcal{A}_E$ is an affine space modeled on the vector space $\Gamma(M,T^*M\otimes E^* \otimes E) = \Gamma(M, T^*M \otimes \End E) =  \Omega^1(\End E)$.
\end{prop}
\begin{proof}
First observe that if $\nabla^1$ and $\nabla^2$ are connections on $E$, then their difference $A = \nabla^1 - \nabla^0$ satisfies $A(f\sigma) = fA(\sigma)$. Hence $A$ defines a vector bundle morphism $E \to T^*M \otimes E$. It follows that $A \in \Gamma(\Hom(E,T^*M \otimes E)) \cong \Omega^1(\End(E))$.
\end{proof}
In particular, every connection on the trivial bundle $M \times k^n$ is of the form 
$\nabla = d + A$, where $A \in \Omega^1(\End(k^n))$. In a basis of $k^n$ we write\footnote{We follow the Einstein summation convention that repeated indices are summed over. This does not apply to indices labeling covers (usually $\alpha,\beta,\gamma$).} $A(e_j) = A^i_je_i$, then, 
$$\nabla(f_1,\ldots, f_n) = (df_1,\ldots df_n) + (A^i_1f_i,\ldots A^i_nf_i).$$ \ Thus one can think of $A$ as a 1-form with values in matrices, or, equivalenty, as a matrix of 1-forms $A^i_j$. Both viewpoints are sometimes helpful. \par
If $\Psi\colon E \to F$ is an isomorphism of vector bundles, and $\nabla$ is a connection on $E$, then the map $\check{\nabla} = (\mathrm{id} \otimes \Psi)\circ \nabla \circ \Psi^{-1}$ is a connection on $F$. It is the unique map that makes 
\begin{equation}
\begin{tikzcd}
\Gamma(E) \arrow[rr, "\nabla"]                                       &  & \Gamma(T^*M \otimes E) \arrow[d, "\mathrm{id}\otimes\Psi"] \\
\Gamma(F) \arrow[u, "\Psi^{-1}"'] \arrow[rr, "\check\nabla", dashed] &  & \Gamma(T^*M\otimes F)                                     
\end{tikzcd}
\end{equation}
commute. In particular, consider a local trivialization $\Psi_\alpha\colon \restr{E}{U_\alpha} \to U_\alpha\times k^n$ of $E$. Then a connection $\nabla$ on $E$ induces a connection on $U_\alpha \times k^n$, hence an element $A_\alpha\in \Omega^1(U_\alpha,\End k^n)$ and we have 
\begin{equation}\label{eq:LocalConnAction}
(\nabla\sigma)_\alpha = d\sigma_\alpha + A_\alpha\sigma_\alpha
\end{equation} (in the second term there is matrix-vector multiplication). This is called the \emph{connection 1-form} of $\nabla$ in $U_\alpha$. If $U_\beta$ is another local trivialization, one can ask how $A_\alpha$ and $A_\beta$ are related. 
\begin{prop}
Let $\nabla$ be a connection on the vector bundle $E$ and $A_\alpha,A_\beta$ be the connecction 1-forms on two local trivializations $U_\alpha$, $U_\beta$. Then 
\begin{equation}\label{eq:ConnTransformVector}
A_\beta = g_{\alpha\beta}A_{\alpha}g_{\alpha\beta}^{-1} - (dg_{\alpha\beta})g_{\alpha\beta}^{-1}.
\end{equation}
\end{prop}
\begin{proof}
Let $\sigma \in \Gamma(E)$. Then, we know that 
\begin{equation}\label{eq:ProofConnTransform1}
\sigma_\beta = g_{\alpha\beta}\sigma_\alpha
\end{equation}
and 
\begin{equation}\label{eq:ProofConnTransform2}
(\nabla\sigma)_\beta = g_{\alpha\beta}(\nabla\sigma)_\alpha.
\end{equation}
Expanding \eqref{eq:ProofConnTransform2} using \eqref{eq:LocalConnAction}, we obtain 
\begin{align*}
d\sigma_\beta + A_\beta\sigma_\beta &= g_{\alpha\beta}(d\sigma_\alpha + A_\alpha \sigma_\alpha) 
\end{align*}
On the other hand, using \eqref{eq:ProofConnTransform1} we obtain 
\[(\nabla\sigma)_\beta = d\sigma_\beta + A_\beta\sigma_\beta = d(g_{\alpha\beta}\sigma_\alpha) + A_\beta g_{\alpha\beta}\sigma_\alpha  = dg_{\alpha\beta}\sigma_{\alpha} + g_{\alpha\beta}d\sigma_\alpha + A_\beta g_{\alpha\beta}\sigma_{\alpha}.\]
We conclude that 
$$A_\beta g_{\alpha\beta} \sigma_\alpha = g_{\alpha\beta} A_{\alpha} \sigma_\alpha - dg_{\alpha\beta}s_\alpha.$$
Since this holds for all $s$, we see that 
$$ A_\beta g_{\alpha\beta} = g_{\alpha\beta}A_\alpha - dg_{\alpha\beta}$$
from where the claim follows.
\end{proof}
Hence, we can characterize a connection on a bundle $E = (\mathfrak{U},g_{\alpha\beta})$ as a collection of 1-forms $A_\alpha \in \Omega^1(U_\alpha,\End E)$. Notice that $\End E  = \End( k^n)$ is the Lie algebra of $GL_n(k)$.  This suggest a natural generalization of the concept of connections to principal bundles, as discussed in the next subsection. 
\subsection{Connections on principal bundles}
We start with some definitions. Suppose $G$ is a Lie group with Lie algebra $\mathfrak{g}$ that acts on a manifold $P$ from the right. For fixed $p \in P$, there is a map 
\begin{align*}
\mu_p \colon G &\to P \\
 g &\mapsto pg.
\end{align*}
The differential of this map at the identity element $e\in G$ is a map
\begin{align*}
(d\mu_p)_e\colon \mathfrak{g} \cong T_eG &\to T_pP \\
X \mapsto (d\mu_p)_eX
\end{align*}

\begin{defn}[Fundamental vector field]
 Let the Lie group $G$ act on the manifold $P$ from the right and let $X \in \mathfrak{g}$. Then, the fundamental vector field $X^\sharp$ on $P$ is the section of $TP$ defined by 
 \begin{equation}
 X^\sharp_p = (d\mu_p)_eX.
 \end{equation}
\end{defn}
In particular, $G$ acts on itself by right multiplication. For fixed $g \in G$, this action reads 
\begin{align*}
\mu_g \colon G &\to G \\
h \mapsto gh
\end{align*}
Hence we have $\mu_g = L_g$ (left multiplication by $g$). Let $X \in \mathfrak{X}$. The fundamental vector field of the right action of $G$ on itself is given by $X \mapsto (dL_g)_eX$. 
\begin{defn}[Maurer-Cartan Form]\label{def:MCform}
The Maurer-Cartan Form $\phi \in \Omega^1(G,\mathfrak{g})$ is defined by $\phi_g(X^\sharp_g) \equiv X \in \mathfrak{g}$, where $X^\sharp$ is the fundamental vector field of the right action of $G$ on itself.
\end{defn}
\begin{rem}
From the discussion above it follows that 
\begin{equation}
\phi_g = (dL_{g^{-1}})_g \colon T_gG \to T_eG \cong \mathfrak{g}.
\end{equation}
In particular, for matrix groups it is given by $\phi_g = g^{-1}dg$. 
\end{rem}
We can now define a connection on a principal bundle.
\begin{defn}[Connection on a principal bundle]\label{def:connprinc}
Let $\pi\colon P \to M$ be a principal $G$-bundle, and let $\mathfrak{g}$ be the Lie algebra of $G$. A \emph{connection on $P$} is a 
1-form $\Omega \in \Omega^1(P,\mathfrak{g})$ satisfying\footnote{Here, $\mathrm{Ad}_g \colon \mathfrak{g} \to \mathfrak{g}$ is given by differentiating the map $\mathrm{Ad}_g \colon  G \to G$ at the identity. For matrix groups $G \subset GL_n(k)$ we have $\mathfrak{g} \subset gl_n(k)$ and the adjoint action is $\mathrm{Ad}_{g}X = gXg^{-1}$. }
\begin{enumerate}[i)]
\item For all $g \in G$, 
\begin{equation}
R_g^*\Omega = \mathrm{Ad}_{g^{-1}}\Omega\label{eq:defConn1}\end{equation}
\item For all $X \in \mathfrak{g}$, 
\begin{equation}
\Omega(X^\sharp) = X \in \mathfrak{g} \label{eq:defConn2}
\end{equation}
\end{enumerate}
\end{defn}
Notice that here, the 1-form is on the total space $P$. Recall that a local section $\sigma\colon U \to \restr{P}{U}$ defines a local trivialization $\pi^{-1}(U) \to U \times G$ of $P$ via 
\begin{equation}
p \mapsto (\pi(p),\sigma(\pi(p))).
\end{equation}
Let $\sigma'$ be another section over $U$. Then there exists a map $g\colon U \to G$ such that $\sigma'(x) = \sigma(x)g(x)$. The following lemma describes the behaviour of a connection under such a change of trivialization.
\begin{lem}\label{lem:connTrivChange}
If $\sigma,\sigma'$ are as above, then 
\begin{equation}
(\sigma')^*\Omega = \mathrm{Ad}_{g^{-1}}\sigma^*\Omega + g^*\phi, 
\end{equation}
where $\phi \in \Omega^1(G, \mathfrak{g})$ is the Maurer-Cartan Form introduced above.
\end{lem}
\begin{proof}
Let $x \in U$, and $v\in T_xM$. Then $((\sigma')*\Omega)_xv = \Omega_{\sigma'(x)}d\sigma'_x v$. On the other hand we can write $\sigma'$ as the composition 
\[\begin{tikzcd}
U \arrow[r,"{(s,g)}"]           & P\times G \arrow[r, "{\mu}"]              & P        \\
x  \arrow[r, maps to] & {(s(x),g(x))} \arrow[r, maps to] & s(x)g(x)
\end{tikzcd} \]
We first compute the pullback $\mu^*\Omega$ to $P\times G$. For this, note that $$(d\mu_p)_gw = (d\mu_{pg})_e(dL_{g^{-1}})w$$ 
(this follows from the chain rule). Then, we have 
\begin{align*}
(\mu^*\Omega)_{(p,g)}(v,w) &= \Omega_{pg}(d\mu_p)_g w + \Omega_{pg}(dR_g)_pv \\
&= \Omega_{pg}\underbrace{(d\mu_{pg})_e(dL_{g^{-1}})_gw}_{((dL_{g^{-1}})_gw)^\sharp} + ( R_g^*\Omega)_p v \\
&= (dL_{g^{-1}})_gw + \mathrm{Ad}_{g^{-1}}\Omega_pv = \phi_gw + \mathrm{Ad}_{g^{-1}}v
\end{align*} 
where in the last equality we used the two properties of a connection. Now, the first term is exactly $\phi_g$. Pulling back to $U$ with $(s,g)$, we obtain the result since pulling back with $\sigma$ commutes with the adjoint action of $g$, which acts only on the Lie algebra factor of $\Omega$.

\end{proof}
The next proposition establishes the relationship of this definition with the one of a connection on a vector bundle.  
\begin{prop}
Let $P=(\mathfrak{U},g_{\alpha\beta}) $ be a principal bundle with structure group $G\subset GL_n(k)$, i.e. $G$ is a matrix group\footnote{Notice that the spin group is also a matrix group since the Clifford algebra is isomorphic to a matrix algebra. However, $\Spin_n \nsubseteq Gl_n(k)$, rather, $\Spin_n \subset GL_{2^n}(k)$.}. Then a connection on $P$ is equivalent to a collection of 1-forms $A_\alpha\in \Omega^1(M,\mathfrak{g})$ such that 
\begin{equation}
A_\beta = g_{\alpha\beta}A_\alpha g^{-1}_{\alpha\beta} - dg_{\alpha\beta}g^{-1}_{\alpha\beta}.
\end{equation}
\end{prop}
\begin{proof}
Suppose we are given a connection $\Omega$ on $P$ and let $U_\alpha \in \mathfrak{U}$. Consider the constant section $\sigma_\alpha \colon U_\alpha \to  U_\alpha \times G, u\mapsto(u,1).$ Then, we define $A_\alpha := (\Psi_\alpha^{-1}\circ \sigma_{\alpha})^*\Omega$. Now, notice that the section $\sigma_\alpha$ over the $U_\beta$ is given by $g_{\alpha\beta}$. Hence 
$$\sigma_\beta = \sigma_\alpha g_{\alpha\beta}^{-1}.$$
Now we can apply proposition \ref{lem:connTrivChange} for $g^{-1}$.
 Now, notice that we have 
 $0 = d(gg^{-1}) = dg g^{-1} + gd(g^{-1})$ and hence $d(g^{-1}) = - g^{-1}dg g^{-1}$. This implies that for a matrix group, we have $\phi_{g^{-1}} = gd(g^{-1}) = - dg g^{-1}$. This proves that a connection $\Omega$ is described by such 1-forms in a local trivialization. \\
Conversely, assume that we are given a family of such 1-forms. Then we set $\Omega_\alpha(u,g) := \mathrm{Ad}_{g^{-1}}A_\alpha + \phi_g$ on $U_{\alpha} \times G$ and define $\Omega$ on $\pi^{-1}(U_{\alpha})$ as $\Psi_{\alpha}^*\Omega_{\alpha}$. We then glue together the connection using a partition of unity. The resulting 1-form $\Omega$ which is a connection since the local pieces are, and the conditions \eqref{eq:defConn1} and \eqref{eq:defConn2} are convex. 
\end{proof} 
We have the following corollary: 
\begin{cor}
A connection $\nabla$ on a vector bundle $E$ induces a connection $\Omega$ on the bundle of frames $\mathrm{Fr}(E)$ and vice versa. 
\end{cor}
Hence, one can study connections on vector bundles by studying connections on principal bundles. This will be our approach in this course. 
\subsubsection{Curvature}
An important notion associated to a connection is the concept of \emph{curvature}. 
For this, we need the concept of Lie Bracket on Lie algebra-valued forms, which is defined on elements of the form $A=\alpha \otimes \xi, B=\beta \otimes \xi'$, where $\alpha\in\Omega^k(M),\beta\in\Omega^l(M),\xi,\xi'\in\mathfrak{g}$, by
\begin{align}
[\cdot,\cdot]\colon\Omega^k(M,\mathfrak{g}) \times \Omega^l(M,\mathfrak{g}) \to \Omega^{k+l}(M,\mathfrak{g}) \notag \\
[\alpha\otimes \xi, \beta \otimes \xi'] := \alpha \wedge \beta \otimes [\xi,\xi']
\end{align}
and extended bilinearly. In particular, for matrix groups we have $[\xi,\xi'] = \xi\xi' -\xi'\xi$ and then 
\begin{equation}
[A,B] = A \wedge B - (-1)^{|A||B|} B\wedge A = -(-1)^{|A||B|}[B,A]
\end{equation}
where the wedge product operation is defined by matrix multiplication: If $A=\alpha \otimes \xi, B=\beta \otimes \xi'$ as above, then
\begin{equation}
A \wedge B = \alpha \wedge \beta \otimes \xi\xi'
\end{equation}
The Lie bracket satisfies 
\begin{align}
d[A,B] &= [dA,B] + (-1)^{|A|}[A,dB] \label{eq:LieBracketLeibniz} \\ 
[A,[B,C]] &= [[A,B],C] + (-1)^{|A||B|}[B,[A,C]]\label{eq:GradedJacobi}
\end{align}
Now, the curvature is easily defined from the abstract viewpoint on connections:
\begin{defn}[Curvature]
Let $\Omega \in \Omega^1(P,\mathfrak{g})$ be a connection on a principal $G$-bundle $\pi\colon P \to M$. Then, the \emph{curvature} of $\Omega$ is the 2-form $F \in \Omega^2(P,\mathfrak{g})$ defined by 
\begin{equation}
\mathbf{F} = d\Omega + \frac{1}{2}[\Omega,\Omega]\label{eq:defCurv1}
\end{equation}
\end{defn}
We summarize some properties of the curvature as exercises. 
\begin{exc}\label{exc:curv}
Let $\mathfrak{U}$ be a local trivialization of $P$. Denote $F_{\alpha} := (s_\alpha)^*\mathbf{F} = dA_\alpha + \frac12 [A_\alpha,A_\alpha].$
Then 
\begin{equation}
F_\beta = g_{\alpha\beta}F_\alpha g_{\alpha\beta}^{-1}
\end{equation}
\end{exc}
It follows that the $F_\alpha$ define a section $F \in \Omega^2(M,\mathrm{Ad} P)$. 
\begin{exc}
Let $E \to M$ be a vector bundle and let $\nabla$ be a connection on $E$. Define the two-form $F^\nabla \in \Omega^2(M,\End E)$ by 
\begin{equation}
F^\nabla(X,Y) = \nabla_X \nabla_Y - \nabla_Y\nabla_X - \nabla_{[X,Y]}.
\end{equation}
Show that this is the curvature 2-form of the associated connection on $\mathrm{Fr}(E)$. \\
\emph{Hint:} Work over a trivializing chart and remember the formula for the de Rham differential of a 1-form: 
$$d\omega(X,Y) = X\omega(Y) - Y\omega(X) - \omega([X,Y]).$$
\end{exc}
\subsubsection{Exterior Derivative}
A connection on a principal bundle $P \to M$ induces an \emph{exterior derivative} on $\mathrm{ad} P$-valued differential forms. In a trivializing chart $U_\alpha$, it is defined by 
\begin{equation}
(d_\Omega\omega)_\alpha = d\omega_\alpha + [A_\alpha,\omega_\alpha]\label{eq:defExtDerivative}
\end{equation}
\begin{prop}
\begin{enumerate}[i)]
\item The exterior derivative in local trivializations by \eqref{eq:defExtDerivative} defines a map 
\begin{align*}
d_\Omega \colon \Omega^k(M,\mathrm{ad} P) &\to \Omega^{k+1}(M,\mathrm{ad} P) \\
\omega \mapsto d_\Omega \omega
\end{align*}
\item We have 
\begin{equation}
d_\Omega d_\Omega \omega = [F_\Omega,\omega]
\end{equation}
\item The curvature satisfies 
\begin{equation}
d_\Omega F_\Omega = 0,
\end{equation}
the Bianchi identity. 
\end{enumerate}
\end{prop}
\begin{proof}
\begin{enumerate}[i)]
\item One simply checks by direct computation that $$(d_\Omega\omega)_\beta = g_{\alpha\beta}(d_\Omega\omega)_\alpha g_{\alpha\beta}^{-1}.$$
\item By the first point, it is enough to check this in a trivializing chart. Here, again the proof is a simple computation: 
\begin{align*}
(d_\Omega d_\Omega \omega)_\alpha &= d_\Omega (d\omega_\alpha + [A_\alpha,\omega_\alpha]) \\
&= d(d\omega_\alpha) + [A_\alpha,d\omega_\alpha] + d[A_\alpha,\omega_\alpha] + [A_\alpha,[A_\alpha,\omega_\alpha]]\\
&= [A_\alpha,d\omega_\alpha] + [dA_\alpha,\omega_\alpha] - [A_\alpha,d\omega_\alpha] + \frac12[[A_\alpha,A_\alpha],\omega_\alpha]\\
&= [F_\alpha,\omega_\alpha]
\end{align*}
where we have used \eqref{eq:LieBracketLeibniz} and \eqref{eq:GradedJacobi}.

\item Again one can check this in a trivializing chart. Here we simply compute 
\begin{align*}
(d_\Omega F)_\alpha &= dF_\alpha + [A_\alpha,F_\alpha]  \\
&= d(dA_\alpha) + \frac12 d[A_\alpha,A_\alpha] + [A_\alpha,dA_\alpha]+\frac12[A_\alpha,[A_\alpha,A_\alpha]]
\end{align*}
The last term vanishes due to \eqref{eq:GradedJacobi} and the other terms cancel due to \eqref{eq:LieBracketLeibniz}.
\end{enumerate}
\end{proof}
In particular, if $\Omega$ is flat, we have that $d_\Omega \colon \Omega^\bullet(M,\mathrm{ad}P) \to \Omega^\bullet(M,\mathrm{ad}P)$ squares to zero and we can define the cohomology  
\begin{equation}
H^\bullet(M,\mathrm{ad}P) := \ker(d_{\Omega})/\mathrm{im}(d_\Omega).
\end{equation} 
\section{Chern-Simons action functional}
After establishing the necessary preliminaries, let us turn to the definition of the Chern-Simons action functional. The original reference is \cite{Chern1974}. These notes closely follow the review \cite{Freed1995}.\\
\subsection{The Chern-Simons 3-form}
 We first fix a compact, connected and simply connected matrix group $G \subset GL(n)$ (the prime example being $SU(n), n\geq 2)$) with Lie algebra $\g$. Next, fix an $\mathrm{ad}$-invariant non-degenerate symmetric bilinear form $\langle\cdot,\cdot\rangle$ on $\g$ (the prime example being the Killing form on $SU(n)$, which is a multiple of the trace). Here $\mathrm{ad}$-invariant means that 
\begin{equation}
\langle \mathrm{ad}_xy,z\rangle = - \langle y,\mathrm{ad}_xz\rangle 
\end{equation}
(equivalently, $\langle\cdot,\cdot\rangle$ is invariant under the adjoint action of $G$ on $\g$). It follows that 
$$\langle[\cdot,\cdot],\cdot\rangle \colon \wedge^3 \g \to \R$$
is a Lie-algebra 3-cocycle (i.e. it is completely antisymmetric in all 3 arguments and closed under the Chevalley-Eilenberg differential).  
\begin{defn}
Let $G \hookrightarrow P \twoheadrightarrow M$ be a principal $G$-bundle and $\Omega$ a connection on $P$. Then we we define the \emph{ first  Pontryagin form} of $P$ to be 
\begin{equation}
p_1(\Theta) = \frac12 \langle F_\Theta, F_\Theta\rangle \in \Omega^4(P).
\end{equation}
\end{defn}
\begin{prop}
For any trivializing cover $\{U_\alpha\}$, the pullbacks $s^*_\alpha p_1(\Theta)$ piece together into a global 4-form $p_1^M(\Theta) \in \Omega^4(M)$ which is closed. 
\end{prop}
\begin{proof}
If one changes the trivializing chart the pullback of the curvature gets conjugated, hence by $\mathrm{ad}$-invariance we have $p_1^M(\Theta)_\alpha = p_1^M(\Theta)_\beta$. Hence the pullbacks piece together into a global 4-form. The closedness follows from the Bianchi identity and $\mathrm{ad}$-invariance.
\end{proof}
\begin{defn}
The cohomology class $[p_1^M(\Theta)] \in H^4(M)$ is called \emph{first Pontryagin class of $P$}. 
\end{defn}

On $P$, the form $p_1(\Theta)$ is not only closed, but also exact:
\begin{defn}
The \emph{Chern-Simons 3-form} $cs(\Theta) \in \Omega^3(P)$ is defined by  
\begin{align}
cs(\Theta) &:= \frac12\langle \Theta,d\Theta\rangle + \frac16\langle\Theta,[\Theta,\Theta]\rangle \\
&= \frac12\langle \Theta, F_\Theta\rangle - \frac{1}{12}\langle\Theta,[\Theta,\Theta]\rangle.
\end{align}
\end{defn}
This Chern-Simons form is a primitive of the Pontryagin form:
\begin{prop}
We have 
\begin{equation}
d cs(\Theta) = p_1(\Theta).
\end{equation}
\end{prop}
\begin{proof}
Exercise (use the Bianchi identity!)
\end{proof}
It is interesting to observe how the Chern-Simons form transforms under gauge transformations. If $\varphi \colon P \to P$ is a gauge transformation, i.e. an automorphism of $P$, there is an associated map $g_\varphi\colon P \to G$ defined by the requirement 
\begin{equation}
\varphi(p) = p \cdot g_\varphi(p). 
\end{equation}
\begin{prop}\label{prop:cstrafoP}
We have 
\begin{equation}
\varphi^*cs(\Theta) = cs(\Theta) + \frac12 d(g_\varphi\Theta g_\varphi^{-1} \wedge g_\varphi^{-1}dg_\varphi) - \frac12 g^*_\varphi(\langle\phi,[\phi,\phi]\rangle), \label{eq:cstrafoP}
\end{equation}
where $\phi \in \Omega^1(G,\g)$ is the Maurer-Cartan form defined in \ref{def:MCform}.
\end{prop}
\begin{proof}
The proof is a straightforward computation and left to the reader as an exercise.
\end{proof}
This transformation behaviour is very different from the one of the Pontryagin 4-form.  In particular, usually there is no globally defined Chern-Simons 3-form on $M$.  However, we can ask ourselves if the cohomology class of the Pontryagin 4-form 
$[p_1^M(\Theta)] \in H^4(M)$ (called the Pontryagin class) is trivial.  A particular case when this happens is when the bundle admits a global section (i.e. it is trivial). In that case we conclude that the class of the Pontryagin 4-form is trivial in cohomology. It follows that if the cohomology class of the Pontryagin 4-form is non-trivial, then the bundle cannot be trivial - this is an example of the use of \emph{characteristic classes}. 
\subsection{The action functional}
The following is a relatively simple but crucial fact in low-dimensional Gauge Theory. 
\begin{lem}\label{lem:PBtriv}
Let $G$ be a connected and simply connected Lie group. Then any principal $G$-bundle over a manifold $M$ of dimension less than 3 is trivializable.
\end{lem}
\begin{proof}
We give a proof for the interested reader  using a little algebraic topology and obstruction theory. Every principal $G$-bundle on a manifold $M$ corresponds to a map $M \to BG$, where $BG$ denotes the classifying space of $G$. Homotopic maps correspond to isomorphic bundles. But the fact that $\pi_0(G) = \pi_1(G) = \pi_2(G) = 0$ (vanishing of first two homotopy groups follows from assumptions on $G$, vanishing of $\pi_2(G)$ is a general fact for connected Lie groups\footnote{See e.g. \cite[Chapter V]{Brocker2013}.}) implies that $\pi_1(BG) = \pi_2(BG) = \pi_3(BG) = 0$. Since $BG$ is connected, always $\pi_0(BG) = 0$ and hence\footnote{All the obstructions for two maps not to be homotopic vanish. For a deeper discussion of obstruction theory see \cite{Husemoeller1993}.} for any manifold of dimension less than 3 we have 
$$[M \to BG] \cong *,$$
i.e. all maps are homotopic to the constant map, which corresponds to the trivial bundle. Hence all principal bundles are isomorphic to the trivial one. 
\end{proof}
This fact is very much not true if the Lie group is not simply connected (e.g. $G = U(1)$). In that case, Chern-Simons theory becomes a lot more complicated. See \cite{Freed2002} for a discussion. 
\begin{defn}
Let $M$ be a compact oriented 3-manifold and $G \hookrightarrow P \twoheadrightarrow M$. Let $s$ be a global section of $P$ (guaranteed to exist by the previous Lemma). Then we define the \emph{Chern-Simons action functional} by  
\begin{equation}
S_{CS}[s,\Theta] = \int_M s^*cs(\Theta) = \int_M\frac12\langle A,dA\rangle + \frac16 \langle A,[A,A]\rangle
\end{equation}
where we defined $A := s^*\Theta$. 
\end{defn}
As a consequence of Proposition \ref{prop:cstrafoP}, this action functional is \emph{almost} invariant under gauge transformations (automorphisms of $P$). 
\begin{prop}\label{prop:Scstrafo}
Let $M$ be a closed manifold and $\varphi\colon P \to P$ be a gauge transformation. Let $s$ be a global section of $P$ and define $g = g_\varphi \circ s$. Then we have 
\begin{equation}
S_{CS}[\varphi \circ s,\Theta] = S_{CS}[s,\varphi^*\Theta] = S_{CS}[s,\Theta] - \frac{1}{12}\int_M g^*\langle \phi, [\phi,\phi]\rangle \label{eq:Scstrafo}
\end{equation}
\end{prop}
\begin{proof}
This follows directly from Equation \eqref{eq:cstrafoP}.
\end{proof}
Note that if $s$ and $s'$ are any two section of $P$, then there exists a global gauge transformation $\varphi$ with $\varphi \circ s = s'$: Concretely, $\varphi = \Phi_{s'} \circ \Phi_s^{-1}$ where $\Phi_s\colon M \times G \to P$ is the trivialization given by $\Phi_s(x,g) = s(x)g$. Hence \eqref{eq:Scstrafo} tells us two things: How to relate the action functionals in different trivializations, and how the action functional transforms under gauge transformations. \\
Proposition \ref{prop:Scstrafo} motivates the following assumption: 
\begin{ass}\label{ass:integralclass}
The bilinear form $\langle\cdot,\cdot\rangle$ is such that 
\begin{equation}
\frac{1}{12}\langle\phi,[\phi,\phi]\rangle \in H^3(G,\Z) \subset H^3(G,\R) (\cong H^4(BG,\Z))
\end{equation}

\end{ass}
This assumption implies in particular that 
\begin{equation}
\int_M g^*\langle \phi,[\phi,\phi]\rangle \in 12\Z
\end{equation}
and hence that for every integer $k$ the exponentiated Chern-Simons action 
$$e^{2\pi \ii k S_{CS}[s,\Theta]} \equiv e^{2\pi \ii k S_{CS}[
\Theta]} $$
is independent of the choice of trivialization and invariant under gauge transformations. 
\section{Critical Points}\label{sec:crit_pts}
Now that we know the space of fields of the theory (the space of connections on the (unique up to isomorphism) principal $G$-bundle on $M$) and the action functional (the integral of the Chern-Simons 3-form) the next step to understand the theory is to understand the critical points of the action functional. We do not enter into the technical details of derivatives in infinite dimensions here, but rather just define a critical point to be a connection $A$ such that for all $B \in \Omega^1(M,\mathrm{ad} P)$\footnote{The space of connections is an affine space with \emph{espace vectoriel directeur}  $\Omega^1(M,\mathrm{ad} P)$, hence a tangent vector to a connection $A$ is an element $B \in \Omega^1(M,\mathrm{ad} P)$, and a curve with this tangent vector at $0$ is simply $A+tB$.}
\begin{equation}
\frac{d}{dt}\bigg|_{t=0}S_{CS}[A + tB]  = 0.
\end{equation}
A quick computation shows that 
\begin{equation}
S_{CS}[A + tB] = S_{CS}[A] + t\int_M\langle B,dA + \frac12[A,A]\rangle + O(t^2),
\end{equation}
whence we conclude that
\begin{equation}
\frac{d}{dt}\bigg|_{t=0}S_{CS}[A + tB] = \int_M\langle B,dA + \frac12[A,A]\rangle  = \int_M \langle B, F_A\rangle.
\end{equation}
This is often rewritten as 
\begin{equation}
\delta S[A] = \int_M\langle \delta A, F_A\rangle.
\end{equation}
It follows that the critical points of the action functional are precisely the flat connections. Since the curvature $F_A$ is a two-form with values in the adjoint bundle (Exercise \ref{exc:curv}), flat connections are sent to flat connections under gauge transformations and it makes sense to ask about the quotient of the space of flat connections under gauge transformations\footnote{To be precise in what follows, here we relax the notion of gauge transformation to the one of bundle \emph{iso}morphism (instead of automorphism).}. 
\begin{defn}
For a manifold $M$ and a Lie group $G$, we define\footnote{We refer to \cite{Freed1993} for a deeper discussion on the technicalities of this quotient.}the \emph{moduli space of flat $G$-connections} by 
\begin{equation}
MFC(M,G) = \frac{\{(P,\Theta) |P \text{ principal $G$-bundle, $\Theta$ flat connection on $P$ } \}}{ \{\text{isom. of principal $G$-bundles}\}}
\end{equation}
\end{defn}
This is space of considerable interest in topology and geometry, but in this guise completely inaccessible. We will give another characterization in the next section. 
\section{The representation variety}
The moduli space of flat connections is a complicated and intriguing object. In this section we sketch a proof of the often used fact that it has an equivalent characterization that  shows how it is determined by the topology of $M$ and the algebra of the group $G$: 
\begin{thm}\label{thm:repvar}
There is a bijection 
\begin{equation}
MFC(M,G) \cong \frac{\mathrm{Hom}(\pi_1(M),G)}{G}
\end{equation}
where the group $G$ acts on group homomorphisms with codomain $G$ by conjugation. 
\end{thm}
This result is nothing less than astonishing: On the left-hand we are identifying solutions to complicated non-linear differential equation via the action of an infinite-dimensional group\footnote{Group\emph{oid}, to be completely precise.}, while the right-hand side depends on $M$ only through the first fundamental group - a rather crude topological invariant of $M$! Of course, the quotient on the right-hand side often turns out to be quite complicated too, but the theorem certainly provides a vast improvement in understanding $MFC(M,G)$. The rest of this section is devoted to a sketch of the proof of Theorem \ref{thm:repvar}, which is highly instructive in itself, mainly following \cite{Taubes2011}. 
\subsection{The horizontal distribution of a connection}
An important tool to study a connection is its horizontal distribution. 
\begin{defn}
Let $\Theta \in \Omega^1(M,\g)$ be a connection on $P$, then its  \emph{horizontal distribution} is the subbundle of $TP$ given by 
\begin{equation}
H_\Theta = \ker \Theta \subset TP
\end{equation}
\end{defn}
\begin{rem}
The dimension of this kernel at every point is $\dim M$, hence this is indeed a subbundle.
\end{rem}
\begin{prop}
\begin{enumerate}
\item Denote the kernel of $d\pi \colon TP \to TM$ by $VP$ \footnote{This is called the \emph{vertical tangent bundle} or \emph{tangent bundle along the fibers} of $P$.}. Then we have 
\begin{equation}
TP = VP \oplus H_\Theta.\label{eq:hordisti}
\end{equation} 
\item The horizontal distribution is equivariant with respect to the right $G$-action, 
\begin{equation}
(H_{\Theta})_{pg} = (H_{\Theta})_p \cdot g\label{eq:hordistii}
\end{equation}
(here the action on the right hand side is the derivative of the right $G$-action on $P$). 
\end{enumerate}
\end{prop}
\begin{proof}
\begin{enumerate}
\item Let $p \in P$, the map $\iota\colon \g \to V_pP$, $X \mapsto (X^\sharp)_p$ (the fundamental vector field of $X$ evaluated at $p$) is an isomorphism (it is injective because the $G$-action is free, and they have the same dimension). Since $\Theta_p(X^\sharp) = X$, we have that $\iota \circ \Theta_p$ is a projection to $V_pP \subset T_pP$, and $\ker(\iota \circ \Theta_p) = \ker(\Theta_p) =(H_\Theta)_p$ is a complement of $V_pP$. 
\item For $v \in T_pP$, we have $(dR_g)_pv \in T_{pg}P$ and 
\begin{equation}
\Theta_{pg}((dR_g)_pv) = (R_g^*\Theta)_pv = \mathrm{Ad}_{g^{-1}}(\Theta_p(v)).
\end{equation}
Since $(dR_g)_p, \mathrm{Ad}_{g^{-1}}$ are linear isomorphisms, we conclude that $v \in \ker \Theta_p \Leftrightarrow (dR_g)_pv \in \ker \Theta_{pg}$. This proves the claim. 
\end{enumerate}
\end{proof}
\begin{rem}
On any smooth fiber bundle $P \twoheadrightarrow M$, one can define a connection as a subbundle $H$ of $TP$ satisfying \eqref{eq:hordisti} (and define the connection one-form by the corresponding projection). Over principal bundles one asks \eqref{eq:hordistii} in addition\footnote{Some sources do not ask for this, and call connections satisfying \eqref{eq:hordistii} \emph{principal connections}, but since we only care for this type of connections, we do not need to make this distinction.}. This definition is equivalent to the Definition \ref{def:connprinc}. Working out the details is an instructive exercise.
\end{rem}
Flatness of connections corresponds to an important property of the distribution:
\begin{prop}
The horizontal distribution is integrable if and only if $\Theta$ is flat. 
\end{prop}
\begin{proof}
Let $X,Y \in \Gamma(H_\Theta)$ (i.e. $X,Y$ are vector fields tangent to $H_\Theta$). By definition, this means that $\Theta(X) = \Theta(Y) \equiv 0$. An easy computation then shows that $\Theta([X,Y]) = -F_\Theta(X,Y)$, which proves the statement. 
\end{proof}
\subsection{Parallel transport}
Any connection defines a notion of parallel transport as follows. Again let $P$ be a principal $G$-bundle over $M$. 
\begin{defn}
Let $p \in P$ and $\gamma\colon [0,1] \to M$ and path with $p\in P_{\gamma(0)}$. Then a curve $\tilde{\gamma}(t)$ is called a horizontal lift of $\gamma$ if $\pi(\tilde{\gamma}(t)) = \gamma(t)$ and $\dot{\tilde{\gamma}}(t) \in (H_\Theta)_{\tilde{\gamma}(t)}$ for all $t \in [0,1]$. 
\end{defn}
It follows from the basic theory of differential equations that horizontal lifts always exist and are unique (see e.g. \cite{Kobayashi1996}). 
\begin{defn}
Let $p \in P$ and $\gamma\colon [0,1] \to M$ with $p \in P_{\gamma(0)}$. Then we define the \emph{parallel transport along $\gamma$} by \begin{equation}
Pt(\Theta,\gamma)(p) = \tilde{\gamma}(1)
\end{equation}
\end{defn}
The parallel transport along $\gamma$ defines a map 
$$Pt(\Theta,\gamma) \colon P_{\gamma(0)} \to P_{\gamma(1)}.$$
A crucial result is the following. 
\begin{thm}
If $\Theta$ is flat, the parallel transport along a path $\gamma$ depends only on its homotopy class.
\end{thm}
Instead of giving a detailed proof, for which we refer to \cite{Kobayashi1996} or \cite{Taubes2011}, we give an easy example that gives an idea of the corresponding phenomenon.
\begin{expl} Take $G = \R$ and consider the trivial $G$-bundle $P = M \times \R \to M$. Any connection on $P$ is of the form $\Theta = dx - A$, where $x$ is the coordinate on $\R$ and $A\in \Omega^1(M)$. If $\gamma\colon[0,1]\to M$ is a path we claim that the horizontal lift starting at $p = (\gamma(0),x_0)$ is $\tilde{\gamma}(t)=(\gamma(t),x_0 + \int_{\gamma|_{[0,t]}}A)$. Indeed, we have $\dot{\tilde{\gamma}}(t) = (\dot{\gamma}(t), A_{\gamma(t)}(\dot{\gamma}(t))\frac{\partial}{\partial x})$. It follows that $\Theta(\dot{\tilde{\gamma}}(t)) =- A(\dot{\gamma}(t)) + A(\dot{\gamma(t)}) = 0$. Since the horizontal lift is unique, we conclude that parallel transport along $\gamma$ is given by 
$$(\gamma(0),x_0) \mapsto (\gamma(1), x_0 + \int_\gamma A).$$
Proving that the parallel transport depends only on the homotopy class is equivalent to proving that the parallel transport around contractible loops is trivial. If $\gamma$ is a contractible loop then $\gamma$ bounds a disk $D \subset M$ and we have $\int_\gamma A = \int_D dA = \int_D F_A$, since for abelian groups we have $F_A = dA$. T his shows that the parallel transport along contractible loops is trivial if and only if $F_A =0$. The same basic idea also applies in the non-abelian case, but some more involved concepts are needed. 
\end{expl}
\begin{defn}
Let $\gamma$ be a closed loop and $p \in P_{\gamma(0)}$. Then we define the \emph{holonomy} of $Hol_p(\Theta,\gamma) \in G$ by 
\begin{equation}
Pt(\Theta,\gamma)(p) = p \cdot Hol_p(\Theta,\gamma).  
\end{equation}
\end{defn}
\begin{lem}\label{lem:holonomy}
If $q = p \cdot g$, then $Hol_q(\Theta,\gamma) = g^{-1}Hol_p(\Theta,\gamma)g$
\end{lem}
\begin{proof}
It is a consequence of equivariance of the horizontal distribution (\eqref{eq:hordistii} ) that the parallel transport commutes with the right $G$-action. Therefore 
\begin{align*} Pt(\Theta,\gamma)(q) &= Pt(\Theta,\gamma)(p\cdot g)  =  Pt(\Theta,\gamma)(p)\cdot g \\
&=p\cdot Hol_p(\Theta,\gamma)\cdot g = p \cdot (g \cdot g^{-1}) \cdot  Hol_p(\Theta,\gamma)\cdot g = q \cdot g^{-1}Hol_p(\Theta,\gamma)g.
\end{align*}
\end{proof}
\subsection{The isomorphism}
Fix a point $p \in P$ and let $x = \pi(p) \in M$. Let $\Theta$ be  a flat connection on $P$. Then, one can define a representation 
\begin{align}\rho_{P,\Theta,p} \colon \pi_1(X,x) &\to G \notag \\
[\gamma] &\mapsto Hol_p(\Theta,\gamma)
\end{align}
\begin{lem}
The conjugacy class of $\rho_{P,\Theta,p}$ does not depend on $p$ (in particular, not on $x$). 
\end{lem}
\begin{proof}
If we change $p$ in the fiber over $x$ this follows from Lemma \ref{lem:holonomy}. Dependence on $x$ is slightly more subtle since in principle also the fundamental group changes (up to inner isomorphism). We refer to \cite{Taubes2011} for a proof. 
\end{proof}
$\rho_{P,\Theta,p}$ is called the \emph{holonomy representation of $\Theta$ at $p$\footnote{Sometimes dependence on the point $p$ is dropped and one then understands the quotient by conjugation. If $\Theta$ is not flat one can still define the holonomy map,  but it does not descend to $\pi_1$: One then usually speaks of the image of the holonomy map at $p$, which is called the \emph{holonomy group} of $\Theta$ at $p$ (the same remark about the basepoint applies)}.}
We then have the following theorem, which is the more precise version of Theorem \ref{thm:repvar}. Denote $$\mathcal{A}_{flat} = \{(P,\Theta)|\text{$\Theta$ flat connection on } P \}.$$
\begin{thm}
The map 
\begin{align*}T \colon \mathcal{A}_{flat} &\to \mathrm{Hom}(\pi_1(X,x),G)/G \\(P,\Theta) &\mapsto [\rho_{P,\Theta,p}]_G 
\end{align*}
is independent of $x$ and $p$, and the gauge equivalence class of $(P,\Theta)$. It descends to an isomorphism on the quotient: 
\begin{equation}
T\colon MFC(M,G) \longrightarrow \mathrm{Hom}(\pi_1(X,x),G)/G
\end{equation}
\end{thm}
The proof, for which we again refer to the literature, depends on the construction of an inverse map. This is given by assigning to a representation $\rho\colon \pi_1(X) \to G$ the principal $G$-bundle $P$ defined by $P = \hat{X} \times_\rho G = \hat{X} \times G /\sim$. Here $\hat{X}$ is the universal cover of $X$ on which $\pi_1(X)$ acts by deck transformations, and $(\hat{x},g) \sim (\pi \cdot \hat{x}, \rho(\pi)^{-1}g$).
\subsection{Topology and smooth structure}
To put a topology on $MFC(M,G)$ we use that the fundamental group $\pi_1(M)$ is finitely generated if $M$ is compact. If $x_1, \ldots, x_n$ are the generators of $\pi_1(M)$, then any map $\rho\colon \pi_1(M) \to G$ is completely determined by the images $(\rho(x_1),\ldots,\rho(x_n))\in G^n$. Thus we can identify $Hom(\pi(M),G)$ with a subset of $G^n$, and equip it with the subspace topology. Consequently, one can endow $MFC(M,G) \cong Hom(\pi_1(M),G)/G$ with the quotient topology\footnote{One could also endow it with the quotient topology of the defining quotient. It has been shown that these topologies agree (i.e. the map $T$ is a homeomorphism if one endows the domain with the quotient topology and the codomain with the topology discussed in this section).}. One can check that this topology is independent of the choice of generators. \\
Certain points in the moduli space of flat connections have neighbourhoods that admit a smooth neighbourhood, where ``smoothness'' means the following: A map $$f\colon U \subset \Hom(\pi_1(X),G) \to \R $$ is smooth if there is a smooth map $\tilde{f}\colon V \to \R$ such that $U \subset V \subset G^n$ and $f = \tilde{f}|_U$.   To this point, we recall that a flat connection $\Theta$ gives a differential $d_\Theta$ on $\Omega^\bullet(M,\mathrm{Ad} P)$ given in a trivializing chart by $\omega \mapsto d\omega + [A,\omega]$. A map $f \colon U \to \R$ on open subset $U \subset MFC(M,G)$ is smooth if and only if $f \circ \pi\colon \pi^{-1}(U) \to \R$ is smooth. Finally, a point $x \in MFC(M,G)$ is smooth if and only if it there is a neighbourhood $U$ of $x$ and a homeomorphism $\phi\colon U \to \phi(U) \subset \R^N$ such that $f\colon U \to \R$ is smooth if and only if $f \circ \phi^{-1}$ is smooth. We refer to 
\cite{Walker1992} for a proof of the following.
\begin{thm}
\begin{enumerate}
\item A point $[\Theta] \in MFC(M,G)$ is smooth iff $H_\Theta^0(M,\mathrm{ad}P) = H^0(M) \otimes Z(\g)$, where $Z(\g)$ is the center of $\g$.
\item The Zariski tangent space to $MFC(M,G)$ at $\Theta$ is $H_\Theta^1(M,\mathrm{ad}P)$. In particular, at smooth points, $[\Theta] \in MFC(M,G)$ has a neighbourhood diffeomorphic to an open subset of $H_\Theta^1(M,\mathrm{ad}P)$. 
\end{enumerate}
\end{thm}

\chapter{Perturbative Quantization}\label{ch:PertQuant}
After considering the classical Chern-Simons theory, in this chapter we explain the methods of perturbative quantization that we will later use for Chern-Simons theory. Mathematically, this method is composed of two simple steps: 
\begin{enumerate}
\item Understand the behaviour of a finite-dimensional integral $I(\hbar) = \int_Xe^{\frac{\ii}{\hbar}S}\mu$ as $\hbar \to 0$. We will see that it depends only on $X$ through $Crit(S)$. Otherwise it is completely determined by $S$
\item Extend the formula resulting from the considerations above to the field theory situation. 
\end{enumerate}
We will see that special care is needed for gauge theories (like Chern-Simons theory). A couple of references explaining this approach are \cite{Polyak2005},\cite{Reshetikhin2010},\cite{Mnev2017}.
We start with the asymptotics of oscillatory integrals, which is a classic topic of microlocal analysis. We loosely follow the presentation of \cite{Mnev2017}. 
\section{Asymptotics of oscillatory integrals}
We recall briefly a few facts about derivatives on manifolds. 
If $X$ is a manifold and $S \colon X \to \R$ is a smooth function, then a point $p \in X$ is called critical if the differential $dS\colon T_pX \to \R$ vanishes. In that case, there is a coordinate-independent bilinear form $H_pS$ called the \emph{Hessian} of $S$ at $p$. With respect to any coordinate system it takes the form 
\begin{equation}
(H_pS)_{ij} =\restr{ \frac{\de^2S}{\de x^i \de y^j}}{p}
\end{equation}
Under a change of coordinates it transforms as a bilinear form on $T_pX$. In particular, its determinant is not well-defined but depends upon choice of coordinates\footnote{ Under coordinate changes, its determinant changes with the square of the determinant of the coordinate change, as opposed to the determinant of an endomorphism, which is invariant under change of coordinates.}. The higher derivatives, however, do not transform tensorially, but rather as jets\footnote{In the usual tautological way, jets are defined to be section of bundles that transform like derivatives of a function under changes of coordinates.}. We say that $p$ is a non-degenerate critical point of $S$ if the Hessian is non-degenerate. \\ 
Also, let us briefly explain the notation for asymptotic behaviour: Namely, we define 
\begin{equation}
f(\hbar) \simeq_{\hbar \to 0} O(\hbar^N) \colon\Leftrightarrow \exists C_N\in \R, \left|\frac{f(\hbar)}{\hbar^N}\right| \leq C_N
\end{equation}
Similarly one defines 
\begin{equation}
f(\hbar) \simeq_{\hbar \to 0} g(\hbar) + O(\hbar^N) \colon \Leftrightarrow f(\hbar) - g(\hbar) \simeq_{\hbar \to 0} O(\hbar^N).
\end{equation}
\subsection{Stationary phase formula}
The first result is the following theorem, also known as the \emph{stationary phase formula}. 
\begin{thm}\label{thm:stationaryphase}
Let $X$ be a compact manifold with $\dim X = n$ and $S\colon X \to \R$ have finitely many non-degenerate critical points. Then we have
\begin{equation}
I(\hbar)\colon=\int_X  e^{\frac{\ii}{\hbar}S}\mu \simeq_{\hbar \to 0} (2\pi\hbar)^{\frac{n}{2}} \sum_{x_0 \in Crit(S)}e^{\frac{\ii}{\hbar}S(x_0)} \frac{e^{\frac{\ii\pi}{4}\mathrm{sign} H_{x_0}S}}{|\det H_{x_0}S|^{\frac12}}\mu_{x_0} + O(\hbar^{1+n/2})\label{eq:stationaryphase}
\end{equation}
where one chooses coordinates $y^1,\ldots,y^n$ around $x_0$ to define the determinant of the Hessian  and $\mu_{x_0}$ by $\mu(x_0) = \mu_{x_0}dy^1 \wedge \ldots \wedge dy^n$. 
\end{thm}
Note that the expression on the right hand side is independent of the choice of coordinates: The signature $\mathrm{sign}H_{x_0}$ is independent under coordinate change, and the changes in $\det H_{x_0}$ and $\mu_{x_0}$ cancel. \\
The proof of this theorem follows from a sequence of lemmata. 
\begin{lem}[Fresnel integrals]\label{lem:Fresnel}
Let $Q$ be a symmetric bilinear form on $\R^n$, and $Q_0$ any positive definite symmetric bilinear form. Then 
\begin{equation}
\lim_{\varepsilon \to 0} \int_{\R^n} d^nx e^{iQ(x,x) - \varepsilon Q_0(x,x)} = \pi^{n/2}\frac{e^{\frac{\ii\pi}{4}\mathrm{sign}Q}}{|\det Q|^{1/2}}
\end{equation}
\end{lem}
The proof of this lemma follows from the one-dimensional case, which is proven by standard complex analysis, together with a change of coordinates. We interpret the left-hand side as the definition of $\int_{\R^n}e^{iQ(x,x)}$ (which is only conditionally convergent). 
\begin{lem}\label{lem:statphaseR}
Let $g \in C_c^\infty(\R)$ and consider $I(\hbar) = \int_\R g(x)e^{\frac{\ii}{\hbar}x}dx$.
Then \begin{equation}I(\hbar) \simeq_{\hbar \to 0} O(\hbar^\infty) \Leftrightarrow I(\hbar) \simeq_{\hbar \to 0}  O(\hbar^N) \forall N \in \N.
\end{equation}
\end{lem}
The proof is just integration by parts: 
\begin{proof}
Let $N \in \N$, then 
$$\left|\frac{I(\hbar)}{\hbar^N}\right| = \left|\int_\R g(x)\left(-i\frac{d}{dx}\right)^Ne^{\frac{\ii}{\hbar}x}dx\right| \leq \int_\R \left|g^{(N)}(x)\right|dx =\colon C_N.$$ 
\end{proof}
Here we are using the ``phase function'' (the multiplier of $\frac{\ii}{\hbar}$ in the exponent) $f(x) = x$, which has no critical points on  $\R$. The conclusion is that the integral vanishes faster than any power of $\hbar$. This is true also in the multi-dimensional case: 
\begin{lem}\label{lem:nocritpts}
Let $g\in C_c^\infty(\R^n)$ and $f \in C^\infty(\R^n)$ such that $f$ has no critical points on the support of $g$. Then 
\begin{equation}
\int_{\R^n}d^nx g(x) e^{\frac{\ii}{\hbar}f(x)} \simeq_{\hbar \to 0} O(\hbar^\infty)
\end{equation}
\end{lem}

\begin{proof}
W.l.o.g. we can assume that $f$ has no critical points at all. Thus the sets $f^{-1}(y), y \in \R$ are embedded submanifods in $\R^n$ and we can rewrite the integral as 
$$\int d^nx g(x) e^{\frac{\ii}{\hbar}f(x)} = \int_{\R}dy e^{\frac{\ii}{\hbar}y}\int_{f^{-1}(y)}g(x)dvol_{f^{-1}(y)}(x),$$
which brings us in the situation of the previous lemma. 

\end{proof}
Finally, let us look at the situation where there is a single non-degenerate critical point. This can be described by the following situation: 
\begin{lem}\label{lem:quadcritpt}
Let $Q$ be a non-degenerate quadratic form on $\R^n$ and $g \in S(\R^n)$. Let $g_N$ be the $N$-th Taylor expansion of $g$. Define 
$$I(\hbar) := \int_{\R^n} (g(x) - g_N(x)) e^{\frac{\ii}{\hbar}Q(x,x)}.$$ 
Then 
\begin{equation}
I(\hbar) \simeq_{\hbar \to 0} O(\hbar^{n/2 + \lfloor \frac{N+2}{2}\rfloor})
\end{equation}
\end{lem}
\begin{proof}
The proof goes again through integration by parts. Let $Q(x,x)= \frac12 Q_{ij}x^ix^j$, then  
$$\frac{\de }{\de x^j}e^{\frac{\ii}{\hbar}Q(x,x)} = \left(\frac{2\ii}{\hbar}\right) x^iQ_{ij}e^{\frac{\ii}{\hbar}Q(x,x)}.$$
We conclude that $$D = \frac{-\ii}{2}(Q^{-1})^{ij}\frac{1}{x^i}\frac{\de }{\de x^j}$$ satisfies $$De^{\frac{\ii}{\hbar}Q(x,x)} = \frac{1}{\hbar}e^{\frac{\ii}{\hbar}Q(x,x)}$$ and hence 
$$\frac{I(\hbar)}{\hbar^m}= \int_{\R^n}d^nx (g(x)-g_N(x))D^me^{\frac{\ii}{\hbar}Q(x,x)}.$$ To integrate by parts we have to find the operator $D^T$ such that $\int_{\R^n} f Dg d^nx = \int_{\R^n}(D^Tf)g d^nx$, this operator is given by 
$$D^T = \frac{\ii}{2}(Q^{-1})^{ij}\frac{\de }{\de x^j}\frac{1}{x^j}$$
as one immediately verifies. Since $g_N$ is the $N$-th Taylor approximation of $g$, we have $|g(x)-g_N(x)| = C|x|^{N+1}(1 + O(|x|)$, where $C$ is a constant (possibly 0). Applying the operator $D^T$ decreases the power of the absolute value by 2 (once by dividing and once by taking a derivative). Thus 
$$|(D^m)^T(g-g_N)(x)| = C'|x|^{N+1-2m}(1+O(|x|)).$$
It is an elementary exercise (using e.g. polar coordinates) to verify that, if $k \in \Z,\varepsilon>0$ and $n\in \N$, we have that $$\int_{D_\varepsilon(0)}|x|^kd^nx < \infty \Leftrightarrow k > -n.$$
We conclude that $\left|\frac{I(\hbar)}{\hbar^m}\right| \leq C_m$ if $N+1 - 2m >  -n \Leftrightarrow m < \frac{N+n+1}{2}$. This is almost the estimate claimed in the Lemma, for the remaining part (which is not too difficult) we refer to \cite{Mnev2017}.
\end{proof}
We are now ready to prove the stationary phase formula \eqref{eq:stationaryphase}. 
\begin{proof}[Proof of Theorem \ref{thm:stationaryphase}]
To prove the stationary phase formula, cover the manifold $X$ with charts $\{U_\alpha\}$ such that each of the critical points $x_1,\ldots,x_N$ is contained in exactly one chart. Let $\{\psi_{\alpha}\}$ be a partition of unity subordinate to the open cover $\{U_\alpha\}$. Then, we have 
\begin{align*}
\int_Xe^{\frac{\ii}{\hbar}S}\mu &= \sum_{\alpha}\int_{U_\alpha}d^ny \rho_\alpha(y)\psi_\alpha(y) e^{\frac{\ii}{\hbar}S(y)} \\\text{(compact support of $\psi_\alpha$)} &=\sum_{\alpha}\int_{\R^n}d^ny \rho_\alpha(y)\psi_\alpha(y) e^{\frac{\ii}{\hbar}S(y)} \\
\text{Lemma \ref{lem:nocritpts}} &\simeq \sum_{k = 1}^N\int_{\R^n}d^ny \rho_\alpha(y)\psi_{\alpha}(y)\exp\left(\frac{\ii}{\hbar}\left(S(x_k) + \frac12(H_{x_k})_{ij}y^iy^j + P(y)\right)\right) + O(\hbar^\infty)\\
\end{align*}
Here we have denoted $P(y)$ the terms of degree at least 3 in the Taylor series of $S$. 
Applying Lemma \ref{lem:quadcritpt} with $g=\rho\psi_\alpha e^{\frac{\ii}{\hbar}P(y)}$ and  $N=0$ we find 
\begin{align*}
\sum_{k=1}^N \int_{\R^n}d^ny g e^{\frac{\ii}{2\hbar}(H_{x_k}S)_{ij}y^iy^j} &\simeq_{\hbar \to 0} \sum_{k=1}^N e^{\frac{\ii}{\hbar}S(x_k)}\int_{\R^n}d^n(y)\rho(x_0) e^{\frac{\ii}{2\hbar}(H_{x_k}S)_{ij}y^iy^j} + O(\hbar^{1+n/2}) \\
\text{Lemma \ref{lem:Fresnel}}&=\sum_{k=1}^Ne^{\frac{\ii}{\hbar}S(x_k)}(2\pi\hbar)^{n/2}\rho(x_0)\frac{e^{\frac{\ii\pi}{4}}\mathrm{sign}H_{x_k}S}{|\det H_{x_k}S|^{1/2}} + O(\hbar^{1+n/2})
\end{align*}
which concludes the proof.
\end{proof}
\subsection{Higher order corrections}

With just a little extra effort, we can derive a closed formula for the asymptotic behaviour to all orders in $\hbar$. The importance of Fresnel moments becomes obvious from the following Lemma. 
\begin{lem}\label{lem:inftyapprox}
Assume that $S\colon X \to R$ has finitely many non-degenerate critical points $x^1,\ldots,x^N$ and around each critical point there are coordinates $y^1,\ldots,y^n$ in which $\mu = \mu_{x_k}dy^1\wedge \ldots \wedge dy^n$ (with $\mu_{x_k}$ constant). Then 
\begin{equation}
I(\hbar) = \int_Xe^{\frac{\ii}{\hbar}S}\mu = \sum_{k=1}^N\mu_{x_k}e^{\frac{\ii}{\hbar}S(x_k)} \int_{\R^n}d^n y e^{\frac{\ii}{2\hbar}(H_{x_k}S)_{ij}y^iy^j}e^{\frac{\ii}{\hbar}P(y)}  + O(\hbar^\infty)
\end{equation}
\end{lem}
\begin{proof}
We just repeat the proof of Theorem \ref{thm:stationaryphase} with $N=\infty$. The Taylor series of $g = \rho(y)\psi_\alpha(y)e^{\ii/\hbar P(y)}$ depends only on $P$, since both $\rho$ and $\psi_\alpha$ are constant in a neighbourhood of $x_k$ (again, $P$ is the degree  $\geq 3$ part of the Taylor series of $S$ at $x_k$). 
\end{proof}
Expanding $e^{\ii/\hbar P(y)}$ in a power series, we are hence led to study the Fresnel moments 
$$\int d^n y e^{\frac{\ii}{2\hbar}Q_{ij}y^iy^j}y_{i_1}\ldots y_{i_n}.$$
The standard way to do this is to introduce the generating function 
\begin{equation}
Z^{Q}[J] := \int_{\R^n}e^{\frac{\ii}{2\hbar}Q_{ij}y^iy^j + J_iy^i}d^ny 
\end{equation}
(we often write just $Z[J]$ if $Q$ is understood from the context) which has the property that 
\begin{equation}
\int d^n y e^{\frac{\ii}{2\hbar}Q_{ij}y^iy^j}y_{i_1}\ldots y_{i_n} = \frac{\de }{\de J_{i_1}}\cdots \frac{\de }{\de J_{i_n}}\bigg|_{J=0} Z^{Q}[J]
\end{equation}
Also, $Z[J]$ can be explicitly computed by completing the square: 
\begin{equation}
Z[J] = Z[0]e^{\frac{\hbar}{2\ii}(Q^{-1})^{ij}J_iJ_j}
\end{equation}
and we know $Z^{Q}[0]$ from Lemma \ref{lem:Fresnel}. We thus arrive at the following proposition: 
\begin{prop}
Denote $H_{x_k}S =:Q_k$, $Z_k[J] =:Z^{Q_k}[J]$ and $P_k(y)$ the terms of degree 3 and higher in the Taylor series of $S$. With the assumptions of Lemma \ref{lem:inftyapprox}, we then have
\begin{equation}
I(\hbar) \simeq_{\hbar \to 0} \sum_{k=1}^N \mu_{x_k}e^{\frac{\ii}{\hbar}S(x_k)}Z_k[0]\restr{e^{\frac{\ii}{\hbar}P_k\left(\frac{\de}{\de J}\right)}}{J=0}\frac{Z_k[J]}{Z_k[0]} + O(\hbar^{\infty})\label{eq:asy_exp}
\end{equation}
\end{prop}
\begin{rem}
A priori, this series depends on the choice of coordinates around the critical points, because the higher derivatives of the action do. However, one can show that this dependence cancels out in the sum over all diagrams at every order \cite{Johnson-Freyd2010}. 
\end{rem}
\begin{rem}
The limit $\hbar \to 0$ is known as the ``semiclassical'' limit in physics, whereas ``perturbative'' usually refers to taking the coupling constant(s) to 0. At least in the case where there is a single coupling constant, the two expansions are equivalent, as one can see by rescaling the fields with $\sqrt{\hbar}$. 
\end{rem}
\subsection{Feynman diagrams}
To label the terms in \eqref{eq:asy_exp}, the  asymptotic expansion of oscillatory integrals to all orders in $\hbar$, it is convenient to introduce Feynman diagrams. We try to give a self-contained but slightly condensed introduction here, but there are plenty of excellent sources in the literature. A very pedagogical introduction for mathematicians is \cite{Polyak2005}. In this exposition we follow closely \cite{Reshetikhin2010}, \cite{Mnev2017}. In physics the use of Feynman diagrams is usually derived somewhat differently\footnote{See e.g. \cite{Peskin1995} or any other QFT textbook for an account of this.}, but the outcome is completely equivalent. Mathematically, Feynman diagrams label the terms appearing in Gaussian (or Fresnel) moments. The physical interpretation is that they represent processes that happen between particles. The corresponding Gaussian (or Fresnel) moment is interpreted as the probability amplitude of that process.  \\
We thus set out for a graphical representation of $\restr{e^{\frac{\ii}{\hbar}P_k\left(\frac{\de}{\de J}\right)}}{J=0}Z_k[J]$. We briefly introduce some combinatorial terminology: 
\begin{defn}
Let $I$ be a finite set. 
\begin{itemize} 
\item A \emph{partition} of $I$ is a collection $\mathcal{P} = \{I_1,\ldots,I_n\}$ of pairwise disjoint subsets of $I$ such that $I = \bigcup_{k=1}^n I_i$. The set of all partitions of $I$ is denoted $\mathfrak{P}_I$. The set of all partitions of $\{1,\ldots,n\}$ is denoted $\mathfrak{P}_n$
\item A \emph{perfect matching} $m$ on $I$ is a partition of $I$ into two-element subsets. The set of all perfect matchings is denoted $\mathfrak{M}_I$. The set of all perfect matchings of $\{1,\ldots,n\}$ is denoted $\mathfrak{M}_n$.
\end{itemize}
\end{defn}
A perfect matching on $I$ is given by $\mathfrak{m} = \{\{a_1(\m),b_1(\m)\},\ldots,\{a_n(\m),b_n(\m)\}\}$ where $a_i(\m) \neq b_i(\m), \{a_i(\m),b_i(\m)\} \cap \{a_j(\m),b_j(\m)\}= \emptyset$ and $I = \bigcup_{j=1}^n \{a_j(\m),b_j(\m)\}$. Notice that $\mathfrak{M}_n$ is empty if $n$ is odd. 
A central step is the following Lemma often called Wick's Lemma. 
\begin{lem}[Wick's Lemma]\label{lem:Wick}
\begin{equation}
 \restr{\frac{\de }{\de J_{i_1}}\cdots\frac{\de }{ \de J_{i_n}}}{J=0} Z[J]/Z[0] = \begin{cases} 0 &  n \text{ odd} \\
 \left(\frac{\hbar}{\ii})\right)^m \sum_{\m \in \mathfrak{M}_n} (Q^{-1})^{i_{a_1(\m)}i_{b_1(\m)}}\ldots(Q^{-1})^{i_{a_l(\m)}i_{b_l(\m)}}& n = 2m
 \end{cases}
 \end{equation}
\end{lem}
\begin{proof}
If $n$ is odd, the claim follows from noticing that $Z[J]/Z[0] = e^{\frac{\hbar}{2\ii}(Q^{-1})^{ij}J_iJ_j}$ does not contain terms of odd orders. If $n = 2m$, we realize that the only surviving term is the order $m$ term in the exponential series 
\begin{align*} \restr{\frac{\de }{\de J_{i_1}}\cdots\frac{\de }{ \de J_{i_n}}}{J=0} Z[J]/Z[0] &=  \restr{\frac{\de }{\de J_{i_1}}\cdots\frac{\de }{ \de J_{i_n}}}{J=0}\left(\frac{\hbar}{2\ii}\right)^m\frac{1}{m!}((Q^{-1})^{ij}J_iJ_j)^m \\
&= \restr{\frac{\de }{\de J_{i_1}}\cdots\frac{\de }{ \de J_{i_n}}}{J=0}\left(\frac{\hbar}{2\ii}\right)^m\frac{1}{m!}(Q^{-1})^{k_1l_1}J_{k_1}J_{l_1}\cdots(Q^{-1})^{k_ml_m}J_{k_m}J_{l_m}
\end{align*}
 A term in this sum survives if and only if every derivative can be matched to a $J$, that is there is a permutation $\sigma \in S_n$ such that $(i_{\sigma(1)},\ldots,i_{\sigma(n)}) = (k_1,l_1,\ldots,k_m,l_m)$.  Such a permutation defines a perfect matching $\m = \{\{(\sigma(1),\sigma(2)\},\ldots,\{\sigma(n-1),\sigma(n)\}\}$. Two permutations give the same term in the sum precisely if they correspond to the same matching. The claim now follows from the observation that the same matching appears $2^m m!$ times (we can exchange the two elements of a pair and permute the pairs among themselves).  
\end{proof}
Thus we have explained how every term in $e^{\ii/\hbar P (\de/\de J)}$ acts on $Z[J]/Z[0]$. To go one step further, we expand the formal power series $P$ (remember it starts in degree $3$) 
\begin{equation}
P(y) = \frac{1}{3!}P_{ijk}y^iy^jy^k + \frac{1}{4!}P_{ijkl}y^iy^jy^ky^l + \ldots = \sum_{k = 3}^{\infty}\frac{1}{k!}P_{i_1i_2\ldots i_k}y^{i_1}\cdots y^{i_k} \label{eq:defP}
\end{equation}
where all $p$'s are symmetric in all indices. We now expand the exponential using the multinomial theorem as 
\begin{align}
\exp(\ii\hbar P(\de/\de J)) &= \sum_{N=0}^\infty\left(\frac{\ii}{\hbar}\right)^N\frac{1}{N!}P(\de/\de J) ^N \notag\\
&=  \sum_{N=0}^\infty\left(\frac{\ii}{\hbar}\right)^N\frac{1}{N!}\sum_{j_3+\ldots+j_l = N}{N \choose j_1 \cdots j_l}\left(\frac{1}{3!}P_{3}(\de/\de J)\right)^{j_3} \cdots\left(\frac{1}{l!}P_{l}(\de/\de J)\right)^{j_l} \notag \\
&=\sum_{l=3}^\infty\sum_{j_3,\ldots,j_l=0}^\infty\left(\frac{\ii}{\hbar}\right)^{\sum j_i}
\frac{1}{j_3!(3!)^{j_3}\cdots j_l!(l!)^{j_l}} \notag\\
&\hspace{2cm}P_{i_1i_2i_3}\cdots  P_{i_{3j_3-2}i_{3j_3-1}i_{3j_3}}\cdots P_{i_{n-l+1}\ldots i_n}\frac{\de }{\de J_{i_1}}\cdots\frac{\de }{\de J_{i_n}}, \label{eq:formula1}
\end{align}
where in the last line we defined $n = \sum_{j=1}^ln_j j $ and enumerated the vertices accordingly. We want to apply Wick's Lemma \ref{lem:Wick}.  To keep track of the corresponding terms one can introduce graphs. We give here a definition adapted to our needs. 
\begin{defn}
\begin{itemize}
\item 
A \emph{graph} $\Gamma=(H,V,E)$ consists of a finite set $H$ together with a partition $V$ of $H$ and a perfect matching $E$ of $H$. 
\item Two graphs $\Gamma = (H,V,E)$, $\Gamma' = (H',V',E')$ are \emph{isomorphic} if there is a bijection $\varphi\colon H\to H'$ with $\varphi(V) = V'$ and $\varphi(E)=E'$.
\end{itemize} 
\end{defn}
\begin{rem}\label{rem:isographs}
Let $H=\{1,\ldots,n\}$, and $\Gamma=(H,V,E)$ be a graph. Suppose $V$ has $n_j$ blocks of size $j=1,\ldots l$, $n= \sum n_jj$. Then $\Gamma$ is isomorphic to a graph where $V$ is the standard partition $\calP_0$ with $n_j$ blocks of size $j$, i.e. $\calP_0 = \{\{1\},\ldots,\{n_1\},\{n_1+1,n_1+2\},\ldots,\{n-l+1,\ldots,n\}\}$. 
\end{rem}
We introduce some further terminology. The set $H$ is called the set of \emph{half-edges}, $V$ is called the set of \emph{vertices}, and $E$ is called the set of \emph{edges}. 

Next, we consider the automorphism group  $\mathrm{Aut}(\Gamma)$ of a graph. To this end, we first note that the symmetric group $S_n$ acts on $\mathcal{P}_n$, stabilizing the number and size of blocks. The stabilizer subgroup $(S_n)_\calP$ of a partition\footnote{Two partitions of a set $I$ are the same if they coincide as sets of subsets of $I$.} $\calP$ with $n_j$ blocks of size $j$ is isomorphic to $\prod_{j} S_{n_j} \rtimes (S_j)^{n_j}\subset S_n$.  The following proposition is immediate:
\begin{prop} 
Let $\Gamma = (H,V,E)$ be a graph. Then the automorphism group of $\Gamma$ is the stabilizer group of the pair $(V,E)$ under the diagonal action of $S_H$. 
\end{prop}

We are now ready to give an expression for  $
\restr{e^{\frac{\ii}{\hbar}P\left(\frac{\de}{\de J}\right)}}{J=0}\frac{Z[J]}{Z[0]}$. To this end we define the Feynman weight of a graph. 
\begin{defn}
Let $(P_k)_{k=1}^\infty$ be a family of symmetric tensors $P_k \in \mathrm{Sym}^{V^*}$ and $K \in \mathrm{Sym}^2V$. Let $\Gamma=(H,V,E)$ be a graph. Let $L=\{l \colon H \to \{1,\ldots, |H|\}, l \text{ bijective }\}$ be the set of all labeling of the half-edges of $\Gamma$. Then we define 
\begin{equation}
F^{P,K}(\Gamma) \equiv F(\Gamma) = \sum_{l\in L}\prod_{v=\{h_1,\ldots,h_{|v|}\} \in V}P_{l(h_1)\ldots l(h_{|v|})}\prod_{e=\{h_1,h_2\} \in E}K^{l(h_1)l(h_2)}
\end{equation}
\end{defn} 
The following Proposition follows from symmetry of the tensors: 
\begin{prop}
If $\Gamma$ is isomorphic to $\Gamma'$, then $F(\Gamma) = F(\Gamma')$.
\end{prop}
The main theorem of this section is the following: 
\begin{thm}
If $P$ is a formal power series on $V^*$ starting in degree 3, and $Q$ is a non-degenerate symmetric bilinear form on $V$, then 
\begin{equation}\restr{e^{\frac{\ii}{\hbar}P\left(\frac{\de}{\de J}\right)}}{J=0}\frac{Z^Q[J]}{Z^Q[0]} =\sum_{[\Gamma]}\frac{(-\ii\hbar)^{-\chi(\Gamma)}}{|\mathrm{Aut}(\Gamma)|}F^{P,Q^{-1}}(\Gamma).
\end{equation}
Here the sum goes over isomorphism classes of graphs that are at least trivalent (all the blocks in the partition $V$ have size at least 3), we use the components of $P$ to define symmetric tensors $P_k$ as in \eqref{eq:defP}, and $\chi(\Gamma)$ is the Euler characteristic $\chi(\Gamma) = V(\Gamma)-E(\Gamma)$. 
\end{thm}
\begin{proof}
We start with the expression in Equation \eqref{eq:formula1} and apply Wick's Lemma \eqref{lem:Wick}. Every term in the sum is labeled by a standard partition $\calP_0[n]$ of $n= \sum j n_j$ with $n_j$ blocks of size $j$. By the Wick Lemma, we get a sum over all matchings of $[n]=\{1,\ldots,n\}$. The term corresponding to the perfect matching $\m$ is precisely $F(\Gamma)$, where $\Gamma = ([n],\calP_0[n],\m)$. By Remark \ref{rem:isographs}, we obtain all isomorphism classes of graphs in this way. Notice that they appear with the correct power $(-\ii\hbar)^{|E|-|V|}$ of $-\ii\hbar$. The only question left is the combinatorial factor. The stabilizer group of the standard partition $(S_n)_{\calP_0}\subset S_n$ acts on graphs $(\calP_0,m).$ It satisfies $|(S_n)_{\calP_0}|= \prod n_j! (j!)^{n_j}!$ and graphs isomorphic to $(n,\calP_0,\m)$ are precisely given by the $\m$-orbit $(S_n)_{\calP_0}\m$ of this action. Thus, 
\begin{align*}
\restr{e^{\frac{\ii}{\hbar}P\left(\frac{\de}{\de J}\right)}}{J=0}\frac{Z^Q[J]}{Z^Q[0]} &= \sum_{[\Gamma=([n],\calP_0,\m)]}\frac{(-\ii\hbar)^{1-\chi(\Gamma)}}{|(S_n)_{\calP_0}} |(S_n)_{\calP_0}\cdot \m| F^{P,Q^{-1}}(\Gamma) \\
&= \sum_{[\Gamma]}\frac{(-\ii\hbar)^{1-\chi(\Gamma)}}{|\mathrm{Aut}(\Gamma)|}F^{P,Q^{-1}}(\Gamma),
\end{align*}
where the last equality follows from the orbit-stabilizer theorem\footnote{Thanks to M.Berghoff for pointing this out.} 
$$|\mathrm{Aut}(\Gamma)| = |((S_n)_{\calP_0})_\m| = \frac{(S_n)_{\calP_0})}{|(S_n)_{\calP_0}|}.$$
\end{proof}
\subsection{Feynman graphs and rules}\label{sec:FeynmanGraphs}
Above we have discussed graphs describing the asymptotic behaviour of the partition function. The vertices and edges were all indistinguishable. One can easily extend this discussion to keep track of different terms in the quadratic operator by ``decorating'' the graphs.  Automorphisms of graphs then have to be replaced by automorphisms of decorated graphs. Loosely, one says that one computes a quantity ``by Feynman graphs and rules'': One just specifies the different types of vertices and edges appearing in the graphs - the possible graphs are then called the ``Feynman graphs''. One specifies how to compute the Feynman weight of a graph by specifying it on the generators - these prescriptions are called the ``Feynman rules'', and they can be read off from the action functional. This often provides an elegant way of generating all the terms in a complicated expression such as $$\restr{e^{\frac{\ii}{\hbar}P\left(\frac{\de}{\de J}\right)}}{J=0}\frac{Z^Q[J]}{Z^Q[0]},$$or other terms of interest.  

\section{Oscillatory integrals with degenerate phase functions and Faddeev-Popov method} 
Many action functionals that appear in physics actually do not have non-degenerate critical points, including the one that the forms the main focus of these lectures, the Chern-Simons action functional. In that case, the Hessian at the critical point has a kernel  and one can ask if the vectors in the kernel can be extended to symmetries of the action. In that case, there is a (local or global) distribution of symmetry $\calV \subset TM$, i.e. all vector fields tangent to $\calV$ annihilate the action $L_VS = 0$. $\calV$ might or might not be integrable. However, in many examples, a stronger statement is true: Not only do these vector fields exist, but they actually come from a group action on the space of fields. In these cases one can use the Faddeev-Popov method. \\
\subsection{Setup}
Again, we consider the integral 
\begin{equation}
I(\hbar) = \int_X e^{\ii/\hbar S}\mu.
\end{equation}
We now assume that there is a Lie group $G$ of dimension $k$ which acts freely on $X$ such that $S$ and $\mu$ are $G$-invariant. We denote the quotient map by $p\colon X \to X/G$. The dimension of the quotient is $n-k=:l$. In addition we ask that $\mu$ is horizontal\footnote{I.e. that $L_v\mu = 0$ for all vertical vector fields, where a vector field $v$ is vertical if $dp (v) \equiv 0$.}. We denote this action by 
$\rho\colon G \times X \to X$, $\rho(g,x) = g\cdot x$, and the corresponding infinitesimal action by $\rho^\#\colon \g \to \Gamma(TM)$, i.e.
\begin{equation}
\rho^\#(\xi)_x = \restr{\frac{d}{dt}}{t=0}\exp(t\xi)\cdot x \in T_xX.
\end{equation}
We frequently denote $\rho^\#(\xi)_x = \rho_x^\#(\xi)$, this defines, for all $x\in X$, a map $\rho_x^\# \colon \g \to T_xX$. Since the action is free, we can write 
\begin{equation}
\int_X e^{\frac{\ii}{\hbar}S}\mu = Vol(G) \int_{X/G} e^{\frac{\ii}{\hbar}\tilde{S}}\tilde{\mu}.
\end{equation}
Here $S = p^*\tilde{S}$ and $\mu = p^*\tilde{\mu} \wedge \chi$, $\tilde{S},\mu{S}$ are a function (resp. a volume form) on the quotient and $\chi \in \Omega^k(y)$ is such that $\iota_{v_l}\ldots\iota_{v_k}\chi =1$, where $v_a = \rho^\#(T_a)$ are fundamental vector fields for the Lie algebra action (i.e. $T_a$ is a basis for the Lie algebra $\g$). The integral over the quotient now admits an asymptotic expansion if the following assumption is satisfied: 
\begin{ass} The critical orbits of $S$ are isolated and $\tilde{S}$ has a non-degenerate Hessian at these orbits. 
\end{ass}
Even if this assumption is satisfied, we still often do not have a good way to think of the quotient and $\tilde{S}$, $\tilde{\mu}$ (think e.g. of the action of gauge transformations on connections). The Faddeev-Popov method gives an alternative way to compute the integral by means of more tractable data. 
\subsection{Gauge fixings} 
Let $x_0$ be a critical point of $S$. We say that a chart $(U;y^1,\ldots,y^{l},z^1,\ldots,z^k)$ is \emph{adapted} if $\left(\frac{\de^2S}{\de y_i\de y_j}\right)_{i,j=1}^{l}$ is non-degenerate and $v_a = f_a^b\frac{\de}{\de z^b}$.

We define a map $\phi\colon U \to \g$ by $\phi(y,z) = z^bT_b$. It follwos that $\phi$ defines a local section $s$ of the quotient by $s(x) = [x] \cap \phi^{-1}({0})$. Such a section is known as a (local) \emph{gauge-fixing}. $\phi$ is called the \emph{gauge-fixing function}.
\begin{ass}
There is a global function $\phi\colon X \to \g$ (a global gauge-fixing function) such that $\phi^{-1}(0)$ intersects every orbit of $G$ exactly $N$ times.
\end{ass}
Very often it will not be possible at all to find functions $\phi$ with  $N=1$: Consider e.g. the example of the circle acting on the cylinder $I \times S^1$ by rotation. Any global gauge-fixing will intersect every orbit at least twice (exercise). \\
Given such a global gauge-fixing, we can rewrite the integral over the quotient as an integral over a subset of $X$ employing a delta function: 
\begin{equation}
I(\hbar) = vol(G)\int_{X/G}e^{\frac{\ii}{\hbar}S}\mu = \frac{vol(G)}{N}\int_X\delta^{(l)}(\phi)p^*\tilde{\mu}e^{\frac{\ii}{\hbar}S}.
\end{equation}
Here $\delta^{(l)}(\phi) = \delta(\phi)d\phi^1 \wedge \cdots \wedge d\phi^k$.  Our next goal is to rewrite the integral in terms of the original measure $\mu$, i.e. find a function\footnote{The Radon-Nikodym Derivative of the corresponding measures.} $J$ such that $d\phi^1 \wedge\cdots \wedge d\phi^k p^*\tilde{\mu} = J \mu$. 
\begin{lem}
$J$ is given by $J(x) = \det FP(x)$ where $FP(x)$ is the Faddeev-Popov operator 
\begin{align*}
FP(x) \colon \g &\to \g\\
\xi \mapsto d_x\phi(\rho^\#_x\xi)
\end{align*}
\end{lem}
In the basis $T_a$, the Faddeev-Popov operator is given by $FP(x)^a_b = d\phi^a(v_b(x))$. 
\begin{proof}
Extend $d\phi^1,\ldots,d\phi^k$ to a basis of $T^*M$ by $\alpha_1\ldots,\alpha_l$ s.t $\alpha_j(v_k) = 0$ and $\mu = d\phi^1\wedge \cdots\wedge d\phi^k\wedge \alpha_1 \wedge \cdots \wedge \alpha_l$. Then, one simply computes 
\begin{align*}
\iota_{v_k}\ldots\iota_{v_1}\mu &=\sum_{\sigma\in S_k}(-1)^{\mathrm{sgn}(\sigma)} d\phi^1(v_{\sigma(1)})\cdots d\phi^k(v_{\sigma(k)})\wedge \alpha_1 \wedge \cdots \wedge \alpha_k \\
&= \det FP(x) \alpha_1 \wedge \cdots \wedge \alpha_k
\end{align*}
which proves the claim. 
\end{proof}
We can use this Lemma in order to rewrite the integral in terms of $\mu$:
\begin{cor}
\begin{equation}
I(\hbar) =\frac{vol(G)}{N}\int_X \mu \delta(\phi(x))\det FP(x) e^{\frac{\ii}{\hbar}S}\label{eq:FP2}
\end{equation}
\end{cor}
Our next goal is to rewrite this as the integral of an exponential so that we can again use the Feynman diagrammatic methods in the last chapter. 
We can rewrite the Delta function as a Fourier transform
\begin{equation}
\delta(\phi(x))= \frac{1}{2\pi\hbar}\int_{\lambda \in \g^*}e^{\frac{\ii}{\hbar}\langle \lambda,\phi(x)\rangle}d^m\lambda.\label{eq:FP3}
\end{equation}
What about the determinant? The main realization in the Faddeev-Popov formalism (sometimes called the Faddeev-Popov trick) is that it can be written as an integral of an exponential by introducing ``odd'' coordinates, which we shall briefly introduce in the next section.
\subsection{Berezin integrals}
Berezin integrals, introduced by Berezin \cite{Berezin1966}, (see also \cite{Cattaneo2011b}  for an introduction), are integrals over \emph{super vector spaces}. Those, in turn, are combinations of even and odd vector spaces - even vector spaces are just the usual ones, whereas odd vector spaces are defined through their algebra of functions. \\
\subsubsection{Odd vector spaces}
\begin{defn}
 Let $V$ be a vector space over $\R$, then we define\footnote{Purists might prefer to say that the category of odd vector spaces is by the definition the opposite category of the categors of free finite-dimensional anticommutative algebras.} the odd vector space $\Pi V$ by $\mathcal{O}(\Pi V) = \wedge^\bullet V^*$.
\end{defn}
The upshot is that the coordinates $\xi^1,\ldots,\xi^n$ on $V$ \emph{anticommute} 
\begin{equation}
\xi^i\xi^j = \xi^j\xi^i.
\end{equation}
These are known as \emph{Grassmann}, \emph{fermionic} or simply \emph{odd} coordinates. 
\subsubsection{Derivatives}
The algebra of functions on an odd vector space is a super algebra, i.e. a $\Z_2$-graded algebra. The grading is given by the degree in the $\xi$'s (modulo 2) and derivatives are super-derivations of this super-algebra, i.e. $\de/\de \xi_i$ is defined by 
\begin{align}
\frac{\de }{\de \xi^i} (fg) &= \frac{\de }{\de \xi^i}f g + (-1)^{|f|}\frac{\de }{\de \xi^i}g \\ 
\frac{\de }{\de \xi^i}\xi^j = \delta_i^j 
\end{align}
\subsubsection{Integrals}
To define integrals we first define them on the ``odd line'' $\pi \R$. 
Here the integral is completely defined by the two requirements 
\begin{align}
\int_{\Pi \R} D\theta \frac{\de g}{\de \theta} &= 0 \\
\int_{\Pi \R} D\theta \theta &= 1. \label{eq:defBerint2}
\end{align}
The algebra of functions on the odd line is $\mathcal{O}(\Pi \R) = \R + \R\theta$. The integral of a function $f = a + b \theta$ is then simply $$\int_{\Pi\R}D\theta f = b.$$
\textbf{Warning:} The symbol $D\theta$ is \emph{not} a one-form. In fact, if we define a new coordinate $\theta' = \lambda \theta$, then Equation \eqref{eq:defBerint2} implies 
$$1 = \int_{\Pi\R}D\theta' \theta' = \int_{\Pi \R} D\theta' \lambda \theta $$ and thus $D\theta' = \frac{1}{\lambda}D\theta$ (in contrast, one-forms transform as $d\theta' = \lambda d\theta)$!). 
To define the integral on $\Pi\R^n$, we iterate 
$$\int_{\Pi\R^n}D^n\theta := \int_{\Pi \R}\left(\int_{\Pi\R}\cdots\int_{\Pi\R}D\theta_1\right).$$
An easy exercise shows that this implies 
\begin{equation}\int_{\Pi\R^n}D^n\theta f = f^{top}\label{eq:defBerInt3}\end{equation}
where $f^{top}$ is the component of $f$ that lies in $ \wedge^n\R^n$. 
This is turn means that in the odd case, integral and derivative coincide: 
$$\int_{\Pi\R^n}D^n\theta f =\frac{\de }{\de \theta_n }\ldots \frac{\de }{\de \theta_1}f.$$
Equation \eqref{eq:defBerInt3}suggests how to define the integral in a abstract odd vector space $\Pi V$: Namely, by the choice of an identification $\wedge^{top}V^* \cong \R$. Such an identification is equivalent to an element $\mu \in \wedge V^*$ (called a \emph{Berezinian}) and we can define the integral 
\begin{equation}
\int_{\Pi V}\mu f = \langle \mu , f \rangle. 
\end{equation} 
In particular, the odd vector space $\Pi(V \oplus V^*)$ has a canonical Berezinian given by 
\begin{equation}
\mu^c = D\theta^nD\bar{\theta}_n\cdots  D\theta^1D\bar{\theta}_1 \label{eq:canonicalBer}
\end{equation} 
for any choice of coordinates $\theta^i$ with dual coordinates $\theta_i$ on $\Pi V^*$. 
The crucial property of the canonical Berezinian is that for an endomorphism $B \in \mathrm{End}(V) \cong V^* \otimes V \subset \mathcal{O}(\Pi V \oplus \Pi V^*)$ 
we have 
\begin{equation}
\det B = \int_{\Pi V \oplus \Pi V^*}\mu^c e^{B_i^j\theta^i\bar{\theta}_j}.\label{eq:oddexp}
\end{equation}
We have thus completed the task we set out for and found a representation of the determinant as the integral of an exponential!
\subsection{The Faddeev-Popov action}
Equations \eqref{eq:FP3} and \eqref{eq:oddexp} allow us to rewrite the integral \eqref{eq:FP2} as the integral of an exponential. To this end, we collect the Faddeev-Popov data we have constructed so far.
\begin{defn}[Faddeev-Popov data]
Let $X$ be a manifold with a free Lie group action $\rho\colon G \times X \to X$, and $S\colon X \to \R$ a $G$-invariant function, $\mu \in \Omega^{top}(X)$ a $G$-basic form\footnote{$G$-invariant and horizontal for the quotient.}. 
Suppose we are given a global gauge-fixing function $\phi\colon X \to \g$. Then we define 
\begin{itemize}
\item The \emph{Faddeev-Popov space of fields} 
\begin{equation}
X_{FP} = X \times \g^* \times \Pi(\g \oplus \g^*) \ni (x,\lambda,c,\bar{c}) \label{eq:FPfields}
\end{equation}
$\lambda$ is called the Lagrange multiplier and $(c,\bar{c})$ the \emph{ghost-antighost pair}.
\item For every $x\in X$ the Faddeev-Popov operator $FP(x) \colon \g \to \g$ by 
\begin{equation}
FP(x) = d \phi(x) \circ \rho^\#_x
\end{equation}
\item The \emph{Faddeev-Popov action} $S_{FP} \in \mathcal{O}(X_FP)$, defined by 
\begin{equation}
S_{FP}(x,\lambda,c,\bar{c}) = S(x) + \langle \lambda,\phi(x)\rangle + \langle \bar{c},FP(x) c\rangle \label{eq:FPaction}
\end{equation}
\item The \emph{Faddeev-Popov Berezinian} $\mu_{FP}$ by 
\begin{equation}
\mu_{FP} = \frac{vol(G)}{N (2\pi\ii)^m}\mu d^n \lambda \mu^c_{\Pi(\g\oplus\g^*)}
\end{equation}
where $\mu^c_{\Pi(\g \oplus \g^*)}$ is the canonical Berezinian on $\Pi(\g \oplus \g^*)$  introduced in \eqref{eq:canonicalBer}. 
\end{itemize}
\end{defn}
The following is one of the two main results of this section:
\begin{thm}
In the notation as above, we have 
\begin{equation}
\int_Xe^{\frac{\ii}{\hbar}S}\mu = \int_{X_{FP}}e^{\frac{\ii}{\hbar}S_{FP}}\mu_{FP}
\end{equation}
\end{thm}
\subsubsection{Critical Points}
The main point of the Faddeev-Popov action \eqref{eq:FPaction}
is that - in contrast to the action $S$ - its critical points are non-degenerate, at least under the assumption that the FP operator is non-degenerate\footnote{In the finite-dimensional case this is equivalent to $\phi^{-1}(0)$ intersecting  the $G$-orbits transversally. }
\begin{prop}
Assume that $FP(x)$ is non-degenerate for all $x$. Then the critical points of $S_{FP}(x)$ are given by 
\begin{equation}
\begin{cases}
c &= 0 \\
\bar{c} &= 0 \\
\phi(x) &= 0 \\
\lambda &= 0\\
d_xS  &= 0 
\end{cases}
\end{equation}
I.e. $Crit(S_{FP}) = \{Crit(S) \cap \phi^{-1}(0)\} \times \{(0,0,0)\} \subset X_{FP}$. 
\end{prop}
\begin{proof}
The first three equations follow (using non-degeneracy of $FP$) by taking derivatives with respect to $\bar{c},c$ and $\lambda$ respectively. For the last equation, notice that taking derivative with respect to $x$ yields 
$$d_xS + \langle\lambda,d_x\phi\rangle = 0.$$ 
Solutions of this equation are extrema of $S$ under the condition $\phi = 0$. By $G$-invariance of $S$, such conditional extrema are also global extrema and hence $dS =0$ at these points. It follows that also $\lambda = 0$. 
\end{proof}
\begin{prop}
Assume $FP(x)$ is non-degenerate for all $x\in X$. Then the critical points of $S_{FP}$ are non-degenerate. 
\end{prop}
\begin{proof}
To see this one computes the Hessian of $S_{FP}$ with respect to an adapted chart $x=(y,z)$ in $X$ at a critical point $x_0 \equiv ((y_0,0),0,0,0) \in X_{FP}$: 
\begin{equation}
H_{x_0}S_{FP} = \left(\begin{array}{c|c|c|c|c}
(\frac{\de^2S}{\de y^i\de y^j})_{j=1}^l & 0 & d_y\phi^T &  0 & 0 \\
0 & 0 & d_{z}\phi^* & 0  \\
d_y\phi & d_{z}\phi^* & 0 & 0 & 0 \\
0 & 0 & 0 & 0 & FP(x_0)^* \\
0 & 0 & 0 &  FP(x_0) &0 \\
\end{array}
\right)
\end{equation}
In adapted coordinates, we have $\phi(y,z) = z$, whence we conclude that $d_y\phi = 0$ and $d_z\phi = Id$. By assumption ,$(\frac{\de^2S}{\de y^i\de y^j})_{j=1}^l$ and $FP(x_0)$ are non-degenerate. Hence $H_{x_0}S_{FP}$ is non-degenerate. 
\end{proof}
\subsection{Stationary phase method}
Since we have replaced the degenerate function $S$ with the non-degenerate $S_{FP}$, our next aim is to generalize the methods developed in the last section to the case at hand. The main difference is that the determinant of the odd quadratic operator appears in the nominator (rather than the denominator) of the prefactor.
Let us first write the Hessian of $S_{FP}$ at $x_0$ and its inverse without adapted chart: 
\begin{equation}
H_{x_0}S_{FP} = \left(\begin{array}{c|c|c|c}
(\frac{\de^2S}{\de x^i\de x^j})_{j=1}^n & d_x\phi^T &  0 & 0 \\
d_x\phi & 0 & 0  & 0  \\
 0 & 0 & 0 & FP(x_0)^* \\
 0 & 0 &  FP(x_0) &0 \\
\end{array}
\right)
\end{equation}
\begin{equation}
(H_{x_0}S_{FP})^{-1} = \left(\begin{array}{c|c|c|c}
K & \gamma &  0 & 0 \\
\gamma & 0 & 0  & 0  \\
 0 & 0 & 0 & (FP(x_0)^{-1})^* \\
 0 & 0 &  (FP(x_0)^{-1}) &0 \\
\end{array}
\right)
\end{equation}
Here $\gamma = \rho_{x_0}^\# \circ FP(x_0)^{-1}$
We denote by  $Q_{x_0}$ the Hessian of $S$ restricted to $\phi^{-1}(0)$. 	
We can then express the asymptotic behaviour of $I(\hbar)$ as 
\begin{equation}
I(\hbar) \simeq_{\hbar \to 0}  vol(G)(2\pi\hbar)^{(n-k)/2}\sum_{\text{crit. $G$-orbits } [x_0]} e^{\frac{\ii}{\hbar}S(x_0)}\frac{e^{\frac{\ii}{\hbar}\mathrm{sign}Q_{x_0}}}{|\det Q_{x_0}|^{1/2}}\det FP(x_0) \sum_{\Gamma}\frac{(-\ii\hbar)^{1-\chi(\Gamma)}}{|\mathrm{Aut}(\Gamma)|}F(\Gamma)
\end{equation}
Here the sum goes still over Feynman graphs, but the graphs can have different generators (in the spirit of subsection \ref{sec:FeynmanGraphs}), which we list below. The graphs are generated by the following half-edges:  

\begin{table}[!h]
\centering
\begin{tabular}{c|c|c}
coordinate & half-edge & name  \\
\hline  \\ 
$x^i$ & \begin{tikzpicture} \node[vertex] (a) at (0,0) {};
\node at (1,0) {$i$}
edge (a);
 \end{tikzpicture} & ``field(s)''
\\
$\lambda_j$ & \begin{tikzpicture} \node[vertex] (a) at (0,0) {};
\node at (1,0) {$j$}
edge[dotted] (a);
 \end{tikzpicture} & Lagrange multiplier\footnote{In physics literature, one often uses another version of the Faddeev-Popov formalism where $\lambda$ is not needed, if it is introduced, it is sometimes called ``Nakantishi-Lautrup field'' (e.g. \cite{Weinberg2005}).}\\
$c^j$ & \begin{tikzpicture} \node[vertex] (a) at (0,0) {};
\node at (1,0) {$j$}
edge[ghost] (a);
 \end{tikzpicture}  & ghost \\
$\bar{c}_j$ & \begin{tikzpicture} \node[vertex] (a) at (0,0) {};
\node at (1,0) {$i$}
edge[antighost] (a);
 \end{tikzpicture} & antighost \\
\end{tabular}
\caption{Half-edges in FP Feynman diagrams, $i$ runs from $1$ to $n$, whereas $j$ runs from $1$ to $k$.}
\end{table}

Edges are perfect matchings on the space of half-egdes. Note  that that the Feynman rules do not associate edges (as opposed to vertices, see below) with individual terms in the action, but rather, edge types that are not listed below evaluate to 0 (the corresponding block in the inverse of the Hessian vanishes). 

\begin{table}[!h]
\centering
\begin{tabular}{c|c|c}
edge & operator & name \\
\hline \\
\begin{tikzpicture} \node[] (a) at (-1,0) {$i$};
\node at (1,0) {$j$}
edge (a); \end{tikzpicture} & $K^{ij}$ & propagator in gauge $\phi = 0 $ \\
\begin{tikzpicture} \node[coordinate] (a) at (0,0) {};
\node at (1,0) {$j$}
edge[dotted] (a); 
\node at (-1,0) {$i$} edge (a); 
\end{tikzpicture} & $ \beta_i^j $ & `` $x$-$\lambda$ propagator'' \\
\begin{tikzpicture} \node[coordinate] (a) at (0,0) {};
\node at (1,0) {$j$}
edge[antighost] (a); 
\node at (-1,0) {$i$} edge[ghost] (a); 
\end{tikzpicture}& $FP^i_j$ & ghost propagator\\
\end{tabular}
\caption{Edges in FP Feynman diagrams and their corresponding operators}
\end{table}

Vertices are given by a partition of the set of half-edges. Under the Feynman rules, like in the case before, they are given by third and higher degree derivatives of the action. Every term in the action yields a corresponding vertex. Since the action is linear in $\lambda, c$, and $\bar{c}$, there is at most one of the corresponding half-edges at every vertex\footnote{Again, although in principle there are of course graphs with e.g. 5 $\lambda$ half-edges at one vertex, they evaluate to 0 under the Feynman rules, since the corresponding derivative of the action vanishes. Usually one does not count them among the Feynman graphs appearing in the theory.}. See Table \ref{tab:vertices}. 

\begin{table}[!h]
\centering
\begin{tabular}{m{3cm}|m{6cm}|m{3cm}}
vertex & tensor & name \\
\hline \\
\begin{tikzpicture}
\node[vertex] (o) at (0,0) {};
\node[coordinate, label=below:{$i_1$}] at (30:1) {$i_1$}
edge[] (o);
\node[coordinate, label=above:{$i_2$}] at (60:1) {$i_2$}
edge[] (o);
\node[coordinate, label=below:{$i_s$}] at (145:1) {$i_s$}
edge[] (o);
\draw[dotted] (90:0.5) arc (90:130:0.5);
\end{tikzpicture}
&  \hfil$\restr{\frac{\de^s}{\de x^{i_1} \ldots \de x^{i_s}}S(x)}{x=x_0}, s \geq 3$  & $x$  $s$-vertex\footnote{In gauge theories the field (here represented by $x$) usually has a particular name (``photon'' in QED, ``gluon'' in QCD, etc.) and the vertex is then called after the field. } \\
\begin{tikzpicture} 
\node[vertex] (o) at (0,0) {};
\node[coordinate, label=below:{$i_1$}] at (30:1) {$i_1$}
edge[] (o);
\node[coordinate, label=above:{$i_2$}] at (60:1) {$i_2$}
edge[] (o);
\node[coordinate, label=below:{$i_l$}] at (145:1) {$i_s$}
edge[] (o);
\draw[dotted] (90:0.5) arc (90:130:0.5);
\node[coordinate, label=below:{$j$}] (j2) at (-160:1) {};
\draw[dotted] (o) -- (j2);
\end{tikzpicture}
&  \hfil$\restr{\frac{\de^s}{\de x^{i_1} \ldots \de x^{i_l}}\phi^j(x)}{x=x_0}, l \geq 2$ & $\lambda$ $l$-vertex \\
\begin{tikzpicture}
\node[vertex] (o) at (0,0) {};
\node[coordinate, label=below:{$i_1$}] at (30:1) {$i_1$}
edge[] (o);
\node[coordinate, label=above:{$i_2$}] at (60:1) {$i_2$}
edge[] (o);
\node[coordinate, label=below:{$i_m$}] at (145:1) {$i_m$}
edge[] (o);
\draw[dotted] (90:0.5) arc (90:130:0.5);
\node[coordinate, label=below:{$j_1$}] (j2) at (-160:1) {};
\node[coordinate, label=below:{$j_2$}] (j3) at (-20:1) {};
\draw[antighost] (o) -- (j2);
\draw[ghost] (o) -- (j3);
\end{tikzpicture}
& \hfil $\restr{\frac{\de^s}{\de x^{i_1} \ldots \de x^{i_m}}FP^{j_1}_{j_2}(x)}{x=x_0}, m \geq 1	$  & ghost $m$-vertex
\end{tabular}
\caption{Vertices in FP Feynman diagrams and their corresponding tensors} \label{tab:vertices}
\end{table}
Mathematically, the graphs are modeled by tuples $(H=H_e \cup H_o,V,\m_e,\m_o)$  where $H_e$ denotes the set of even and $H_o$ the set of all half-edges. $m_e$, $m_o$ are perfect matchings on $H_e, H_o$ respectively. The additional structure in the Feynman rules above (no $\lambda - \lambda$ edges, only directed $\bar{c} - c$ edges) comes from the vanishing of the corresponding block in the Hessian, i.e. $F(e) = 0$ on these unwanted edges, $F(v)$  is automatically 0 on vertices that were not listed above because the corresponding derivatives of $S_{FP}$ vanish.  
\begin{rem}
There are two special cases that simplify the diagrams in question. In examples, the gauge-fixing function can often chosen to be linear in the fields. In that case, the $\lambda$ vertices with dotted edges vanish (they correspond to at least two derivatives of the gauge-fixing function). This is the case in the covariant gauges usually applied in Yang-Mills theories, discussed in the example below. \\
Another special case occurs if the FP operator is locally constant in $x$. In that case the weight of the ghost vertex vanishes. This happens in abelian Yang-Mills theories (such as QED).
\end{rem}
\subsection{Example: Yang-Mills}
Yang-Mills theory can be used to describe the dynamics of gluons, the particles that transmit the electro-magnetic and weak and strong forces. 
The space of fields is the space of connections on a trivial principal bundle $M \times G$, where $M$ is a Riemannian manifold and $G \subset GL(n)$ and connected and compact Lie group\footnote{In particle physics, $G = U(1)$ describes the electromagnetic force, $G=SU(2)$ the weak force and $G=SU(3)$ the strong force.} 
\begin{equation}
F_M = Conn(M\times G ) \cong \Omega^1(M,g)
\end{equation}
with action functional 
\begin{equation}
S_M = \int_M \mathrm{tr}\frac12 F_A \wedge *F_A.
\end{equation}
Here $*$ denotes the Hodge star and we remind the reader that $F_A$ is the curvature of the connection $A$ given by $F_A = dA + \frac12[A,A].$ 
The gauge group $\mathcal{G} =C^\infty(M,G) $ acts on $F_M$ via 
\begin{equation}
g\cdot A = \rho(g)A = g^{-1}Ag + g^{-1}dg
\end{equation}
and the infinitesimal action of the Lie algebra $Lie(\mathcal{G}) = \Omega^0(M,\g)$ is given by 
\begin{equation}
\rho^\#_A(\alpha) = d_A\alpha \in T_AF_M. 
\end{equation}
A gauge that is often chosen is the Lorenz gauge 
\begin{equation}
\phi(A) = d^*A = 0
\end{equation} 
so that the gauge fixing function is $\phi = d^*\colon \Omega^1(M,\g) \to \Omega^0(M,\g)$ and the Faddeev-Popov operator $FP(A) d\phi \circ \rho_A^\# \colon \Omega^0(M,\g)\to \Omega(M,\g)$ is given by 
$$FP(A)\alpha = d^*d_A\alpha.$$ 
In the infinite-dimensional setting, one has to chose an appropriate model for $\Omega^0(M,\g)^*$. A possible choice is $\Omega^d(M,\g)$, where the action of $\lambda \in \Omega^d(M,\g)$ on $\alpha\in\Omega^0(M,\g)$ is given by 
\begin{equation}
\langle \lambda,\alpha\rangle = \int_M\mathrm{tr}\lambda \wedge \alpha. 
\end{equation}
The space of Faddeev-Popov fields is then 
\begin{equation}
F_{FP} = F_M \times \Omega^d(M,\g) \times \Pi\Omega^0(M,\g) \times \Pi\Omega^d(M,\g) \ni (A,\lambda,c,\bar{c})
\end{equation}
and the Faddeev-Popov action is 
\begin{equation}
S_{FP}[A,\lambda,c,\bar{c}] = \int_M \mathrm{tr} \left(\frac12 F_A \wedge * F_A+ \lambda \wedge d^*A + \bar{c} \wedge d^*d_Ac\right). 
\end{equation}
Perturbative evaluation of the path integral 
$$Z^{pert} = \int^{pert}_{F_{FP}}e^{\frac{\ii}{\hbar}S_{FP}}  \simeq \sum_{\Gamma}\frac{(-\ii\hbar)^{-\chi(\Gamma)}}{|\mathrm{Aut(\Gamma)}}F(\Gamma)$$
leads to a series of terms labeled by Feynman graphs. As discussed above, there are in principle three types of edges in these diagrams but since the gauge-fixing is linear, the ``dotted'' edges (or ``$A$-$\lambda$'' propagators) do not appear. See Table \ref{tab:YMedges}The propagator in Yang-Mills is usually denoted with a curly line, and called \emph{gluon} propagator (after the particles corresponding to excitations of the gauge field). 
\begin{table}[!h]
\centering
\begin{tabular}{c|c|c}
edge & operator & name \\
\hline \\
\begin{tikzpicture} \node[] (a) at (-1,0) {$i$};
\node at (1,0) {$j$}
edge[gluon] (a); \end{tikzpicture} & $K^{ij}$ & gluon propagator in gauge $\phi = 0 $ \\
\begin{tikzpicture} \node[coordinate] (a) at (0,0) {};
\node at (1,0) {$j$}
edge[antighost] (a); 
\node at (-1,0) {$i$} edge[ghost] (a); 
\end{tikzpicture}& $FP^i_j$ & ghost propagator\\
\end{tabular}
\caption{Edges in Yang-Mills diagrams in Lorenz gauge and their corresponding operators}
\label{tab:YMedges}
\end{table}

The vertices correspond to the cubic and higher terms in the action functional, see Table \ref{tab:YMvertices}. 
\begin{table}[!h]
\centering
\begin{tabular}{c|c|c}
vertex & term in action functional & name \\
\hline
\begin{tikzpicture} 
\node[vertex] (o) at (0,0) {};
\node[coordinate] at (30:1) {$i_1$}
edge[gluon] (o);
\node[coordinate] at (150:1) {$i_2$}
edge[gluon] (o);
\node[coordinate] at (270:1) {$i_s$}
edge[gluon] (o);
\end{tikzpicture}
&  \hfil$\frac12\int_M \mathrm{tr}[A,A]*dA$ & 3-gluon vertex \\
\begin{tikzpicture} 
\node[vertex] (o) at (0,0) {};
\node[coordinate] at (30:1) {$i_1$}
edge[gluon] (o);
\node[coordinate] at (120:1) {$i_2$}
edge[gluon] (o);
\node[coordinate] at (210:1) {$i_s$}
edge[gluon] (o);
\node[coordinate] at (300:1) {$i_s$}
edge[gluon] (o);
\end{tikzpicture}
&  \hfil$\frac18\int_M \mathrm{tr}[A,A]*[A,A]$ & 4-gluon vertex \\
\begin{tikzpicture}
\node[vertex] (o) at (0,0) {};
\node[coordinate] at (90:1) {$i_1$}
edge[gluon] (o);
\node[coordinate, label=below:{$j_1$}] (j2) at (-160:1) {};
\node[coordinate, label=below:{$j_2$}] (j3) at (-20:1) {};
\draw[antighost] (o) -- (j2);
\draw[ghost] (o) -- (j3);
\end{tikzpicture}
& \hfil $\int_M \bar{c} \wedge d^*[A,c]$  & ghost vertex
\end{tabular}
\caption{Vertices in FP Feynman diagrams and their corresponding tensors} \label{tab:YMvertices}
\end{table}
A major problem is that closed gluon loops will result in divergent integrals. For Yang-Mills theory this problem has been solved by the process of renormalization, see e.g. \cite{Hollands2008}.

\subsection{Example: Gravity}
In principle, the perturbative Faddeev-Popov formalism can be applied to gravity just in the same way. The space of fields is the space of Lorentzian metrics on $M$: 
\begin{equation}
F_M = Met^{1,3}(M)
\end{equation}
and the action functional is 
\begin{equation}
S_M[g] = \int_M R_g dvol_g,
\end{equation}
where $R_g$ denotes the Ricci scalar. The gauge group is the group of Diffeomorphisms of $M$, and its Lie algebra is the Lie algebra of vector fields: $\mathcal{G}=Diff(M), Lie(\mathcal{G}) = \mathfrak{X}(M)$. The infinitesimal action of $\mathfrak{X}(M)$ on $F_M$ is 
\begin{equation}
\rho^\#_g(\xi) = L_\xi(g) \in T_gMet^{1,3}(M) \simeq S^2T^*M.
\end{equation}
A gauge-fixing condition often used is the \emph{de Donder} or \emph{harmonic coordinates} gauge-fixing given by 
$$\phi(g)= \Gamma_{\mu\nu}^\rho g^{\mu\nu} = 0, $$
where $\Gamma$ denotes the Christoffel symbols, and we have 
$$\langle \bar{c},FP(g)c\rangle = g^{\mu\nu}g^{\rho\sigma}\bar{c}_\rho\de_\mu\de_\nu c_\sigma.$$
Due to the highly non-linear nature of all three terms in the Faddeev-Popov action, all possible graphs discussed above can appear (with any valence of ``graviton'' legs). For this theory, no satisfactory renormalization procedure has yet been found\footnote{Which does not mean that it does not exist. See, for instance, the discussion in \cite{Kreimer2008}. }.  See also the discussion in \cite{Prinz2018}.
\section{BRST symmetry}

As already remarked in Henneaux and Teitelboim \cite{Henneaux1994}, a curiosity in gauge-fixing is that to ``remove'' unwanted degrees of freedom, one \emph{adds} extra variables. An explanation of this phenomenon is the mathematical concept of resolutions. The BRST symmetry (discovered independently by Becchi, Rouet and Stora in \cite{Becchi1975}, \cite{Becchi1976} and Tyutin in \cite{Tyutin1976}) is the starting point of the connection between gauge theory and homotopical (or derived) algebra.  
\subsection{BRST operator}
Remember that 
$$X_{FP} = X \times \g^* \times \Pi(\g \oplus \g^*) \ni (x,\lambda,c,\bar{c})$$
from where we have
$$ \mathcal{O}(X_{FP}) = C^{\infty}(X) \otimes \hat{S}\g \otimes \wedge^\bullet\g^* \otimes \wedge^\bullet\g.$$
Again let $T_a$  be a basis of $\g$, $f_{ab}^c$ the corresponding structure constants and $\rho^\#(T_a) = v^i_a\frac{\de }{\de x^i}$. 
$\mathcal{O}(X_{FP})$ is a super algebra (the $Z_2$-grading comes from the exterior algebras) and there is an odd derivation $Q$ on this superalgebra (put differently, an odd vector field on $X_{FP}$): 
\begin{defn} 
The \emph{BRST operator} or \emph{BRST symmetry} is the derivation $Q \colon \mathcal{O}(X_{FP}) \to \mathcal{O}(X_{FP})$ defined on generators by
\begin{align}
Qx^i  &= c^av^i_a(x)\\
Qc^c &= \frac12 f^c_{ab}c^ac^b \\
Q\bar{c}_a &= \lambda_a \\
Q\lambda_a &= 0
\end{align}
\end{defn}
The connection to homological algebra comes from the following Lemma: 
\begin{lem}
The operator $Q\colon \mathcal{O}(X_{FP}) \to \mathcal{O}(X_{FP})$ squares to 0.
\end{lem}
\begin{proof}
This follows from the Jacobi identity and the fact that $\rho^\#$ is a homomorphism of Lie algebras. 
\end{proof}
\subsection{BRST cohomology}
Thus $\mathcal{O}(X_{FP})$ is a ($\Z_2$-graded) complex when equipped with $Q$. The following observation is obvious from the definition. 
\begin{prop}
The complex $(\mathcal{O}(X_{FP}),Q)$ splits as a direct sum of two subcomplexes $(\mathcal{O}(X_{FP}),Q)=((\mathcal{O}(X_{min}),Q_{CE})\oplus (\mathcal{O}(X_{aux}),Q_{aux})$, where 
\begin{align}
X_{min} &= X \times \Pi\g \\
Q_{CE} &=  c^av^i_a(x)\frac{\de }{\de x^i} +  c^ac^b\frac12f^c_{ab}\frac{\de }{\de c^c} \\ 
\end{align} 
and 
\begin{align}
X_{aux} &= \g^* \oplus \Pi\g^* \\
Q_{aux} &= \lambda_a\frac{\de }{\de \bar{c}_a} 
\end{align}
\end{prop}
This allows us to compute the cohomology of the BRST operator: 
\begin{prop} The cohomology of the BRST operator is 
\begin{equation}
H^k(\mathcal{O}(X_{FP}),Q) = H^k(\mathcal{O}(X_{min}),Q_{CE}) = \begin{cases}  
C^{\infty}(X/G) & k=0 \\
0 & k = 1
\end{cases}
\end{equation}
\end{prop}
\begin{proof}
The first equality follows from the fact that the cohomology $Q_{aux}$ is trivial. The second equality follows from the fact that $H^k(\mathcal{O}(X_{min}),Q_{CE})$ is the Chevalley-Eilenberg complex for the Lie algebra action of $\g$ on $C^{\infty}(X)$.
\end{proof}
Thus both functions on $X_{FP}$ and $X_{min}$ are resolutions of the quotient $X/G$. Since the ``physical'' observables are elements of the quotient we interpret the even cohomology of the BRST operator as the physically relevant information.
\subsection{Gauge fixing} 
It follows from gauge invariance of $S$ that $S$ is $Q$-closed.  
From the definition of $Q$, the following Lemma is immediate: 
\begin{lem}
We have 
\begin{equation}
S_{FP} = S + Q\psi,
\end{equation}
where 
\begin{equation}
\psi(x) = \langle\bar{c} \phi(x)\rangle
\end{equation}
is called the \emph{gauge fixing fermion}.
\end{lem}
Thus the gauge fixing does not change the cohomology class of the physical action. 
\begin{rem}
It should be noted that if we are only interested in a resolution of the gauge symmetries, $X_{min}$ suffices. However, to implement the gauge-fixing condition, we have to add the acyclic complex $X_{aux}$. 
\end{rem}
One can promote the ``BRST logic'' (try to resolve the action of the gauge symmetries on $X$) to a gauge-fixing formalism, and there is also a quantum version of it. However, we will head straight to the BV formalism, and refer to the literature (e.g. \cite{Henneaux1994}) for a deeper discussion of BRST concepts.  \section{Batalin-Vilkovisky formalism and effective actions}
The Batalin-Vilkovisky formalism, introduced by Batalin and Vilkovisky in the early '80s (\cite{Batalin1981},\cite{Batalin1983}) is the most powerful gauge-fixing formalism on the market. Its main advantages over the BRST formalism are 
\begin{itemize}
\item it can deal with situations where the distribution of symmetries is not integrable (this happens mostly in string theory) 
\item it has a natural framework for renormalization and effective actions.
\end{itemize}
The very rough idea is that - in a situation where one cannot define the perturbative integral even on $X_{FP}$ - one adds, for all fields $\phi \in X_{FP}$ an ``antifield'' $\phi^+$ and the perturbs the region of integration away from the zero section $X_{FP}  = \{\phi^+ = 0\}$ to a new region $\mathcal{L}$: 

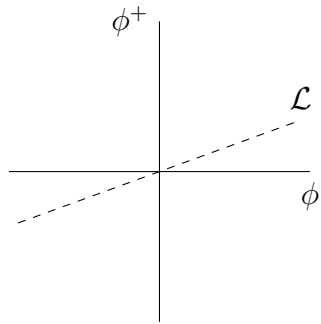
\begin{figure}[!h]
\centering
\begin{tikzpicture} 
\node[coordinate, label=below:{$\phi$}] at (2,0) {};
\draw (-2,0) -- (2,0); 
\node[coordinate, label=left:{$\phi^+$}] at (0,2) {};
\draw (0,-2) -- (0,2); 
\node[coordinate, label=above:{$\mathcal{L}$}] at (20:2) {};
\draw[dashed] (-160:2) -- (20:2);
\end{tikzpicture}
\end{figure} 
In these lecture notes we discuss a simple version of the BV formalism that uses only $\Z_2$-grading and super vector spaces. This avoids the technical complications of graded manifolds. There are plenty of resources on BV formalism, the author has found the introductory texts \cite{Fiorenza2003} and \cite{Cattaneo2019} and references therein helpful. 
\subsection{Odd symplectic vector spaces}
We begin with the discussion of odd symplectic vector spaces.
\begin{defn}
An \emph{odd symplectic vector space} (or \emph{BV vector space}) $(V,\omega)$ is a super vector space  $V = V_0 \oplus \Pi V_1$ with a non-degenerate symmetric bilinear pairing $\omega\colon V_0 \times V_1 \to \R$.  
\end{defn}
Let us explain how one can interpret the pairing $\omega$ as an odd symplectic form on $V$.
One can extend $\omega$ to a bilinear form on $V$ by defining it to vanish on $V_0$ and $V_1$. The symmetry of $\omega$ then implies that it is antisymmetric in the $\Z_2$-graded sense. It is odd since it pairs only even and odd vectors. On a vector space, any antisymmetric bilinear form is symplectic ($d\omega = 0$ since $\omega$ is constant). 
\begin{defn}
A subspace $\mathcal{{L}} \subset V$ of an odd symplectic vector space $(V,\omega)$ is called Lagrangian if $\restr{\omega}{\mathcal{L}} = 0 $ and $\dim \mathcal{L} = \frac12\dim V$. \end{defn}
For example, $V_0$ and $\Pi V_1$ are Lagrangian subspace of $V$ (non-degeneracy of $\omega$ forces $\dim V_0= \dim V_1 = \frac12 \dim V$). 
In coordinates $y^i$ on $V_0$ and $\theta_i$ on $V_1$ the odd symplectic form reads 
$$\omega = dy^i \wedge d\theta^i.$$
We can also choose more general Darboux coordinates $z= (x^i,x_i^+) $ where the $x^i$ can be even or odd, as long as $x_i+$ have opposite parity form $x^i$. 

\subsection{Odd Poisson bracket and BV algebra}
The odd symplectic form induces an odd Poisson bracket $(\cdot,\cdot) \colon \mathcal{O}(V) \times \mathcal{O}(V) \to \mathcal{O}(V)$. In Darboux coordinates $z^\alpha = (x^i,x_i^+)$ it reads 
\begin{equation}
(f,g) = \frac{\de_r f}{\de x_i^+}\frac{\de_l g}{\de x^i} - (-1)^{(|f|+1)(|g|+1)}\frac{\de_r f}{\de x_i^+}\frac{\de_lg}{\de x^i}
\end{equation}
where the \emph{left and right derivatives} are defined by 
\begin{align}
\frac{\de_lf}{\de z^\alpha} &= \frac{\de f}{\de z^\alpha} \\
\frac{\de_rg}{\de z^\alpha} &= (-1)^{|z^\alpha||g|}\frac{\de g}{\de z^\alpha}. 
\end{align}
This bracket is called \emph{BV bracket}, sometimes also the \emph{Buttin bracket} or \emph{antibracket}. 
There is an odd operator $\Delta \colon \mathcal{O}(V) \to \mathcal{O}(V)$ called the \emph{BV Laplacian} given\footnote{The story of the BV Laplacian is a lot subtler and more interesting than we make it appear here. See the papers \cite{Khudaverdian2004},\cite{Severa2006} for a more detailed discussion of this aspect.}  in these Darboux coordinates by 
\begin{equation}
\Delta = \sum_i \frac{\de }{\de x_i^+}\frac{\de }{\de x^i}
\end{equation}
\begin{lem} The BV operator squares to zero: $\Delta^2 = 0$. 
\end{lem}
The triple $(\mathcal{O}(V), \Delta, (\cdot,\cdot))$ satisfies the axioms of what is known as a \emph{BV algebra}. 
\subsection{BV integral} 
We now turn to the main idea of the BV formalism. Suppose that we equip $V$ with the Berezinian $ \mu_{BV} = d^nyD^n\theta$.
Then its square root $\mu_{BV}^{1/2}$ is the density that transforms with the square root of the superdeterminant of the transition functions. The restriction of $\mu_{BV}^{1/2}$ is a \emph{density} and one can define the \emph{BV integral}  
\begin{equation}
\int_{\mathcal{L}} f \mu_{BV}^{1/2}.
\end{equation} of an element $f \in \mathcal{O}(V)$. 
The importance of this integral stems from the following theorem. 
\begin{thm}[Batalin-Vilkovisky \cite{Batalin1981,Batalin1983}, Schwarz \cite{Schwarz1993}]\label{thm:BVint}
If $\Delta f = 0$, then the BV integral 
\begin{equation}
\int_{\mathcal{L}}f\mu_{BV}^{1/2}
\end{equation}
depends only on $[f] \in H_{\Delta}^\bullet(\mathcal{O}(V))$ and the homology class of $\mathcal{L}$.
\end{thm}
In particular, the integral is invariant under small perturbations of the Lagrangian. 
\subsection{Master equations}
The functions $f$ we are interested in are of the form $f = e^{\frac{\ii}{\hbar}S}$. The compatibility between the BV Laplacian and the bracket implies 
\begin{equation}
\Delta e^{\frac{\ii}{\hbar}S} = \left(\frac{1}{2}(S,S) - \ii\hbar\Delta S\right)e^{\frac{\ii}{\hbar}S}
\end{equation}
We are thus interested in action functionals $S$ that satisfy 
\begin{equation}
\frac{1}{2}(S,S) - \ii\hbar\Delta S = 0\label{eq:QME}
\end{equation}
Equation \eqref{eq:QME} is known as the \emph{Quantum Master Equation} (QME). Its reduction modulo $\hbar$ is 
\begin{equation}
(S,S) = 0,\label{eq:CME}
\end{equation}
the \emph{classical Master Equation} (CME). We are thus facing the following two problems: 
\begin{prob}[Classical BV extension]
Given an action functional $S\colon X \to \R$, find a supervector space $X_{BV}$ and $S_{BV}\in\mathcal{O}(X_{BV})$ such that 
\begin{itemize}
\item $X \subset (X_{BV})_0$ and $\restr{S_{BV}}{X} = S$.
\item The CME holds: $(S_{BV},S_{BV})  =0$. 
\end{itemize}
\end{prob}
\begin{prob}[Quantum BV extension]
Given $X_{BV}$ and $S_{BV}$ as above, find $\tilde{S} \in \mathcal{O}(X_{BV})[[\hbar]]$ such that the QME holds.
\end{prob}
A very simple case is when $\Delta S = 0$. The QME is then automatically satisfied. 
\subsection{Gauge-fixing}
The gauge-fixing in the BV formalism is the choice of the Lagrangian submanifold $\mathcal{L}$. If we can find a quantum BV action functional $\tilde{S}_{BV}$ satisfying $\Delta e^{\frac{\ii}{\hbar}S} = 0$, then the integral 
\begin{equation}
\int_{\mathcal{L}}e^{\frac{\ii}{\hbar}S}
\end{equation}
is independent under deformations of $\mathcal{L}$. If we find a Lagrangian $\mathcal{L}$ such that the quadratic part of the action $\tilde{S}_{BV}$ has \emph{non-degenerate} critical points when restricted to $\mathcal{L}$, then the BV integral can be defined perturbatively. 
\subsection{Faddeev-Popov and BRST solutions to Classical Master Equation}
If one has FP data, then one can construct a solution to the classical master equation as follows. 
Suppose we are given $X_{FP},S_{FP}$ as above (for simplicity we assume that $X$ is a vector space). Then, define\footnote{$z^\alpha$ is used with different meaning here from above.}
\begin{equation}
X_{BV} = \Pi T^*X_{FP} = X_{FP} \oplus \Pi X^*_{FP} \ni (z^\alpha,z_{\alpha}^+)
\end{equation} 
and 
\begin{equation}
S_{BV} = S + Q(z^\alpha)z_{\alpha}^+ = S + c^av_a^i(x)x_i^+ + \frac{1}{2}f_{ab}^cc_c^+ + \lambda_a\bar{c}^{a,+}.
\end{equation}
Given a gauge-fixing function $\phi(x)$ the gauge-fixing Lagrangian can be defined by  
\begin{equation}
\mathcal{L} = \left\lbrace z^+ = \frac{\de \psi}{\de z^\alpha}\right\rbrace
\end{equation} where $\psi = \left\langle \bar{c},\phi(x) \right\rangle$ is the BV version of gauge-fixing fermion.  One can then show that $\restr{S_{BV}}{\mathcal{L}} = S_{FP}$. \\
A similar construction works for the ``minimal'' space of BRST fields $X_{min} = X \oplus \Pi \g$, one can define \begin{equation}
X_{BV} = X_{min} \oplus \Pi X_{min}^*\label{eq:minBVfields_BRST}\end{equation} 
and 
\begin{equation}S_{BV} = S + Q_{CE}(z^\alpha)z_\alpha^+.\label{eq:minBVaction_BRST}\end{equation} This satisfies the classical master equation.  
\begin{rem}
In both cases, $\Delta S = 0$ if the corresponding $Q_{BV} = (S,\cdot)$ is divergence-free for $\mu_{BV}$. In the infinite-dimensional case the measure does not make sense and one has to regularize the Laplacian to obtain the QME. 
\end{rem}
\subsection{BV pushforward and effective actions}
Suppose we have a BV vector space $(X_{BV},\omega_X)$ with a splitting into BV vector spaces $X_{BV} = Y_{BV} \times Y'_{BV}$ such that the symplectic form splits: $\omega_{X} = \omega_{Y} + \omega_{Y'}$. Then  also the BV Laplacians split: $\Delta_X = \Delta_Y + \Delta_{Y'}$. We also assume that $\mu_{BV}^X = \mu_{BV}^Y\mu_{BV}^{Y'}$. 
\begin{defn}
The \emph{effective action $S_{eff} \in \mathcal{O}(Y)$ on $Y$ induced by $S_{BV} \in \mathcal{O}(X_{BV})$} is defined by 
\begin{equation}
e^{\frac{\ii}{\hbar}S_{eff}(y)} = \int_{y' \in \mathcal{L}}e^{\frac{\ii}{\hbar}S(y+y')}(\mu^{Y'}_{BV})^{1/2}
\end{equation}
\end{defn}
The fiber integral on the right hand side is called a \emph{BV pushforward}. The following theorem is the generalization of Theorem \ref{thm:BVint} to BV pushforwards: 
\begin{thm}[\cite{Cattaneo2008},\cite{Cattaneo2017}]\label{thm:BVpushforward}
\begin{enumerate}
\item If $S\in \mathcal{O}(X_{BV})$ satisfies the Quantum Master Equation, then so does $S_{eff} \in \mathcal{O}(Y_{BV})$, i.e. 
\begin{equation}
\Delta_Y e^{\frac{\ii}{\hbar}S_{eff}(y)} = 0. 
\end{equation}
\item Let $\calL_t$ be a smooth family of Lagrangians with $\calL_0 = \calL$. Denote $S_{eff,t}$ the effective action defined by the gauge-fixing Lagrangian $\calL_t$. Then there is $X \in \mathcal{O}(Y)$ s.t.
\begin{equation}
\restr{\frac{d}{dt}S_{eff,t}}{t=0} = \Delta_Y X.
\end{equation}
\end{enumerate}
\end{thm}
Thus, from a solution of the QME on $X$ we can define a family of effective actions on BV subspaces by BV pushforwards: 

\begin{equation}
X\leadsto Y^{(1)} \leadsto Y^{(2)} \leadsto \ldots \leadsto Y^{(N)} \leadsto \ldots
\end{equation}
The corresponding effective actions all satisfy the corresponding Quantum Master Equation. In particular, this is a version of ``Wilson renormalization\footnote{See the work of Anselmi \cite{Anselmi1994} on BV formalism and renormalization, or \cite{Costello2011} for a modern treatment of the subject.}''. In our setting (for topological theories) we are interested in a BV pushforward from the space of fields to the space of zero modes. These prevent a theory from being handled in the FP or BRST formalisms and is one of the main reason for resorting to BV formalism. The other reason is that the BV formalism is better adapted to manifolds with boundary and cutting and gluing, as was recently shown in the work of Cattaneo, Mnev and Reshetikhin  in \cite{Cattaneo2014}, \cite{Cattaneo2017}. Thus one can hope to compute the BV effective action by cutting the spacetime into simple pieces and ``gluing'' them back together. 

\chapter{Perturbative Quantization of Chern-Simons Theory}
In this final chapter we turn to the perturbative quantization of Chern-Simons theory. To construct the perturbative partition function we go through several steps. We start off with the Batalin-Vilkovisky formulation of Chern-Simons theory. We then fix a classical background - i.e. a flat connection $A_0$ - and gauge fix in the neighbourhood of that classical background. Since there can be zero modes around, we then formulate the partition function as a formal BV pushforward. Finally, we show that while this partition function 
\section{BV Chern-Simons theory}
\subsection{Setup}
Let $M$ be a compact oriented 3-manifold and $G$ a simple and simply connected Lie group with Lie algebra $\g$. Let $P\to M$ be a principle $G$-bundle. By Lemma \ref{lem:PBtriv}, there exists a section $s\colon M \to P$. The space of fields is $F_M = Conn(P)$ which we identify with $\Omega^1(M,\g)$ using the section $s$.  Fixing an invariant symmetric bilinear form on $\g$ satisfying assumption \ref{ass:integralclass}, we define the Chern-Simons action functional 
$$S_{CS}[A] = \int_M \frac12\langle A,dA\rangle+\frac16\langle A,[A,A]\rangle $$
and we have that its exponential 
$$\exp\left(\frac{\ii}{\hbar}S_{CS}\right), \hbar = \frac{1}{2\pi k}$$ 
is independent of the choice of section $s$ and invariant under gauge transformations
$$A \mapsto A^g = g^{-1}Ag + g^{-1}dg, g\in C^{\infty}(M,G).
$$
\subsection{BRST operator}
The Lie algebra of the gauge group $C^{\infty}(M,G)$ is $C^\infty(M,\g)$ with infinitesimal action $\rho^\#_A(c) = d_Ac$.
The space of minimal BRST fields is 
$$F^{min}_{BRST} = \Omega^1(M,\g) \oplus \Pi\Omega^0(M,\g)\ni (A,c)$$
with BRST operator 
\begin{align}
QA &= d_Ac \\
Qc &= \frac{1}{2}[c,c]. \\
\end{align}
\subsection{BV formulation}
The minimal space of BV fields as in Equation \eqref{eq:minBVfields_BRST} is then 
\begin{equation}
\mathcal{F} = \Pi T^*F^{min}_{BRST} = \Pi\Omega^0(M,\g) \oplus \Pi\Omega^1(M,\g)\oplus \Pi\omega^2(M,\g) \oplus \Omega^3(M,\g)\ni(c,A,A^+,c^+)
\end{equation}
where we identify $\Omega^k(M,\g)$ with $(\Omega^{n-k}(M,\g))^*$ using the pairing 
\begin{equation}
\omega_{BV}(A,B) = \int_M\langle A,B\rangle
\end{equation}
which, extended to $\mathcal{F}$, becomes an odd symplectic form. Notice that we can write the BV space of fields as 
\begin{equation}
\calF = \Pi\Omega^\bullet(M,\g),
\end{equation}  where the $\Pi$ simply means that we are shifting the ``natural'' parity of the superspace $\Omega^\bullet(M,\g)$ induced by form degree by one. The BV action, as given in Equation \eqref{eq:minBVaction_BRST}, is 
\begin{equation}
S_{BV} = S_{CS} + Q(\phi^\alpha)\phi_{\alpha}^+ = S + \int_M \langle A^+,d_Ac\rangle + \frac{1}{2} \int_M \langle c^+,[c,c]\rangle
\end{equation}
The following is a crucial observation. If we define the ``superfield'' 
\begin{equation}
\mathcal{A} = c + A + A^+ +c^+ \in \Omega^\bullet(M,\g),
\end{equation}
the Chern-Simons action can be rewritten in the following intriguing fashion: 
\begin{equation}
S_{BV}[c,A,A^+,c^+] = S_{CS}[\mathcal{A}] = \int_M\frac{1}{2}\left\langle \calA,d\calA\right\rangle + \int_M \frac16\left\langle\calA,[\calA,\calA]\right\rangle.
\end{equation}
To see this, in terms of ``total parity'' (i.e. form parity plus ghost parity) all fields are odd. Hence both terms in the action are totally symmetric in all fields. Since only the terms of total form degree 3 contribute to the integral, we get 
\begin{align*}
\int_M\frac{1}{2}\left\langle \calA,d\calA\right\rangle + \int_M \frac16\left\langle\calA,[\calA,\calA]\right\rangle &= \int_M \frac{1}{2} \langle A,dA \rangle + \langle A^+,dc\rangle + \langle c,dA^+\rangle \\
&+ \int_M\frac{1}{6}\langle A,[A,A]\rangle + \frac{1}{2}\langle c^+, [c,c]\rangle + \langle A^+, [A,c]\rangle \\
&= S_{CS}[A] + \int_M \langle A^+, d_Ac\rangle + \int_M\langle c^+,[c,c]\rangle
\end{align*}
using that $d_Ac = dc + [A,c]$. 
\subsection{The Quantum Master Equation}\label{sec:QME1}
Let us now argue that the BV-extended Chern-Simons action formally satisfies the QME 
\begin{equation}
\Delta e^{\frac{\ii}{\hbar}S_BV} = 0. 
\end{equation}
To this end, we first show that the action satisfies the classical master equation $(S_{BV},S_{BV})=0$. The simplest way to see this is to use the superfield formalism and notice that 
$$Q = (S,\cdot) = F_{\calA}\frac{\delta}{\delta\calA}.$$ 
Now 
\begin{align*}Q(S) &= \int_M\langle \calA,dF_\calA\rangle + \frac{1}{2}\langle F_\calA,[\calA,\calA]\rangle   \\
&= \int_M \frac{1}{2}\langle \calA,d[\calA,\calA]\rangle + \frac12\langle d\calA,[\calA,\calA]\rangle + \frac{1}{2}\langle[\calA,\calA],[\calA,\calA]\rangle.
\end{align*}
The first two terms combine into a total derivative, and the last term vanishes by Jacobi and invariance of the pairing.  \\\\ For the BV Laplacian it is better to use field-antifield notation\footnote{The functional derivative is defined by the requirement $\frac{\delta}{\delta \phi^i(x)}\phi^j(y) = \delta_i^j\delta^{(3)}(x,y)$, where $\delta^{(3)}(x,y)$ is the integral kernel of the identity map on $\Omega^\bullet(M)$.}.
$$\Delta =\sum_i \int_M\frac{\delta}{\delta (A^+)^i}\frac{\delta}{\delta A^i} + \frac{\delta}{\delta (c^+)^i}\frac{\delta }{\delta c^i}$$
where we have expanded the fields in an orthonormal basis $T_i$. The only contribution to $\Delta(S_{BV})$ comes from the terms $\langle A^+,[A,c]\rangle$ and $\langle c^+,[c,c]\rangle$ and is proportional to 
$$\Delta(S) \varpropto \sum_i\int_{(x,y) \in M\times M}c^jf_{jii}(\delta^{(3)}(x,y))^2.$$ 
Since we do not know how to deal with the square of this delta form, we assume that $\sum_if_{jii} = 0$ for all $j$. This is condition is equivalent to \emph{unimodularity}: 
\begin{defn}\label{def:unimodular}
A Lie algebra $\g$ is called \emph{unimodular} if $\operatorname{tr}  ad_x = 0$ for all $x \in \g$. 
\end{defn}
We conclude that if $\g$ is unimodular, the Quantum Master Equation is satisfied at least formally, since
$$\Delta e^{\frac{\ii }{\hbar}{S}} = \left(\frac{1}{2}(S,S) -\ii\hbar\Delta(S)\right)e^{\frac{\ii}{\hbar}S}.$$
\section{Gauge fixing using background fields}
Instead of looking for a global gauge-fixing function, we fix a critical point of the Chern-Simons action functional and gauge fix the theory around it. This will turn out to be a bit simpler than trying to find a global gauge-fixing. 

\subsection{Fixing a background field}
As discussed in Section \ref{sec:crit_pts}, the critical points of the Chern-Simons action functional are the flat connections on $P$. For the following, we fix a flat connection $A_0 \in \Omega^1(M,\g)$, i.e. $dA_0+ \frac12[A_0,A_0] = 0$, and work in terms of a fluctuation $\hat{\calA}$ defined by $\calA = A_0 + \hat{\calA}$. Expressing the action in terms of this  decomposition, we obtain, using that the curvature of $A_0$ vanishes
\begin{equation}
S_{CS}[A_0 + \hat{\calA}] = S_{CS}[A_0] + \int_M\frac12\langle \hat{\calA},d_{A_0}\hat{\calA}\rangle + \frac16 \langle \hat{\calA},[\hat{\calA},\hat{\calA}]\rangle =: S_{CS}[A_0]+S^{A_0}_{CS}[\hat{\calA}]. 
\end{equation} 
The first term is simply the Chern-Simons invariant of the flat connection $A_0$, the second term is the Chern-Simons action of $\hat{\calA}$, but with differential twisted by the flat connection $A_0$. We denote this action by $S_{CS}^{A_0}$. 
\subsection{The gauge fixing Lagrangian}
In the spirit of the perturbation theory explained in the last chapter, we now gauge fix the quadratic part of $S^{A_0}$ and treat the cubic part as a perturbation. The gauge-fixing in the BV formalism consists of two steps: First, finding a decomposition of 
the BV space of fields $\calF = Y \times Y'$ such that there exists a gauge-fixing Lagrangian $\mathcal{L} \subset Y'$, and secondly, choosing such a gauge-fixing Lagrangian. The quadratic part of $S^{A_0}$ is 
\begin{equation}S^{A_0}_{free}:= \int_M\langle\hat{\calA},d_{A_0}\hat{\calA}\rangle.\end{equation}One possibility to gauge fix it is given by Hodge theory. Below we recall briefly how this works. 
\subsubsection{Hodge decomposition for $d_{A_0}$}
Let $g$ be a Riemannian metric on $M$. The metric induces a Hodge star $$*\colon\Omega^k(M,\g) \to \Omega^{3-k}(M,\g)$$ and a Hodge pairing on forms 
\begin{equation}
(\omega,\tau) = \int_M \langle\omega,*\tau\rangle
\end{equation}
which turns $\Omega^\bullet(M,\g)$ into a pre-Hilbert space. With respect to the Hodge pairing, the operator $d_{A_0}$ has a formal adjoint $d^*_{A_0}$ uniquely defined by the property 
\begin{equation}
(d_{A_0}\omega,\tau)=(\omega,d_{A_0}^*\tau).
\end{equation}
Central for Hodge theory is the (twisted) Hodge-deRham Laplacian 
\begin{equation}
\Delta_{A_0} \colon= d^*_{A_0}d_{A_0} + d_{A_0}d^*_{A_0} \colon \Omega^\bullet(M,\g) \to \Omega^\bullet(M,\g) 
\end{equation}
Forms in the kernel of $\Delta_{A_0}$ are called \emph{harmonic forms} and denoted $\mathrm{Harm}^\bullet_{A_0}(M)$. 
\begin{lem}
$\omega$ is harmonic if and only if it is both closed and co-closed: $d_{A_0}\omega = d^*_{A_0}\omega =0$.
\end{lem}
\begin{proof}
A closed and co-closed form is obviously harmonic. For the opposite, consider 
$$0 = (\omega,\Delta_{A_0}\omega) = (d_{A_0}^*\omega,d^*_{A_0}\omega) + (d\omega,d\omega).$$
The terms on the right vanish if and only if $\omega$ is both closed and co-closed. 
\end{proof}
\begin{thm}[Hodge theorem]
The map 
\begin{align}
\mathrm{Harm}^\bullet_{A_0} &\to H_{A_0}^\bullet(M,\g) \\
\omega & \mapsto [\omega] 
\end{align}
is an isomorphism.
\end{thm}
For a proof see e.g. \cite{Nicolaescu2007}. 
The key to gauge fixing is the Hodge decomposition, which we now state. 
\begin{thm}
We have 
\begin{equation}
\Omega^\bullet(M,\g) = \mathrm{Harm}^\bullet_{A_0}(M,\g)\oplus d_{A_0}\Omega^\bullet(M,\g) \oplus d_{A_0}^*\Omega^\bullet(M,\g)\label{eq:HodgeDecomposition}
\end{equation}
where all sums are orthogonal with respect to $(\cdot,\cdot)$. 
\end{thm}
A proof can be found in the same reference as above. 
\subsubsection{Lorenz gauge as gauge-fixing Lagrangian}
Using the Hodge decomposition we can write 
\begin{equation} \calF = \calY \times \calY'\label{eq:BVsplit}
\end{equation}
with 
\begin{align*}
\calY &= \mathrm{Harm}^\bullet_{A_0}(M) \\
\calY' &= d_{A_0}\Omega^\bullet(M,\g) \oplus d_{A_0}^*\Omega^\bullet(M,\g)
\end{align*}
Since the direct sum in \eqref{eq:HodgeDecomposition} is orthogonal, the decomposition \eqref{eq:BVsplit} is symplectic. 
\begin{defn}
The \emph{Lorenz gauge} Lagrangian is \begin{equation}\mathcal{L}_g:=\mathrm{im}d_{A_0}^* = \ker d_{A_0}^* \subset \calY'. \label{eq:DefLorenzGauge}
\end{equation}
\end{defn}
Notice that in Equation \eqref{eq:DefLorenzGauge} we have $\mathrm{im} d_{A_0}^* = \ker d_{A_0}^*$ since we are only considering forms in the orthogonal complement of harmonic forms. 
This is indeed a Lagrangian.
\begin{lem}
$\calL_g   \subset d_{A_0}\Omega^\bullet(M,\g) \oplus d_{A_0}^*\Omega$ is Lagrangian, i.e. both isotropic and coisotropic. 
\end{lem}
\begin{proof} Isotropy $(\calL \subset \calL^\perp)$ follows immediately from $(d_{A_0}^*)^2 = 0$ and integration by parts: 
$$\omega_{BV}(d_{A_0}^*\omega,d_{A_0}^*\tau) = \int_M d_{A_0}^*\omega \wedge d_{A_0}^*\tau = \int_M \omega \wedge (d_{A_0}^*)^2\tau = 0$$ 
shows isotropy. To see that it is also coisotropic ($\calL^\perp \subset \calL$) take a form $0\neq\alpha \in \calL$. Then 
$$0 \neq (d_{A_0}\alpha,_{A_0}d\alpha) = \pm \omega_{BV}(d_{A_0}\alpha,d_{A_0}^*(*\alpha)).$$
Hence $d_{A_0}\alpha \notin \calL^\perp$. Since all non-zero forms in $d_{A_0}\Omega^\bullet(M,\g)$ are of the form $d_{A_0}\alpha$, this shows that $d_{A_0}\Omega^\bullet(M,\g) \cap \calL^\perp = \{0\}. $
\end{proof}
\subsection{Inverting the quadratic term}
Notice that the quadratic term $\int_M\langle \hat{\calA},d_{A_0}\hat{\calA}\rangle$ has a unique critical point on $\calL_g$, namely $\hat{\calA} = 0$. The operator $d_{A_0}\colon\Omega^k_{coex}(M,\g) \to \Omega^{k+1}_{ex}(M,\g)$, and can be inverted as follows. Let $\omega = d_{A_0}\tau \in \Omega_{coex}^k(M,\g)$. Then $d_{A_0}\omega = d_{A_0}d_{A_0}^*\tau$ and, since $(d_{A_0}^*)^2 = 0$, we have 
$$d_{A_0}^*d_{A_0}\omega = d_{A_0}^*d_{A_0}d_{A_0}^*\tau = (d_{A_0}^*d_{A_0} + d_{A_0}d_{A_0}^*)d_{A_0}^*\tau = \Delta_{A_0}\omega.$$
But on the orthogonal complement of harmonic forms, the Laplacian is invertible, so we can define 
\begin{equation}
K_{A_0} := d_{A_0}^* \circ \Delta_{A_0}^{-1} \colon \Omega_{ex}^{k+1}(M,\g) \to \Omega^k_{coex}(M,\g). 
\end{equation}
It then follows that $d_{A_0}K_{A_0}\omega = K_{A_0}d_{A_0}\omega = \omega$ for all $\omega \in \calL_g$. Hence $K_{A_0}$ is an inverse to $d_{A_0}$ on $\calL_g$. See also Figure \ref{fig:HodgeDecomp} below.
\begin{figure}
\begin{equation}\notag
\begin{tikzcd}
\mathrm{Harm}^{k+1}\oplus & \Omega_{ex}^{k+1} \arrow[rd, "d^*",near start,  two heads, hook, bend left] \oplus& \Omega_{coex}^{k+1}                                         \\
\mathrm{Harm}^k    \oplus & \Omega_{ex}^{k}                                                 \oplus& \Omega_{coex}^k \arrow[lu, "d",near start, two heads, hook, bend left]
\end{tikzcd}
\end{equation}
\caption{Hodge decomposition at degree $k$ (subscripts and arguments suppressed). $d$,$d^*$ are both isomorphisms restricted to the spaces but not inverse to each other, rather, their composition is the Laplacian $\Delta$.}\label{fig:HodgeDecomp}
\end{figure}
\begin{rem}
We can extend $K_{A_0}$ to an operator $\Omega^\bullet(M,\g) \to \Omega^{\bullet-1}(M,\g)$ by defining 
\begin{equation}
K_{A_0} = d_{A_0}^* \circ (\Delta_{A_0} + P_{A_0})^{-1},
\end{equation} 
where $P_{A_0}$ is the orthogonal projection to harmonic forms. 
In terms of the Hodge decomposition this just means extending the operator $K_{A_0}$ by 0 on coexact and  harmonic forms. This operator satisfies 
\begin{equation}
d_{A_0}K_{A_0} + K_{A_0}d_{A_0} = id - P_{A_0}
\end{equation}
This means that $K_{A_0}$ is a parametrix for $d_{A_0}$ (an inverse up to smoothing operators). 
In the language of homological algebra, the map $K_{A_0}$ defines a chain homotopy between the identity map and the projection to harmonic forms. The triple $(\iota_{A_0},K_{A_0},P_{A_0})$, where $\iota_{A_0}$ is inclusion of harmonic forms, is called a \emph{contracting triple} for the complex $(\Omega^\bullet(M,\g),d_{A_0})$. For more on contracting triples see e.g. \cite{Cattaneo2008}. 
\end{rem}
\subsection{The propagator}\label{sec:prop}
The propagator, loosely speaking, is the integral kernel of the map $K_{A_0}$. There are different conventions for integral kernels. To reflect the topological nature of Chern-Simons theory, we use here the ``topological kernel'' of the map $K_{A_0}$ (this is the terminology of de Rham \cite{Rham1984}). This is a de Rham 2-current $\hat{\eta}$\footnote{Roughly speaking, de Rham currents are to forms what distributions are to functions. For more background the reader is referred to the excellent original text \cite{Rham1984}.} on $M \times M$ with values in $\g \times \g$ such that\footnote{if we were instead to use the ``metric kernel'', it would be $\int_M \hat{\eta}\wedge * y$ instead}
\begin{equation}
(K_{A_0}\omega)_x = \int_{y\in M}\langle\hat{\eta}_{(x,y)}\wedge \omega_y\rangle_{23}
\end{equation}
where we write pairing between currents and forms as integrals, as customary also for distributions, and we define $$\langle\alpha \otimes \xi_1 \otimes \xi_2 \xi_3 \rangle_{23} = \langle\xi_2,\xi_3\rangle \alpha \otimes \xi_1.$$  From the fact that $\Delta_{A_0}$ is an elliptic operator, it follows\footnote{See loc. cit. } that $\hat{\eta}$ can be represented by a smooth 2-form $\eta \in \Omega^2(M \times M - \mathrm{diag},\g\otimes\g)$.
One of the main results of Axelrod and Singer \cite{Axelrod1991a},\cite{Axelrod1994} is that the singularities of $\eta$ are tame enough to ensure that $\eta$ extends to the differential-geometric blow-up of the diagonal $Bl_{\mathrm{diag}}(M\times M)$, which, roughly speaking, is defined by replacing the diagonal with its unit sphere bundle $ST\mathrm{diag}$. For reasons that will become clear below, we denote 
\begin{equation}
C_2(M) := Bl_{\mathrm{diag}}(M\times M)
\end{equation}
\begin{prop} $C_2(M)$ has the following properties.  
\begin{itemize}
\item $C_2(M)$ is a smooth manifold with boundary.
\item The boundary of $C_2(M)$ is diffeomorphic to the unit sphere bundle over $M$: $\de C_2(M) \cong STM \to M$. 
\item $C_2(M)^\circ = C_2(M) - \de C_2(M) = M \times M - \mathrm{diag}$. 
\item $C_2(M)$ is compact (if $M$ is, which we always assume).
\end{itemize}
\end{prop}
A proof can be found in \cite{Sinha2003} or the papers of Axelrod and Singer cited.
\begin{thm}[Axelrod-Singer \cite{Axelrod1991a}]
There is a smooth 2-form $\eta^{ext} \in \Omega^2(C_2(M),\g\otimes\g)$ such that $\iota_{M\times M -\mathrm{diag}}^*\eta^{ext} = \eta$. 
\end{thm}
The smooth 2-form $\eta^{ext}$ will be the propagator that we work with. From now on we will drop the superscript $ext$. It has the following important properties. 
\begin{prop}\label{prop:propprops}
The propagator $\eta \in \Omega_2(C_2(M),\g\times\g)$ satisfies: 
\begin{enumerate}
\item Denote $\chi_i$ an orthonormal basis of harmonic forms and $\chi^i = *\chi_i$. Then 
\begin{equation}
d_{A_0}\eta = \sum_i \pm \pi_1^*\chi_i \pi_2^*\chi^i. 
\end{equation}
Here $\pi_1,\pi_2\colon C_2(M) \to M$ are the extensions to $C_2(M)$ of the restrictions of $\pi_1,\pi_2\colon M \times M \to M$ to $M \times M - \mathrm{diag}$. 
\item Denote $\iota_\de\colon \de C_2(M) \hookrightarrow C_2(M)$ the inclusion. Then $\iota_\de^*\eta$ is a global angular form on the sphere bundle $\de C_2(M) \cong STM$. 
\item Denote $T$ the extension of the map $(x_1,x_2)\mapsto (x_2,x_1)$ to $C_2(M)$. Then 
\begin{equation}T^*\eta = - \eta
\end{equation}
\item For all $i$ we have 
\begin{equation}
\int_{y\in M} \langle\eta(x,y) \wedge \chi_i\rangle_{23} = 0.
\end{equation}
\item We have 
\begin{equation}
\int_{y\in M} \langle\eta(x,y) \wedge \eta(y,z)\rangle_{23} = 0.
\end{equation}
\end{enumerate}
\end{prop}
\begin{proof}
This is the content of Remark 11, Section 4, in \cite{Cattaneo2008}. 
\end{proof}
\section{The effective action}
In the presence of zero modes we can not integrate over all fields at once, however, we can define an effective action on zero modes as a formal BV pushforward. To this end we split the fluctuation $\hat{\calA} = \alpha + \mathsf{a}$, where $\sfa \in \mathrm{Harm}^\bullet_{A_0}(M,\g)$. The quadratic part of the action then becomes
\begin{equation}
\langle\alpha + \sfa, d_{A_0}(\alpha + \sfa) \rangle = \langle \alpha,d_{A_0}\alpha = (\alpha,*d_{A_0}\alpha).
\end{equation}
For $\alpha \in \calL_g$ the operator $*d_{A_0}$ becomes invertible. 
The cubic term becomes $\langle \alpha + \sfa, [\alpha + \sfa,\alpha + \sfa]\rangle$. Expanding this term we obtain 3-valent vertices, but any number of half-edges issuing from these vertices can \emph{end} in $\mathsf{a}$ fields (see below for the precise Feynman rules). We now use the superspace formulation of the perturbative integral: We will get a superdeterminant from the quadratic part and a signature factor from the even part of the quadratic part. All in all, we obtain
\begin{equation}
Z^{A_0}_{CS}(\sfa;g) = \underbrace{e^{\frac{\ii}{\hbar}S_{CS}[A_0]}}_{=:Z_{back}}\underbrace{\operatorname{sdet}\left(\restr{*d_{A_0}}{\calL}\right)^{-1/2}
e^{\frac{\ii\pi}{4}\mathrm{sign}\left(\restr{*d_{A_0}}{\calL_{even}}\right)}}_{=:Z_{free}}\underbrace{\sum_{\Gamma}\frac{(-\ii\hbar)^{\chi(\Gamma)}}{|\mathrm{Aut(\Gamma)|}}F(\Gamma)}_{=:Z_{pert}}. \label{eq:effCS1}
\end{equation}
Here $Z_{back},Z_{free}$ are traditionally called the 0- and 1-loop part, respectively, but in this context this terminology is misleading because in the presence of zero modes the perturbation series $Z_{pert}$ contains 0- and 1-loop graphs (which are absent if there are no zero modes).
\begin{defn}
The \emph{effective action on zero modes with background $A_0$} is defined by 
\begin{equation}
Z^{A_0}_{CS}(\sfa;g) = Z_{back}Z_{free}e^{\frac{\ii}{\hbar}S_{eff}(a)}
\end{equation}
\end{defn}
A standard combinatorial exercise shows
\begin{prop}
The effective action is given by 
\begin{equation}
S_{eff}(\sfa) = \sum_{\Gamma \text{ connected }}\frac{(-\ii\hbar)^{(\mathrm{loops}(\Gamma)}}{|\mathrm{Aut}(\Gamma)|}F(\Gamma)
\end{equation}
\end{prop} 
(of course, this uses the fact that $F(\Gamma_1 \sqcup \Gamma_2)= F(\Gamma_1)F(\Gamma_2)$, which will become clear below). 
\subsection{Feynman graphs and rules}
In Equation \eqref{eq:effCS1} the sum is over all trivalent graphs with leaves, i.e. there are three half-edges emanating from every vertex that can either be connected to another half-edge or end in a leaf. We do not allow tadpoles\footnote{The unimodularity condition (Definition \ref{def:unimodular}) implies that any graph containing a tadpole evaluates to zero under the Feynman rules, so we might as well disregard them from the start} (edges connecting a vertex to itself, also known as short loops). Formally, these graphs are given by a set by a quadruple $\Gamma=(H(\Gamma),V(\Gamma),L(\Gamma),E(\Gamma))$ where 
\begin{itemize}
\item $H(\Gamma)$ is a finite set (the half-edges), 
\item $V(\Gamma)$ is a partition of $H(\Gamma)$ into sets of cardinality three, 
\item $L(\Gamma)$ is any subset of $H(\Gamma)$ (the leaves) 
\item $E(\Gamma)$ is a perfect matching of $H(\Gamma) - L(\Gamma)$.
\end{itemize} 
The set of all such graphs is denoted $\mathrm{Gr}$. 
We think of the leaves are decorated by an $\sfa$ field. See e.g. 
Let us define first the graph configuration space:
\begin{defn}
For every graph $\Gamma$, the \emph{open graph configuration space} is 
\begin{equation}
C_{\Gamma}^\circ(M) := \{c \colon V(\Gamma) \to M | c(v) \neq c(w)  \text{ if  } v \neq w \}
\end{equation}
\end{defn} 
If $c \in C_{\Gamma}^\circ$, we can extend $c$ to $H(\Gamma)$ (the set of half-edges) by defining $c(h) = c(v)$ if $h \in v$. This extension is well-defined since vertices partition the set of half-edges. \\
To formulate the Feynman rules, let $T_a$ be an orthonormal basis of the Lie algebra and define $f_{abc} = \langle T_a,[T_b,T_c]\rangle$. Expand $\eta = \eta^{ab}T_a \otimes T_b$ and $\sfa = \sfa^iT_i$. 
\begin{defn}[Feynman rules]\label{def:FRCS}
A \emph{labeling} of a graph $\gamma$ is a map $H(\Gamma) \to \{1,\ldots,\dim \g\}$. 
The \emph{graph differential form} is a differential form $\omega_\Gamma \in \Omega^\bullet(C_\Gamma^\circ)$ defined by 
\begin{equation}
(\omega_\Gamma)_f:=\sum_l \prod_{v={h_1,h_2,h_3} \in V(\Gamma)}f_{l(h_1)l(h_2)l(h_3)}\prod_{e={h_1,h_2}}\eta^{l(h_1)l(h_2)}(c(h_1),c(h_2))\prod_{h \in L(\Gamma)}\sfa^{l(h)}(c(h)).
\end{equation}
We then define the \emph{Feynman rules map} 
\begin{equation}
F \colon \mathrm{Gr} \to \R
\end{equation}
by 
\begin{equation}
F(\Gamma) = \int_{C_\Gamma^\circ(M)}\omega_\Gamma
\end{equation}
For convenience, we also the define the normalized versions 
\begin{align*}
\tilde{\omega}_\Gamma &= \frac{(-\ii\hbar)^{\chi(\Gamma)}}{|\mathrm{Aut(\Gamma)|}}F(\Gamma) \\
\tilde{F}(\Gamma) &= \frac{(-\ii\hbar)^{\chi(\Gamma)}}{|\mathrm{Aut(\Gamma)|}}F(\Gamma)
\end{align*}
\end{defn}
\begin{expl}
Let $\Gamma_1$  be the theta graph of Figure \ref{fig:theta}. Then we have 
$$\omega_{\Gamma_1}(x,y) = \sum_{i,j,k,l,m,n}f_{ijk}f_{lmn}\eta^{ik}(x,y)\eta^{jl}(x,y)\eta^{kn}(x,y)$$
and 
$$\tilde{\omega}_{\Gamma_1}(x,y) = \frac{(-\ii\hbar)^2}{12}\omega_{\Gamma_1}(x,y).$$
On the other hand, for the graph of Figure \ref{fig:theta_back} with residual fields, we have 
$$\omega_{\Gamma_2} = \sum_{i,j,k,l,m,n}f_{ijk}f_{l,m,n}\mathsf{a}^i(x)\eta^{jm}(x,y)\eta^{kn}(x,y)\mathsf{a}^l(y)$$ 
and 
$$\tilde{\omega}_{\Gamma_2} = \frac{-\ii\hbar}{4}\omega_{\Gamma_2}.$$
\begin{figure}
\centering
\begin{subfigure}{0.4\textwidth}
\centering
\begin{tikzpicture}[scale=1]
\coordinate (O) at (0,0);
 \coordinate (bulk1) at (-1,0) {};
 \coordinate (bulk2) at (1,0.0) {};
 \draw[thick] (bulk1) to[out=60,in=120] (bulk2);
 \draw[thick] (bulk1) to (bulk2);
 \draw[thick] (bulk2) to[out=-120,in=-60] (bulk1);
  \tkzDrawPoints[color=black,fill=black,size=12](bulk1,bulk2);
\end{tikzpicture}
\caption{$\Gamma_{1}$}\label{fig:theta}
\end{subfigure}
\begin{subfigure}{0.4\textwidth}
\centering
\begin{tikzpicture}[scale=1]
\coordinate (O) at (0,0);
 \coordinate (bulk1) at (-1,0) {};
 \coordinate (bulk2) at (1,0.0) {};
 \node[shape=circle,draw,above left] (res1) at (-0.8,0.4) {$\mathsf{a}$}
 edge (bulk1);
 \node[shape=circle,draw,above right] (res2) at (0.8,0.4) {$\mathsf{a}$}
 edge (bulk2);
 \draw[thick] (bulk1) to[out=60,in=120] (bulk2);
 \draw[thick] (bulk2) to[out=-120,in=-60] (bulk1);
  \tkzDrawPoints[color=black,fill=black,size=12](bulk1,bulk2);
\end{tikzpicture}
\caption{$\Gamma^{}_{2}$}\label{fig:theta_back}
\end{subfigure}
\caption{Two graphs in the effective action}
\end{figure}
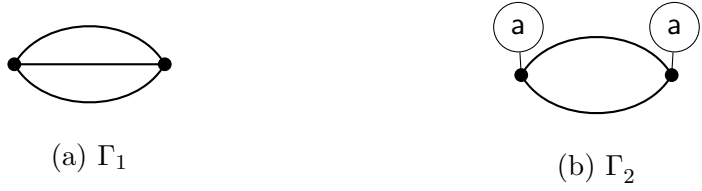 
\end{expl}
In principle, $F(\Gamma)$ does not need to be finite, as the open configuration space $C_{\Gamma}^\circ$ is not compact. However, this is not the case. Again, this was first proven in generality by Axelrod and Singer. 

\begin{thm}[\cite{Axelrod1991a},\cite{Axelrod1994}]
$F(\Gamma)$ is finite for all graphs $\Gamma$. 
\end{thm}
The proof goes through the construction of a compactification $C_{\Gamma}(M)$ of $C_{\Gamma}^\circ(M)$ such that the form $\omega_{\Gamma}$ extends to $C_{\Gamma}(M)$. This is a generalization of the compactification $C_2(M)$ discussed in  Section \ref{sec:prop} above. Instead of giving an explicit construction, we simply state the properties of this compactification that are important for us. Let us fix some notation first.
\begin{defn}[Configuration space of a finite set] For a finite set $S$ denote $C_S(M)^\circ$ the open configuration space given by maps $S \hookrightarrow M$. For convenience denote $C_{[n]}(M)^\circ =: C_n(M)^\circ$. 
\end{defn}
\begin{defn}[Reduced configuration space]

If $V$ is a vector space, the group $V  \rtimes \R_{>0}$ acts on $C_S(V)$ by scaling and translations.  $\tilde{C}_S(V)^\circ$ denotes the quotient of $C_S(V)$ under this group action, and $\tilde{C}_S(V)$ the corresponding compactification. 
\end{defn}
\begin{prop}\label{prop:Compact_Conf_Space}
There is a compact smooth $3n$-dimensional manifold with corners $C_{[n]}(M):=C_n(M)$ that satisfies: 
\begin{enumerate}
\item There is a stratification 
$$C_n(M) = \bigcup C_n(M)^{(k)}$$
with $C_n(M)^{(k)}$ compact,  $C_n(M)^{(k-1)} \subset C_n(M)^{(k)}$, such that $C_n(M)^{(k)}\setminus C_n(M)^{(k-1)}$ is a smooth manifold of dimension $3n-k$ and 
$$C_n(M)^{(0)}\setminus C_n(M)^{(1)} = C_n^\circ(M).$$
\item The  maps $\pi_{ij}\colon C^\circ_n(M) \to C_2(M)^\circ$
given by $(x_1,\ldots,x_n)$ admit smooth extensions $\pi_{ij}\colon C_n(M) \to C_2(M)$ to the  respective compactifications. 
\item The boundary of $C_n(M)$ is $\de C_n(M) = C_n(M)^{(1)}$ is given by 
$$\de C_n(M) = \bigsqcup_{S \subset [n], |S|\geq 2} \de_SC_n(M),$$ 
where $\de_SC_n(M)$ is a fiber bundle over $C_{([n]\setminus S)\sqcup\{*\}}(M)$. The fiber over a configuration $\iota$ is the reduced configuration space $\tilde{C}_{S}(T_\iota(*)M)$ of $|S|$ points in $T_{\iota(*)}M$.
\end{enumerate}
\end{prop}
The second point in this proposition implies that there is a form $\omega^{ext}_\Gamma$ on $C_\Gamma$ (given by pulling back extended propagators) such that $\restr{\omega^{ext}_\Gamma}{C_\Gamma^\circ} = \omega_\Gamma$. Hence
$$\int_{C^\circ_{\Gamma}(M)}\omega_\Gamma = \int_{C^\circ_{\Gamma}(M)}\omega^{ext}_\Gamma$$ 
but the integral on the right hand side is of a smooth differential form over a compact manifold, so it is finite. 
\subsection{The superdeterminant and the Ray-Singer torsion}
It is a surprising fact that the (regularized) superdeterminant of the quadratic part of the theory can be explicitly computed and is given by a known invariant, the so-called Ray-Singer torsion \cite{Ray1971}. This was first shown by Schwarz in 1978 \cite{Schwarz1978} (this was maybe the first explicit connection between partition functions of topological theories and topological invariants). The aim is to compute the regularized superdeterminant 
$$\operatorname{sdet}\left(\restr{*d_{A_0}}{\calL}\right) = \operatorname{det}\left(\restr{*d_{A_0}}{\calL_{even}}\right)\operatorname{det}\left(\restr{*d_{A_0}}{\calL_{odd}}\right)^{-1}$$
where regularized means \emph{zeta-regularized}, i.e. 
$$\log\det A =\lim_{s\to 0} \sum_{\lambda \neq 0}(\log \lambda) \lambda^{-s}.$$
For elliptic operators, the sum on the right hand side always converges for $s >1$ and has an analytic extension to 0. The limit is understood in this sense. See e.g. \cite{Berline2003}. The result is the following: 
\begin{lem}
For the zeta-regularized superdeterminant $\operatorname{sdet}\left(\restr{*d_{A_0}}{\calL}\right)$ we have 
\begin{equation}
\operatorname{sdet}\left(\restr{*d_{A_0}}{\calL}\right)^{-1/2} = \tau(M,A_0)^{1/2},
\end{equation}
where $\tau(M,A_0)$ is the Ray-Singer torsion \cite{Ray1971}

\begin{equation}
\tau(M,A_0) = \prod_{i=0}^3 \det (\Delta_{A_0}^{(i)})^{-\frac{i(-i)}{2}}.
\end{equation}
\end{lem}
\begin{proof}
Note that on $\calL$ we have $(*d_{A_0})^2 = \Delta_{A_0}$, thus $\det \restr{*d_{A_0}}{\Omega^{odd/even}_{coex}} =\left( \det \restr{\Delta_{A_0}}{\Omega_{coex}^{odd/even}}\right)^{1/2}$ and we can rewrite the superdeterminant as 
$$\operatorname{sdet}\left(\restr{*d_{A_0}}{\calL}\right)=\operatorname{det}\left(\restr{\Delta}{\Omega^1_{coex}}\right)^{1/2}\operatorname{det}\left(\restr{\Delta}{\Omega^{even}_{coex}}\right)^{-1/2}.$$
Consider the Hodge decomposition of the complement of harmonic forms:  
\begin{equation}
\begin{tikzcd}
\Omega^3_{ex} \arrow[rd, "d^*", bend left]                 &                                             \\
\Omega^2_{ex} \arrow[rd, "d^*" description, bend left,near start]     & \Omega^2_{coex} \arrow[lu, "d",bend left] \\
\Omega^1_{ex} \arrow[ru, "*" description] \arrow[rd, "d^*", bend left,near end] & \Omega^1_{coex} \arrow[lu, "d" description,near start, bend left]  \\
                                                           & \Omega^0_{coex} \arrow[lu, "d",bend left]            
\end{tikzcd}
\label{diag:Hodge1}
\end{equation}
(again subscripts and arguments of $d$ and $\Omega$ are suppressed to maintain a minimum of readability).
All arrows in this diagram are isomorphisms, and commute with the Hodge-de Rham Laplacian. The right hand column is precisely the gauge-fixing Lagrangian $\calL_g = \Omega_{coex}^\bullet$. The fact that all isomorphisms commute with the Laplacian means that there essentially only two ``independent'' pieces in the Laplacian: The Laplacian restricted to 0-forms and the Laplacian restricted to coexact one-forms. 
More precisely, we have 
\begin{equation}
\det \restr{\Delta}{\Omega^0} = \det \restr{\Delta}{\Omega^1_{ex}} = \det \restr{\Delta}{\Omega^2_{coex}} = \det \restr{\Delta}{\Omega^3} 
\end{equation}and 
\begin{equation}
\det \restr{\Delta}{\Omega^1_{coex}} = \det \restr{\Delta}{\Omega^2_{ex}}.
\end{equation}
Writing $\Delta^{(i)} =  \restr{\Delta}{\Omega^i}$, we can now combine this with the equation 
\begin{equation}
\det\Delta^{(i)} = \det \restr{\Delta}{\Omega^i_{coex}} \det \restr{\Delta}{\Omega^i_{ex}}
\end{equation} to obtain 
\begin{align*}
\operatorname{sdet}\left(\restr{*d_{A_0}}{\calL}\right)&=\operatorname{det}\left(\restr{\Delta}{\Omega^1_{coex}}\right)^{1/2}\operatorname{det}\left(\restr{\Delta}{\Omega^{even}_{coex}}\right)^{-1/2} \\
&= (\det \Delta^{(1)})^{1/2} \operatorname{det}\left(\restr{\Delta}{\Omega^1_{coex}}\right)^{-1/2}\left(\det\Delta^{(0)}\right)^{-1/2}\left(\det(\restr{\Delta}{\Omega^2_{coex}}\right)^{-1/2} \\
&=\frac{(\det\Delta^{(1)})^{1/2}}{(\det\Delta^{(3)})^{3/2}} 
\end{align*}
Now we use the fact $\det \Delta^{(1)} = \det\Delta^{(2)}$ (which again follows from the fact that the Hodge star intertwines the two Laplacians) to rewrite this as 
$$\frac{(\det\Delta^{(1)})^{1/2}}{(\det\Delta^{(3)})^{3/2}} = \frac{1}{(\det\Delta)^{1/2}}\frac{\det\Delta^{(2)}}{(\det\Delta^{(3)})^{3/2}} = \prod_{i=0}^3 \det (\Delta^{(i)})^{\frac{i(-i)}{2}}.$$
The right hand side is exactly the inverse of the Ray-Singer torsion \cite{Ray1971}
$$\tau(M,A_0) = \prod_{i=0}^3 \det (\Delta_{A_0}^{(i)})^{-\frac{i(-i)}{2}}.$$
Hence we conclude that 
$$\operatorname{sdet}\left(\restr{*d_{A_0}}{\calL}\right)^{-1/2} = \tau(M,A_0)^{1/2}.$$
\end{proof}
By the Cheeger-M\"uller Theorem (\cite{Cheeger1977},\cite{Mueller1978},\cite{Cheeger1979}) the Ray-Singer torsion is equal to the Reidemeister \cite{Reidemeister1935} torsion (see e.g. \cite{Nicolaescu2003} for background on the Reidemeister torsion) of the corresponding representation of $\pi_1(M)$, which is a topological invariant of $M$, i.e. independent of the metric $g$. 
\subsection{The phase factor}
Similar to the determinant, we can also zeta-regularize the signature of $\restr{*d_{A_0}}{\Omega^1_{coex}}$: 
\begin{equation}
\operatorname{sign}*d_{A_0}:= \lim_{s\to0}\sum_{\lambda\neq 0}\mathrm{sign}\lambda |\lambda^{-s}|
\end{equation}
(again this limit is understood in the sense of analytic continuation). 
Also this invariant can be expressed in terms of known quantities. To this end we need the following Lemma\footnote{This lemma and its proof were explained to the author by P.Mnev.}.
\begin{lem}
The spectrum of $d_{A_0}*\colon \Omega^{odd}_{ex} \to \Omega^{odd}_{ex}$ is symmetric around 0.  
\end{lem}
\begin{proof}
Let $f \in \Omega^0(M,\g)$ be an eigenfunction of the Laplacian, $\Delta_{A_0}f = \lambda^2 f$. Then we claim that 
\begin{equation}
\omega_{\pm \lambda} = d_{A_0}*d_{A_0}f \pm \lambda d_{A_0}f\label{eq:eigenform}
\end{equation}
is an Eigenform of $\restr{*d_{A_0}}{\Omega^{odd}_{ex}}$ of eigenvalue $\pm\lambda$. Indeed, we have
\begin{align*} d_{A_0}*(\pm\lambda d_{A_0}f + d_{A_0}*d_{A_0} f) &= \pm\lambda d_{A_0}*d_{A_0}f + d_{A_0}\Delta_{A_0}f \\
&= \pm\lambda d_{A_0}*d_{A_0}f + d_{A_0}\lambda^2f = \pm\lambda(\pm\lambda d_{A_0}f + d_{A_0}*d_{A_0} f).
\end{align*}

We claim that all eigenforms of $*d_{A_0}$ are of the form $\omega_{\pm \lambda}$ as defined in Equation \eqref{eq:eigenform} above. Indeed, if $\omega \in \Omega^1_{ex}\oplus \Omega^3_{ex}$ is an eigenform of $*d_{A_0}$ of eigenvalue $\lambda$, then - since $*d_{A_0}$ squares to $\Delta$ on this space -  it is an eigenform of $\Delta_{A_0}$ of eigenvalue $\lambda^2$. But $\Delta_{A_0}$ maps $\Omega^1$ and $\Omega^3$ to itself, hence the 1- and 3-form components of  $\omega$ must be eigenforms of $\Delta_{A_0}$ of eigenvalue $\lambda^2$. But these eigenforms are precisely given by $d_{A_0}f$ (resp. $d_{A_0}*d_{A_0}f$) for $f$ and eigenfunction of $\Delta^{(0)}$ of eigenvalue $\lambda^2$. Hence $\omega$ is equal to $\omega_{\pm\lambda}$ (for one of two signs). We conclude that the spectrum is given precisely by $\{\pm\lambda|\lambda \in  \mathrm{spec}(\Delta^{(0)})\}.$
\end{proof}
As a direct corollory we have: 
\begin{cor}
\begin{equation}
\operatorname{sign}\restr{*d_{A_0}}{\Omega^1{coex}} = \operatorname{sign} \restr{*d_{A_0} + d_{A_0}*}{\Omega^{odd}} =: \psi(A_0,g)
\end{equation}
\end{cor}
Here $\psi(A_0,g)$ is the \emph{Atiyah-Patodi-Singer eta invariant} \cite{Atiyah1973},\cite{Atiyah1975} of the Dirac operator $L_ := *d_{A_0} + d_{A_0}* \colon \Omega^{odd} \to \Omega^{odd}$. In contrast to the torsion, it is \emph{not} invariant under change of the metric $g$. This is the first sign of the anomalous behaviour of the perturbative quantum Chern-Simons theory. We will return to this issue in section \ref{sec:anomaly} below. 
\section{The Quantum Master Equation, again}
In Section \ref{sec:QME1} we have argued that the Quantum Master Equation formally holds if we assume that the Lie algebra $\g$ is unimodular. We thus expect that the conclusion of Theorem \ref{thm:BVpushforward} holds, i.e. that the effective action also satisfies the Quantum Master Equation. However, since the theorem does not apply to the infinite-dimensional setting, the Quantum Master Equation has to be proven ``by hand''. 
\subsection{The Laplacian on the space of residual fields}
The space of residual fields is the space of harmonic forms, 
$$\calY = \Pi\mathrm{Harm}_{A_0}^\bullet(M,\g) \cong H_{A_0}^\bullet(M,\g).$$
It carries the symplectic form 
$$\omega_\calY(\sfa,\mathsf{b}) = \int_M \langle \sfa,\mathsf{b}\rangle.$$ 
To express the Laplacian it is helpful to express the symplectic form in Darboux coordinates. To this end, choose orthonormal bases (w.r.t. the Hodge pairing $(\cdot,\cdot)$) $\chi^i_{(k)}$ of $H_{A_0}^k(M,\g)$ for $k=0,1$ and we let $\chi^i_{(k)} = *\chi^i_{(3-k)}$ for $k=2,3$. We also choose an orthonormal basis $T_a$ of $\g$ and expand the the forms accordingly $\chi^i_{(k)}=\chi_{(k)}^{i,a}T_a.$ Denote by $z_{i,a}^{(k)}$ the coordinates with respect to this basis. Then the symplectic form is expressed as 
$$\omega_{BV} = \sum_{i,a} dz_{i,a}^{(3)}\wedge dz_{i,a}^{(0)} + dz_{i,a}^{(2)}\wedge dz_{i,a}^{(1)}$$

and the BV Laplacian can be expressed as 
\begin{equation}
\Delta_{Y}\sum_{i,a}\frac{\de}{\de z_{i,a}^{(3)}} \frac{\de}{\de z_{i,a}^{(0)}} + \sum_{i,a}\frac{\de}{\de z_{i,a}^{(2)}} \frac{\de}{\de z_{i,a}^{(1)}}\label{eq:BV_Laplacian_ZM}
\end{equation}
The following Lemma is crucial to the proof of the QME. 
\begin{lem}\label{lem:diff_prop}
Denoting the propagator by $\eta \in \Omega^2(C_2(M),\g \otimes \g)$ we have 
\begin{equation}
d_{A_0}\eta = \Delta_Y \pi_1^*\sfa \pi_2^*\sfa.
\end{equation}
\end{lem}
\begin{proof}
First expand the left hand side in the basis $\chi_{i,a}^{(k)}$ to obtain (cf. Lemma \ref{prop:propprops})
$$d_{A_0}\eta = \sum_{k=0}^3\sum_{i,a} \pi_1^*(\chi_{i,a}^{(k)} \pi_2^*(\chi_{i,a}^{(3-k)}).$$
Expanding $\sfa = \sum_{i,a,k}z^{i,a}_{(k)}\chi_{i,a}^{(k)}$, the Lemma follows immediately from Equation \eqref{eq:BV_Laplacian_ZM}. 
\end{proof}
Pictorially this Lemma is expressed as $d_{A_0}($
\begin{tikzpicture}
\node[vertex] at (0,0) {};
\node[vertex] at (1,0) {};
\draw[black] (0,0) -- (1,0);
\end{tikzpicture}$) = \Delta_Y ($
\begin{tikzpicture}[scale=0.75]
\node[vertex] (l) at (0,0) {};
\node[circle,inner sep=1pt,draw] at (0.5,0) {$\sfa$}
edge (l); 
\node[vertex] (r) at (2,0) {};
\node[circle,inner sep=1pt,draw] at (1.5,0) {$\sfa$}
edge (r); 
\end{tikzpicture}$)= $
\begin{tikzpicture}[scale=0.75]
\node[vertex] (l) at (0,0) {};
\node[circle,inner sep=1pt,draw] at (0.5,0) {$\chi$}
edge (l); 
\node[vertex] (r) at (2,0) {};
\node[circle,inner sep=1pt,draw] at (1.5,0) {$\chi$}
edge (r); 
\end{tikzpicture}, i.e. the differential ``cuts'' the edge corresponding to a propagator (in the last picture summation over indices of $\chi$ is implied).
\subsection{Proof of Quantum Master Equation}
The main theorem of this section is that the effective action satisfies the Quantum Master Equation: 
\begin{thm}We have 
\begin{equation}
\Delta Z_{pert} = \Delta_{\calY}e^{\frac{\ii}{\hbar}S_{eff}} = 0.
\end{equation}
\end{thm}
Since $Z_{back},Z_{free}$ do not depend on $\sfa$ this implies that $\Delta_\calY Z^{A_0}(\sfa)_{CS} = 0$. 
There are two main steps in the proof. The first one is to rewrite $\Delta Z$ as a sum of integrals over boundary faces of the compactified configuration space. The second is to show that all such integrals vanish. This is indeed a general cooking recipe for proofs of such equations in theories whose propagators admit extensions to such compactifactions. 
\subsubsection{Applying Stokes' theorem}
\begin{prop}\label{prop:DeltaBVaction}
Using the notation of Definition \ref{def:FRCS}, write 
$$Z_{pert} = \sum_\Gamma \int_{C_\Gamma(M)} \tilde{\omega}_{\Gamma}.$$
Then 
\begin{equation}
\Delta_\calY Z_{pert} = \sum_\Gamma \int_{C_\Gamma(M)}d_{A_0}\tilde{\omega}_\Gamma.
\end{equation}
\end{prop}
\begin{proof}
$\tilde{\omega}_\Gamma$ is a product of propagators $\eta$ and zero modes $\sfa$. From the Leibniz rule and the fact that $d_{A_0}\sfa = 0$, we have that 
$$d_{A_0}\omega = \sum_{e \in E(\Gamma)} \tilde{\omega}_{\Gamma}^e$$
where $\tilde{\omega}_{\Gamma}^e$ denotes the form where to the edge $e$ we associate the form $d_{A_0}\eta = \sum_{i,a,k}\pi_1^*\chi^{i,a}_{(k)}\pi_2^*\chi^{i,a}_{(k)}$. Thus, we can express $$
\sum_\Gamma \int_{C_\Gamma(M)}d_{A_0}\tilde{\omega}_\Gamma = \sum_{\Gamma^e_m}\int_{C_\Gamma}\frac{(-\ii\hbar)^{\chi(\Gamma)}}{|\mathrm{Aut}_m(\Gamma)|}{\omega}^e_{\Gamma}$$
as a sum over graphs with one marked edge. Notice that here we have  replaced $\mathrm{Aut}(\Gamma)$ by automorphisms $\mathrm{Aut}_n(\Gamma)$ of marked graphs to account for the fact that marking different edges might lead to automorphic marked graphs. On the other hand, we have 
$$\Delta\sum_\Gamma\int_{C_\Gamma(M)} \tilde{\omega}_{\Gamma} = \sum_\Gamma\sum_{l_1 \neq l_2 \in L(\Gamma)}\int_{C_\Gamma}\omega^{l_1,l_2}_\Gamma$$
where $\omega^{l_1,l_2}_{\Gamma}$ denotes the form that results from the usual Feynman rules with $\Delta$ applied to the residual fields at the two leafs $l_1,l_2$. Notice that if $l_1,l_2$ are placed at the same vertex, $\omega^{l_1,l_2}_{\Gamma}$ contains a term of the form $f_{jii} = 0$. Hence, we can rewrite $\Delta Z$ by 
$$ \Delta Z_{pert} = \sum_{\Gamma^{l_1,l_2}_m}\int_{C_\Gamma}\frac{(-\ii\hbar)^{\chi(\Gamma)}}{|\mathrm{Aut}_m(\Gamma)|}{\omega}^{l_1,l_2}_{\Gamma}$$ 
where the sum goes over graphs with a pair of marked leafs. Now, by Lemma \ref{lem:diff_prop} we have $\omega^e_{\Gamma} = \omega^{l_1,l_2}_{\Gamma}$ if $l_1,l_2$ are placed at the same vertices as the start and end of $e$. The claim now follows from noticing that there is an obivous bijection between automorphism classes of graphs with a marked edge and automorphism classes with a pair of marked vertices by simply connecting the two marked leaves (with inverse given by cutting the marked edge). 
\end{proof}
We thus conclude that, by Stokes' theorem,
\begin{equation}
\Delta_\calY Z_{pert} = \sum_\Gamma \int_{C_\Gamma(M)}d_{A_0}\tilde{\omega}_\Gamma = \int_{\de C_{\Gamma}}\omega_{\Gamma}. 
\end{equation}
\subsubsection{Vanishing of the boundary contributions}
As explained in Proposition \ref{prop:Compact_Conf_Space}, the boundary of the compactified configuration space has different components, one for every subset of $V(\Gamma)$  of cardinality at least 2: 
$$\de C_{\Gamma} = \sqcup_{S \subset V(\Gamma), |S| \geq 2}\de_S \Gamma.$$ 
We will make extended use of the fact that for a fiber bundle $F \hookrightarrow E \twoheadrightarrow M$, we have a generalised Fubini theorem for integration along the fiber 
$$\int_E \omega = \int_M \left(\int_F\omega\right).$$ We will apply this to $\de_SC_{\Gamma}$, which is a fiber bundle over $C_{\Gamma/S}$ with fiber over $c$ given by $\tilde{C}_{S}(T_{c([S])                    }M)$. \\
One has to distinguish two cases, the case where $|S| = 2$ (the so-called ``principal faces'') and the case $|S| \geq 3$ (the so-called ``hidden faces''). 
\begin{lem} 
We have 
\begin{equation}
\sum_{\Gamma}\sum_{S \subset V(\Gamma),|S|=2} \int_{\de_SC_\Gamma} \omega_\Gamma  =0. 
\end{equation}
\end{lem} Let us first consider the case $|S| = 2$. Let $S = \{v_1,v_2\}$. The boundary stratum is a sphere bundle over the ``small'' diagonal $v_1 = v_2$. If there is no edge connecting $v_1,v_2$, then the form $\omega_{\Gamma}$ is regular on this diagonal and hence has no form degree on the fibers of this bundle. Hence the corresponding contribution vanishes. If there is an edge $e = S$, then the corresponding boundary face is a sphere bundle over $C_{\Gamma/e}$ (here $\Gamma/e$ denotes $\Gamma$ with the edge $e$ contracted). By normalization of the propagator corresponding to $e$, integrating over the fiber of this sphere bundle yields $\pm +1$ and we have 
$$\int_{\de_eC_{\Gamma}}\omega_{\Gamma} = \int_{C_{\Gamma/e}}\omega_{\Gamma/e}.$$ 
These contributions do not vanish individually, but only after we sum over all graphs. The graphs $\Gamma/e$ contain a single 4-valent vertex, and there are three different possibilities for this 4-valent vertex to arise, see figure \ref{fig:IHX}. 
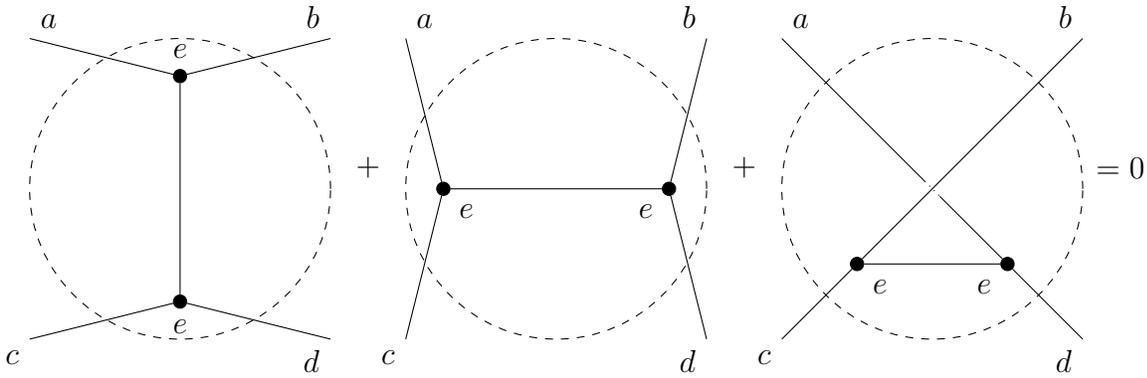
\begin{figure}[h!]
\begin{tikzpicture}
\draw[dashed] (0,0) circle (2cm);
\node[vertex, label=above:{$e$}] (uc) at (0,1.5) {};
\node[vertex,label=below:{$e$}] (ud) at (0,-1.5) {}
edge[black] (uc);
\node[coordinate,label=below left:{$c$}] at (-2,-2){} edge (ud);
\node[coordinate,label=below left:{$d$}] at (2,-2){} edge (ud);
\node[coordinate,label=above left:{$b$}] at (2,2){} edge (uc);
\node[coordinate,label=above right:{$a$}] at (-2,2){} edge (uc);
\node[coordinate,label={$+$}] at (2.5,0) {};
\begin{scope}[shift={(5,0)}]
\draw[dashed] (0,0) circle (2cm);
\node[vertex, label=below left:{$e$}] (uc) at (1.5,0) {};
\node[vertex,label=below right:{$e$}] (ud) at  (-1.5,0) {}
edge[black] (uc);
\node[coordinate,label=below left:{$c$}] at (-2,-2){} edge (ud);
\node[coordinate,label=below left:{$d$}] at (2,-2){} edge (uc);
\node[coordinate,label=above left:{$b$}] at (2,2){} edge (uc);
\node[coordinate,label=above right:{$a$}] at (-2,2){} edge (ud);
\node[coordinate,label={$+$}] at (2.5,0) {};
\end{scope}
\begin{scope}[shift={(10,0)}]
\draw[dashed] (0,0) circle (2cm);
\node[vertex, label=below left:{$e$}] (uc) at (1,-1) {};
\node[vertex,label=below right:{$e$}] (ud) at  (-1,-1) {}
edge[black] (uc);

\node [coordinate] (extra2) at (0.1,-0.1) {} edge[black] (uc);
\node [coordinate] (extra1) at (-0.1,0.1) {} edge[dotted] (extra2);
\node[coordinate,label=below left:{$c$}] at (-2,-2){} edge (ud);
\node[coordinate,label=below left:{$d$}] at (2,-2){} edge (uc);
\node[coordinate,label=above left:{$b$}] at (2,2){} edge (ud);
\node[coordinate,label=above right:{$a$}] at (-2,2){} edge (extra1);
\node[coordinate,label={$=0$}] at (2.5,0) {};
\end{scope}
\end{tikzpicture}
\caption{The famous IHX relation, which holds whenever graphs are identical outside of the dashed circle. }\label{fig:IHX}
\end{figure}
Summing over these three possibilities we obtain the 4-vertex weight $v_{abcd} = \sum_ef_{aeb}f_{ced} + f_{aec}f_{bed} + f_{aed}f_{bec}$ which vanishes by the Jacobi identity $\sum_e f_{[abe}f_{c]ed} = 0$ (the square brackets denote antisymmetrization over uncontracted indices). Thus, after summing over all graphs, these contributions cancel out. \\
In contrast to the principal faces, where we have to sum over all graphs to ensure vanishing, the contributions from hidden faces vanish seperately: 
\begin{lem}
Let $S \subset V(\Gamma)$ have cardinality $|S| \geq 3$. Then $$\int_{\de_SC_\Gamma(M)}\omega_{\Gamma} = 0.$$
\end{lem}
\begin{proof}
This follows from the ``vanishing Lemmata'' by Kontsevich ( \cite{Kontsevich1994}),  but we briefly repeat the argument. The dimension of $\tilde{C}_S(\R^3)$ is $3|S| - 4$. Consider the edge subgraph $\Gamma_S$ on the vertices $S \subset V(\Gamma)$ (i.e. it contains all edges between vertices of $\gamma$ but no leaves). Suppose all vertices in $V(\Gamma)$ are 3-valent. Then the corresponding propagators multiply to a $3|S|$ form, after integrating over the fiber we have a 4-form placed at the point of collapse of $S$ (the point labeled $[S]$ in $C_{\Gamma/S}$, which is zero for dimensional reasons. Thus the contribution vanishes unless $\Gamma_S$ has a vertex of valence 2 or less. If a vertex has valence 0 or 1, integrating over that vertex yields 0 (again or degree reasons). So the only remaining case is when $S$ has a vertex of degree 2. In that case, we consider the integral over that vertex only, it looks like 
$$\int_{y\in\R^3}\eta(x,y)\eta(y,z).$$
On $\tilde{C}_S(\R^3), \eta$ is translation invariant, and the involution $z \mapsto x+z -y $ sends $\eta(x,y)\eta(y,z)$ to itself. Since this involution reverses the orientation we conclude that the integral is zero. 
\end{proof} 
\section{Gauge invariance and the framing anomaly}\label{sec:anomaly}
In the language of BV formalism, what we expect to be gauge invariant is the \emph{BV cohomology class} of the partition function $Z^{A_0}_{CS}(\sfa;g)$. More precisely, from Theorem \ref{thm:BVpushforward} we expect that if the vary the gauge-fixing Lagrangian smoothly, the partition function changes by a $\Delta_Y$-exact term.   In this section we will see that this is not the case, but that there is a partial remedy to this. This is a well-known fact whose appearance in the literature will be discussed in more detail in Subsection \ref{sec:framing_lit} below. 
\subsection{Dependence of $Z_{free}$ on the gauge-fixing metric}
Since we only consider gauge-fixings coming from Riemannian metrics, we will consider gauge fixing Lagrangians $\calL_{g_t}$ induced by a smooth family $g_t$ of Riemannian metrics. Surprisingly, not even the ``free'' part of the partition function $$Z_{free} = \tau(M,A_0)^{1/2}e^{\frac{\ii\pi}{4}\psi(A_0,g)}$$
is gauge invariant. To compute the dependence on the gauge it is necessary to introduce a \emph{framing} and the \emph{gravitational Chern-Simons invariant}: 
\begin{defn}
A \emph{framing} of a three-manifold $M$ is a vector bundle isomorphism $f\colon M \times \R^3 \to TM$, i.e. a trivialization of the tangent bundle. We demand that, if $M$ is oriented then $f$ is orientation-preserving.
\end{defn}
The datum of a framing is equivalent to a section of the frame bundle $s_f\colon M\to\mathrm{Fr}(TM)$ (the section $s$ is simply given by $s_f(x) = (f(e_1),f(e_2),f(e_3))$ where $e_i$ is the canonical basis of $\R^3$). Using this section we can define the gravitational Chern-Simons invariant of $g$ and $f$: 
\begin{defn}
Let $g$ be a Riemannian metric on $M$ and $f$ a framing of $M$. Denote $\Theta_{LC} \in \Omega^1(\mathrm{Fr}(TM),\mathfrak{so}(\R^3))$ denote the connection form of the Levi-Civita connection of $g$. Then we define 
\begin{equation}
S_{CS}^{grav}(g,f) := \int_M\operatorname{tr}s_f^*cs(\Theta_{LC})
\end{equation}
\end{defn}
We then have the following proposition. 
\begin{prop}
For any framing $f$, the quantity 
\begin{equation}
\frac{\ii\pi }{4}\psi(A_0,g) + \frac{\ii\dim  G}{24}\frac{1}{2\pi}S_{CS}^{grav}(g,f)\label{eq:invariant_phase}
\end{equation}
is independent of the metric $g$.
\end{prop}
The proof uses the Atiyah-Patodi-Singer index theorem as in \cite{Witten1989} and goes well beyond the scope of these notes\footnote{A thorough account of the APS theorem, which also explains \eqref{eq:invariant_phase}, is Melrose \cite{Melrose1993} - almost 400 pages!}.
We conclude that if we rescale the free partition function by a phase factor, 
\begin{equation}
Z_{free}' = \exp\left(\frac{\ii\dim G}{24}\frac{1}{2\pi}S_{CS}^{grav}(g,f)\right)Z_{free}
\end{equation}
then $Z_{free}'$ is gauge invariant but depends on the framing. 
\subsubsection{Dependence on the framing}
The dependence on the framing is controllable. First of all, $S^{grav}_{CS}(g,f)$ depends only on the homotopy class of $f$. Homotopy classes of framings are slightly subtle: If we fix one framing $f$, then any other framing is related to it by a map $\gamma\colon M \to SO(3)$, these maps are distinguished by their degree $\deg\gamma$, which is an integer, and an element  $c(\gamma)\in H^1(M, \Z_2)$.\footnote{The degree can be computed as $\deg \gamma = \int_M \gamma^*\omega$, where $\omega$ is a normalized volume form on $SO(3)$, $c(\gamma)$ is the pullback $\gamma^*\alpha$, where $\alpha \in \Omega^1(M)$ is a generator of $H^1(SO(3),Z_2)$ (i.e. $\alpha$ is closed and for a generator $\tau$ of $\pi_1(SO(3))=\Z_2$, $\int_{\tau}\alpha = 1$.)}If we modify the framing by a map $\gamma\colon M \to SO(3)$ of degree $n$ with $c(\gamma) = 0$, then the gravitational Chern-Simons invariant changes by 
$$ \frac{1}{2\pi}S_{CS}^{grav}(M, f\cdot\gamma) =\frac{1}{2\pi}S_{CS}^{grav}(M, f) + 2\pi n$$ and  thus a change of framing changes the rescaled partition function by a phase factor of 
\begin{equation}
Z'_{free}(f \cdot \gamma) = Z_{free}'(f)\exp\left(\frac{2\pi \ii n \dim G}{24}\right).
\end{equation}
\subsection{Gauge dependence of $Z_{pert}$}
The appearance of the gravitational Chern-Simons action is a foreshadowing of what happens at higher order corrections. We present here only a brief version, and again refer to the literature for technical details. To analyze the time derivative of the higher order correction we first need to know what happens to the propagator and the residual fields. 
\begin{lem}
Let $g_t$ be a smooth family of Riemannian metrics and denote $\eta_t$ the propagator defined by $g_t$, $\chi^{i}_{(k),t}$ a basis of harmonic $k$-forms for $g_t$. Then there are $xi^i_{(k)} \in \Omega^{k-1}(M,\g)$,$\lambda \in  \Omega^1(M\times M,\g \otimes \g)$ such that 
\begin{align}
\restr{\frac{d}{dt}}{t=0}\chi^{i}_{(k),t} &= d_{A_0}\xi^i_{(k)} \label{eq:change_harm}\\ 
\restr{\frac{d}{dt}}{t=0}\eta_{t} &= d_{A_0}\lambda + \sum_{i,k}\pi_1^*\xi_{(k)}^i\pi_2^*\chi^i_{(k)} + \pi_1^*\chi_{(k)}^i\pi_2^*\xi^i_{(k)} +  \label{eq:change_prop} \\ 
\end{align}
\end{lem}
\begin{proof}
The first equation follows immediately from the fact that the cohomology class of $\chi^i_{(k),t}$ is independent of $t$. For the second statement, see \cite{Cattaneo2008}. 
\end{proof}
\subsubsection{Stokes' theorem, again}
This allows us to apply Stokes' theorem in a similar fashion to the above. Denote by $Z^t_{pert}$ the perturbative part of the partition function defined using the metric $g_t$, then we have 
\begin{align*}\restr{\frac{d}{dt}}{t=0}Z^t_{pert} &= \restr{\frac{d}{dt}}{t=0}\sum_{\Gamma}\int_{C_{\Gamma}(M)}\tilde{\omega}_{\Gamma} = \sum_{\Gamma}\int_{C_{\Gamma}(M)}\restr{\frac{d}{dt}}{t=0}\tilde{\omega}_{\Gamma}  \\
&= \sum_{\Gamma}\sum_{e \in E(\Gamma)}\int_{C_{\Gamma}(M)}\tilde{\omega^e}_{\Gamma} + \sum_{l \in L(\Gamma)}\int_{C_{\Gamma}(M)}\tilde{\omega^l}_{\Gamma} \end{align*}
where now we denote $\omega^e$ the form obtained by placing $\dot{\eta}$ at the edge $e$ and $\omega^l$ the form obtained by placing $\dot{a} = \sum_{i,k}z_{i}^{(k)}d\xi^i_{(k)} =: d\mathsf{b}$ at the leaf $l$. In this sum, we integrate all terms of the form $d\lambda$ or $d\xi$ by parts. Up to boundary terms, this leaves us with terms where $d_{A_0}$ is applied to propagators. An argument similar to Proposition \ref{prop:DeltaBVaction} shows that these terms sum up to $\Delta X$, where $X$ is given by 
\begin{equation}
X = \sum_{\Gamma^e_m}\int_{C_{\Gamma}(M)}(\omega^e_{\Gamma^e_m})' + \sum_{\Gamma^l_m}\int_{C_\Gamma(M)}(\omega^l_{\Gamma^l_m})'.
\end{equation}
Here $\Gamma^e_m$ ($\Gamma^l_m$) runs over graphs with a marked edge (leaf), $\Gamma$ denotes the graph with the marking forgotten and $\omega^e_{\Gamma^e_m}$ ($\omega^e_{\Gamma^e_m} $) denotes the form obtained from the usual Feynman rules but putting $\lambda$ (resp. $\mathsf{b}$) at the marked edge (resp. leaf). The only slight difference are the terms of the form $\chi \xi$ in $\eta$: these arise when $\Delta_\calY$ eats an $\sfa$-$\mathsf{b}$ pair of residual fields (instead of the usual $\sfa$-$\sfa$-pair). 
\subsubsection{The boundary terms}
For the boundary terms, one performs an analysis similar to the above. Again, summing over all graphs boundary faces corresponding to $|S| = 2$ cancel out by the IHX relation. For the ``hidden faces'', the only difference is the degree count (since $\lambda$ has degree 1) in the case where an entire connected component of a graph collapses that contains no residual fields (all vertices of the edge subgraph $\Gamma_S$ on $S\subset V(\Gamma)$ are trivalent) and such that the marked edge is contained in $\Gamma_S$. In these cases the degree count now says that  we obtain a 3-form at the point of collapse. One can compute the integral of this 3-form and show that is given by a numeric coefficient $\phi(\Gamma,\g)$, which depends only on the graph $\Gamma$ and the Lie algebra $\g$ times the derivative of the gravitational Chern-Simons action, for any framing $f$: 
\begin{equation}
\int_{\de_S\Gamma}(\omega^e_{\Gamma})'' =\left( \int_{C_{\Gamma-\Gamma_S}(M)}\phi(\Gamma_S,\g)\right)\restr{\frac{d}{dt}}{t=0}S_{CS}^{grav}(g,f)
\end{equation}
\subsubsection{Rescaling of the entire partition function}
Summing over all graphs, we see that we obtain 
\begin{equation}
\restr{\frac{d}{dt}}{t=0}Z_{pert} = \Delta X + Z_{pert}\times\sum_{\Gamma \text{connected}}\phi(\Gamma,\g)\restr{\frac{d}{dt}}{t=0}S_{CS}^{grav}(g,f). 
\end{equation}
We have thus arrived at the following result. 
\begin{thm}
There exists a power series $\phi(\hbar)$ of the form 
\begin{equation}
\phi_\g(\hbar) = \phi_0 + \sum_{\Gamma}\frac{(-\ii\hbar)^\chi(\Gamma)}{|\mathrm{Aut}(\Gamma)}\phi(\Gamma,\g)
\end{equation}
such that for every framing $f$, the \emph{BV cohomology class} of the rescaled partition function 
\begin{equation}
(Z^{A_0}_{CS})'(\sfa,g) = Z^{A_0}_{CS}(\sfa,g)e^{\ii\phi(\Gamma,\g)S_{CS}^{grav}(g,f)}
\end{equation}
does not depend on the metric $g$. Moreover $\phi_0 = \frac{2\pi\dim G}{24}$. 
\end{thm}
\subsection{The framing anomaly in the literaure}\label{sec:framing_lit}
The framing anomaly in Chern-Simons theory was observed in various places, starting with Witten \cite{Witten1989}. In the perturbative setting, this anomaly was discussed \cite{Axelrod1991a},\cite{Axelrod1994} and later in \cite{Bott1998},\cite{Bott1999}, \cite{Cattaneo1999},\cite{Cattaneo2008}. In the Kontsevich-Kuperberg-Thurston-Lescop approach \cite{Kontsevich1994},\cite{Kuperberg1999},\cite{Lescop2004},\cite{Lescop2004a} a framing is used to define the propagator, and the resulting perturbative series depends on it.\\
In contrast, the Reshetikhin-Turaev invariants \cite{Reshetikhin1991} do not exhibit any dependence on a framing. This initially was the source of some confusion, since there is no canonical choice of framing for a 3-manifold that could have been implicit in the construction of the Reshetikhin-Turaev invariant. A possible explanation was offered by Atiyah, who showed in \cite{Atiyah1990} that 3-manifolds admit a canonical \emph{2-framing}. He argued that Chern-Simons invariants should more naturally be considered as invariants of 2-framed manifolds (and that the Reshetikhin-Turaev invariants correspond to the Chern-Simons invariants in the canonical 2-framing). However, for perturbative Chern-Simons invariants the role of 2-framings has yet to be understood precisely. 
\printbibliography
\end{document}